\documentclass[prx,aps,floatfix,reprint,twocolumn,superscriptaddress,longbibliography,10pt]{revtex4-2}

\usepackage{graphicx}
\usepackage{amsfonts}
\usepackage{amssymb}
\usepackage{amsmath}
\usepackage{amsthm}
\usepackage{enumitem}
\usepackage{multirow,booktabs}
\usepackage{color}
\usepackage[dvipsnames]{xcolor}
\usepackage[colorlinks=true,linkcolor=dblue,citecolor=dblue,urlcolor=dblue,plainpages=false,pdfpagelabels]{hyperref}

\usepackage{mathtools}
\usepackage{bbm}

\usepackage[capitalize]{cleveref}

\usepackage{physics}

\usepackage{float}
\usepackage{array}
\usepackage{longtable}

\usepackage{verbatim}

\usepackage[utf8]{inputenc} 
\usepackage[T1]{fontenc}
\usepackage{etoolbox}
\usepackage{dsfont}

\usepackage{anyfontsize}  
\usepackage{microtype} 

\usepackage{newtxtext}
\usepackage{newtxmath}

\makeatletter
\g@addto@macro\bfseries{\boldmath}
\makeatother

\makeatletter
\def\maketitle{
\@author@finish
\title@column\titleblock@produce
\suppressfloats[t]}
\makeatother

\DeclareMathOperator{\poly}{poly}

\newcommand{\Lin}{\operatorname{L}}
\newcommand{\Err}{E} 

\newcommand{\Density}{\operatorname{D}}

\newcommand{\id}{\mathds{1}}

\newcommand{\acc}{{\mathrm{accept}}}

\newcommand{\E}{\mathbb{E}}

\theoremstyle{plain}
\newtheorem{theorem}{Theorem}
\newtheorem{lemma}{Lemma}
\newtheorem{corollary}{Corollary}
\theoremstyle{definition}
\newtheorem{definition}{Definition}

\newtheorem{proposition}{Proposition}
\newtheorem{remark}{Remark}

\renewcommand{\qedsymbol}{$\blacksaare$}

\renewcommand{\qedsymbol}{\unskip\nobreak\quad\qedsymbol}
\renewcommand{\qedsymbol}{$\blacksquare$}

\newcommand{\be}{\begin{equation}\begin{aligned}\hspace{0pt}}
\newcommand{\ee}{\end{aligned}\end{equation}}
\newcommand{\ba}{\begin{eqnarray}}
\newcommand{\ea}{\end{eqnarray}}

\definecolor{burgundy}{rgb}{0.5, 0.0, 0.13}

\definecolor{dblue}{RGB}{14, 34, 102}

\newsavebox{\mstrut}
\newcommand{\bbra}[1]{%
    \sbox{\mstrut}{\(#1\)}%
    \mathinner{\left\langle\kern-0.5\ht\mstrut\left\langle{#1}\right|\mkern-2mu\right|}%
}
\newcommand{\kett}[1]{%
    \sbox{\mstrut}{\(#1\)}%
    \mathinner{\left|\mkern-2mu\left|{#1}\right\rangle\kern-0.5\ht\mstrut\right\rangle}%
}

\usepackage[linesnumbered,ruled,vlined]{algorithm2e}
\SetAlFnt{\small}
\SetKwComment{Comment}{// }{}
\crefname{algocf}{Protocol}{Algorithms}

\SetKwInput{KwData}{Input}
\SetKwInput{KwResult}{Output}

\begin{document}

\title{Constant Overhead Entanglement Distillation via Scrambling}

\author{Andi Gu}\email{andigu@g.harvard.edu}
\affiliation{Department of Physics, Harvard University, Cambridge, MA 02138, USA}

\author{Lorenzo Leone}\email{lorenzo.leone@fu-berlin.de}
\affiliation{Dahlem Center for Complex Quantum Systems, Freie Universit\"{a}t Berlin, 14195 Berlin, Germany}

\author{Kenneth Goodenough}\email{kdgoodenough@gmail.com}
\affiliation{Manning College of Information and Computer Sciences, University of Massachusetts Amherst, Amherst, MA 01002, USA}

\author{Sumeet Khatri}\email{skhatri@vt.edu}
\affiliation{Department of Computer Science, Virginia Tech, Blacksburg, VA 24061, USA}
\affiliation{Virginia Tech Center for Quantum Information Science and Engineering, Blacksburg, VA 24061, USA}

\let\oldaddcontentsline\addcontentsline 
\renewcommand{\addcontentsline}[3]{\oldaddcontentsline{#1}{lot}{#3}} 

\begin{abstract}
High-fidelity quantum entanglement enables key quantum networking capabilities such as secure communication and distributed quantum computing, but long-distance entanglement distribution is limited by noise and loss. Entanglement distillation protocols address this problem by extracting high-fidelity Bell pairs from multiple noisy ones. The primary objective is minimizing the resource overhead: the number of noisy input pairs needed to distill each high-fidelity output pair. While protocols achieving optimal overhead are known in theory, they often require complex decoding operations that make practical implementation challenging. We circumvent this challenge by introducing protocols that use quantum scrambling --- the spreading of quantum information under chaotic dynamics --- through random Clifford operations. Based on this scrambling mechanism, our protocol maintains asymptotically \emph{constant} overhead, independent of the desired output error rate $\bar{\varepsilon}$, and can be implemented with shallow quantum circuits of depth $O(\poly \log \log \bar{\varepsilon}^{-1})$ and memory $O(\poly \log \bar{\varepsilon}^{-1})$. Our protocol remains effective even with noisy quantum gates. 
By incorporating error correction, our protocol achieves state-of-the-art performance: starting with pairs of 10\% initial infidelity, we require only 7 noisy inputs per output pair to distill a single Bell pair with infidelity $\bar{\varepsilon}=10^{-12}$, substantially outperforming existing schemes. We demonstrate the utility of our protocols for quantum repeater networks.
\end{abstract}

\maketitle


Quantum entanglement is a fundamental resource for distributed quantum technologies, enabling secure communication~\cite{BB84,Eke91}, networked computation~\cite{BBC+93,gottesman1999gateteleportation,eisert2000nonlocalgates,nielsen2003MBQC,leung2004MBQC}, and enhanced sensing across quantum devices~\cite{toth2012multipartitemetrology,hyllus2012multipartitemetrology,ZZS18,XZCZ19,guo2020distributedsensing,gottesman2012baseline}. However, distributing entangled states between distant parties is challenging due to noise and loss, which degrade entanglement over long distances~\cite{Kim08}. This has motivated the development of \textit{entanglement distillation} protocols, which extract fewer but high-fidelity entangled states from multiple noisy copies through local operations and classical communication (LOCC).

Foundational works~\cite{BDSW96,BBP96,bennett1996concentrating,DEJMPS96} have shown 
that LOCC can increase fidelity, revealing deep connections between entanglement distillation and quantum error correction: error correction protects against noise accumulating over time, while distillation protects against noise distributed between spatially-separated parties. Subsequent theoretical developments approached the problem from multiple angles: concatenated quantum error-correcting codes~\cite{murao1998multiparticlepurification,DEJMPS96,horodecki1999reduction,chau2011distillationQLDPC,roque2024efficient,pattison2024constoverheaddistillation,shi2024stabilizer}, high-rate quantum low-density parity-check (qLDPC) codes for constant overhead distillation~\cite{ataides2025constant}, permutations of Bell-basis elements~\cite{dehaene2003distillationpermutation,maneva2002purification,BM05,KAJ19,jansen2022enumeratingclifforddistillation,addala2023optimizedpurification,goodenough2023ntokdistillation}, and fundamental limits through quantum Shannon theory~\cite{devetak2004father,devetak2005private,devetak2005distillation,hayden2008decoupling,abeyesinghe2009mother,buscemi2010distilling,leditzky2018useful,FWTD19,regula2023probabilistic,siddhu2024entanglementsharing,abdelhadi2024adaptive}. 

In this work, we introduce and analyze a family of entanglement distillation protocols based on random bilocal Clifford operations. 
Unlike existing protocols based on quantum error-correcting codes, all performance metrics, including output fidelity, acceptance probability, and overhead, have exact, closed-form analytical expressions in terms of only two hyperparameters: $n$ and $k$, the number of input and output Bell pairs, respectively. This eliminates the need for burdensome numerical simulations, which simplifies optimization and analysis of our protocol in practical settings. Our protocols achieve remarkable practical performance: they maintain constant-overhead distillation, meaning the number of noisy input Bell pairs required to distill one Bell pair of infidelity $\bar{\varepsilon}$ is asymptotically bounded by a constant, \emph{for every target infidelity $\bar{\varepsilon}>0$}. This is accomplished through concatenation, allowing systematic fidelity improvement through multiple rounds. 
For target infidelities as low as $\bar{\varepsilon}=10^{-12}$, we achieve state-of-the-art resource overheads compared to recent work~\cite{fowler2010surface,ramette2023fault,sinclair2024fault,pattison2024constoverheaddistillation}; see Table~\ref{tab:comparisons} for a summary of and comparison to prior work. A key advantage is that our protocol achieves the same performance regardless of noise form, working equally well for arbitrarily correlated and independent identically distributed (IID) noise. Finally, our protocols can be implemented with very shallow circuits, in depth $O(\poly \log \log \bar{\varepsilon}^{-1})$, and their flexibility allows adaptation to various hardware constraints, such as limited local memory.

The effectiveness of our protocol stems from scrambling: how quantum information, initially localized to a few qubits, rapidly spreads across an entire quantum system. Scrambling through random unitaries mimics the natural dynamics of chaotic quantum systems, where information rapidly spreads across all available degrees of freedom. In our protocol, random Clifford operations serve as efficient scramblers, rapidly spreading errors across the system in a way that makes them detectable through simple local measurements, without requiring complex decoding procedures (see \cref{fig:schematic}). 

Random operations have a rich history in both classical and quantum information theory. Shannon first showed that random codes could achieve the capacity of classical communication channels~\cite{shannon1948mathematical}.  
Early works on entanglement distillation extended this idea to the quantum domain using random classical codes (e.g., the hashing protocol~\cite{BDSW96}) and random unitaries~\cite{hayden2008decoupling}. These approaches face a common challenge: decoding is widely believed to be computationally intractable~\cite{berlekamp1978intractability,vardy1997intractability,hsieh2011NPhardquantumdecoding,kuo2012hardnessquantumdecoding,iyer2015hardnessquantumdecoding}, to such an extent that this hardness forms the basis of several cryptographic schemes~\cite{mceliece1978public,alekhnovich2003more}, rendering these protocols difficult to implement in practice. Our protocol sidesteps this issue entirely by showing that decoding is unnecessary; simply \emph{detecting} errors with these random codes via local measurements suffices to achieve state-of-the-art scaling.

\begin{figure}
    \centering
    \includegraphics[width=0.80\columnwidth]{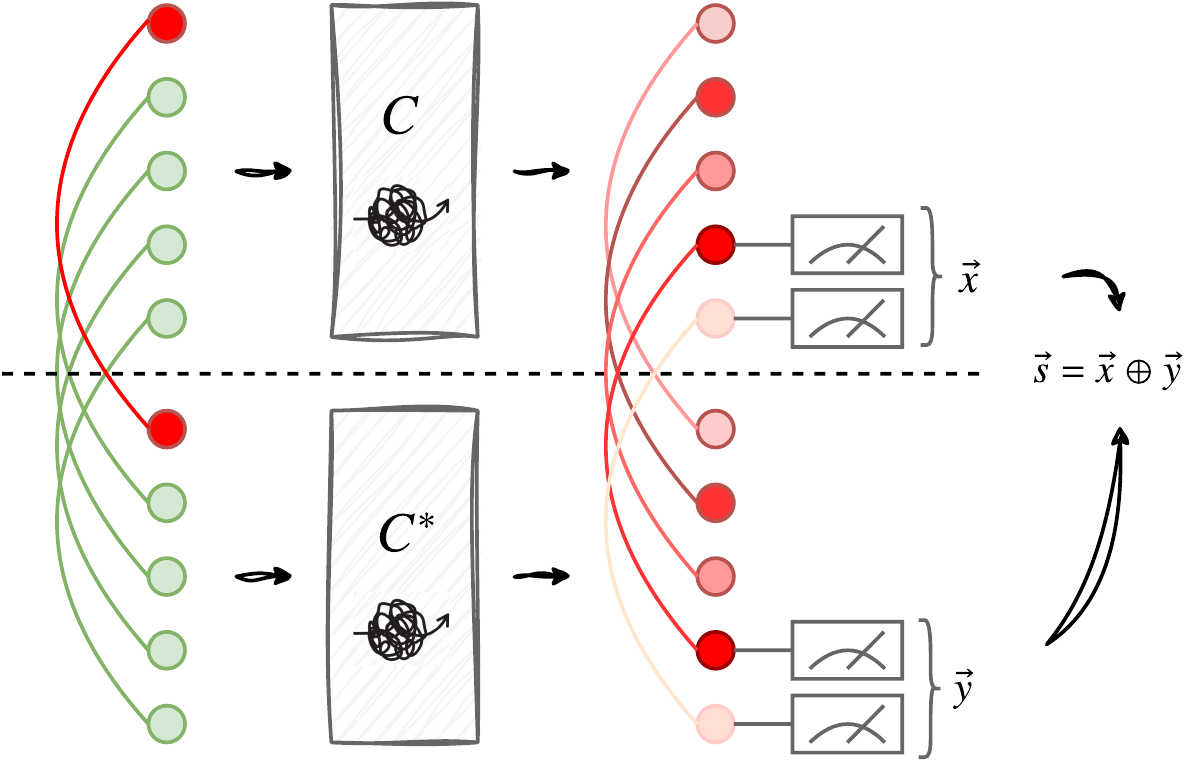}
    \caption{An initial error (indicated in red) gets scrambled into an easily detectable global error under a scrambling unitary $C$, allowing a small number of measurements to detect the presence of the error via the syndrome $\vec{s}$. In the error-detection (`passive') setting, any nonzero syndrome $\vec{s}$ (i.e., disagreement of the measurement outcomes $\vec{x}$ and $\vec{y}$) means the protocol must be rerun, while in the error-correction (`active') setting, $\vec{s}$ is used to infer, then correct, the error.}
    \label{fig:schematic}
\end{figure}

\paragraph*{Random bilocal Clifford protocols.}

In an entanglement distillation protocol, two parties, Alice and Bob, share $n$ noisy entangled qubit pairs. 
Through LOCC, 
they aim to distill $k < n$ pairs with improved fidelity. 
The key performance metric is the overhead $\mathcal{O}\coloneqq n/k$, the number of noisy input pairs needed per high-fidelity output pair. 

Our protocol, presented in \cref{alg:random_bilocal_Clifford},
proceeds as follows. Alice and Bob initially share a $2n$-qubit state $\rho_{A^nB^n}$. This input state can be arbitrary --- it need not be a tensor product of two-qubit noisy Bell pairs, allowing for correlated noise across the qubit pairs. They select a Clifford unitary $C$ uniformly at random, then Alice applies $C$ to her qubits while Bob applies the complex conjugate $C^*$ to his qubits.
This bilocal structure is the central mechanism for error detection through scrambling. When a Pauli error $P$ acts on Bob's half of $n$ Bell pairs, producing the state $\ket{P} \coloneqq (\id_A^{\otimes n} \otimes P)\ket{\Phi}_{AB}^{\otimes n}$, the bilocal operation transforms it as $(C \otimes C^*)\ket{P} = \ket{P'}$, where $P' \coloneqq C^\dagger P C$. Since $C$ is a random Clifford, $P'$ will be high-weight (i.e., global) with overwhelming probability. This identity 
shows how random Clifford operations scramble local Pauli errors into global ones. These global errors become detectable through local measurements, because any $X$ or $Y$ error on a Bell pair leads to Alice and Bob having opposite measurement outcomes. To leverage this error-detection mechanism, Alice and Bob measure their last $m=n-k$ qubits in the computational basis and communicate their measurement outcomes to each other. We denote Alice's measurement outcome as the bitstring $\vec{x}$ and Bob's as $\vec{y}$. The protocol then branches into two possible modes of operation.

\begin{algorithm}[t]
\caption{Entanglement distillation with random bilocal Cliffords}\label{alg:random_bilocal_Clifford}
\KwData{$n\in\{2,3,\dotsc\}$, $k\in\{1,\dotsc,n\}$, $2n$-qubit state $\rho_{A^nB^n}$.}

\KwResult{$2k$-qubit state $\bar{\rho}_{A^kB^k}$.}

accept $\gets$ \texttt{False}

\While{$\lnot$ \textnormal{accept}}{
    
    Alice and Bob select an $n$-qubit Clifford unitary $C$ uniformly at random~\cite{cleve2016near}.
    
    Alice applies $C$ to her qubits, Bob applies the complex conjugate $C^*$ to his. (In a circuit, $C^*$ can be implemented by replacing every phase gate $S$ in $C$ with $S^\dagger$.)
    
    Both measure their first $m \coloneqq n-k$ qubits in the computational basis. Bob sends his measurement outcomes $\vec{y} \in \qty{0,1}^{m}$
    to Alice. 
    
    Alice uses her measurement outcomes, $\vec{x} \in \qty{0,1}^m$, to compute the \emph{syndrome} $\vec{s} = \vec{x} \oplus \vec{y}$.
    
    \eIf(\Comment*[h]{Error-detection}){\textnormal{in `passive' mode}}{
    accept $\gets (\vec{s} == \vec{0})$
    }(\Comment*[h]{`Active' (error-correction)}){
            Alice attempts decoding with the syndrome $\vec{s}$. 
            If it fails, this round has failed; otherwise, it outputs a $k$-qubit Pauli correction that Alice applies to her qubits, and accept $\gets \texttt{True}$.
    }
}
\end{algorithm}

In the \textit{passive} (error detection) setting, the protocol accepts if and only if Alice and Bob's measurement outcomes match exactly, meaning that $\vec{x} = \vec{y}$, so that any detected error leads to rejection. 
In the \textit{active} (error correction) mode, Alice attempts to use the syndrome $\vec{s} = \vec{x} \oplus \vec{y}$ to correct certain errors. If error correction succeeds, Alice applies an appropriate recovery operation and accepts the state; otherwise, the protocol rejects and starts over. 

The performance of our protocol is characterized by the probability of acceptance, $p_{\acc}$, and the joint probability $p_{\acc \land \Phi}$ of both accepting and obtaining the target state $\Phi_{AB}^{\otimes k}$. Their ratio $p_{\acc\land\Phi}/p_{\acc}$ is the output fidelity. 
The expected resource overhead --- the average number of input Bell pairs needed to produce one output Bell pair --- is given by $n/(kp_{\acc})$, because we consume $n$ input pairs per attempt and need $1/p_{\acc}$ attempts on average to succeed in producing $k$ output pairs. 


\begin{figure}
    \centering
    \includegraphics[width=0.95\columnwidth]{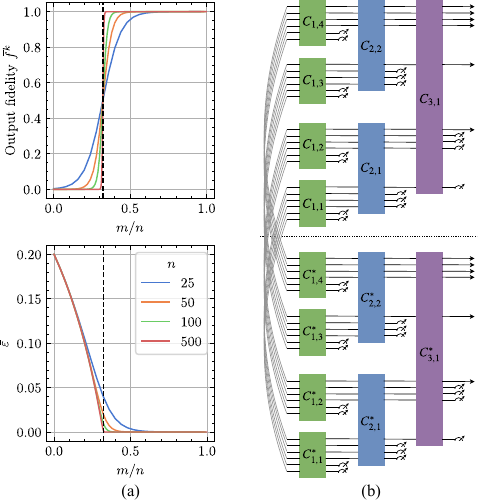}
    \caption{(a) Performance of~\cref{alg:random_bilocal_Clifford} in the passive setting in terms of the fidelity $\bar{f}^k$ of the output state (top) and the single-pair infidelity $\bar{\varepsilon}$ (bottom). The initial single-pair fidelity is set to $f=1-\varepsilon=0.8$. The transition at the fraction $m/n=-\log_2 f$ of measured qubits is shown with a dashed line. (b) A three-layer concatenated protocol, consisting of a $6 \to 4$ protocol (green), an $8 \to 5$ protocol (blue), followed by a $10 \to 6$ protocol (purple).}
    \label{fig:basic_results}
\end{figure}

\paragraph*{Passive setting.}
In the End Matter, we exactly calculate $p_{\acc}$ and $p_{\acc \land \Phi}$ for the passive setting, leading to the following bounds:
\begin{subequations}
\begin{gather}
    f^n \leq p_{\acc} \leq f^n + 2^{-m} (1-f^n),\qq{and} \label{eq:pacc1} \\
\bar{f}^k \geq 1 - 2^{-m}(f^{-n}-1),\label{eq:new-fid}
\end{gather}    
\end{subequations}
where $f \coloneqq \left(\Tr[\Phi_{AB}^{\otimes n}\rho_{A^nB^n}]\right)^{1/n}=1-\varepsilon$ is the initial single-pair fidelity with respect to $n$ perfect Bell pairs $\Phi_{AB}^{\otimes n}$ ($\varepsilon$ is the initial single-pair infidelity) and the single-pair output fidelity is $\bar{f}=1-\bar{\varepsilon}=\left(p_{\acc\land\Phi}/p_{\acc}\right)^{1/k}$~\footnote{The definition of single-pair fidelity comes from the notion of logical error rate in quantum 
error correction, where it makes sense to define a logical error rate per 
logical qubit of an $[[n,k,d]]$ quantum error-correction code. Specifically, if $p_L$ is the overall 
logical error rate of the $[[n,k,d]]$ code, then the per-qubit logical error rate 
is formulated as $1-(1-p_L)^{1/k}$ in both theoretical and experimental works~\cite{roffe2023biastailoredquantum,berent2023softwaretools,wang2025demonstration}.}. Note that if the noise acts independently on each qubit pair (i.e., $\rho_{A^n B^n}$ is a tensor product of $n$ noisy Bell pairs), then $f$ is simply the geometric mean $f=(f_1f_2\dotsb f_n)^{1/n}$ of the fidelities $f_i$ of each pair.

The ability to improve fidelity depends critically on the fraction of measured qubits, $m/n$. As illustrated in \cref{fig:basic_results}(a), there is a transition in output fidelity at $m/n=-\log_2f$: beyond this critical fraction, as $n$ increases, the output fidelity can be increased arbitrarily close to $1$, while below it the output fidelity decays to $0$ (see the Supplemental Material 
for a proof). In contrast, the output infidelity $\bar{\varepsilon}$ improves compared to the input infidelity $\varepsilon$ for \emph{all} values of $m$. While increasing the number $m$ of measured qubits leads to smaller output infidelity, more Bell pairs must be sacrificed through measurement, and the probability of acceptance $p_{\acc}$~\eqref{eq:pacc1} decreases, which increases the number of times the protocol must be rerun and therefore increases the overhead. This suggests operating at moderate values of $m$ that achieve meaningful fidelity improvements while maintaining reasonable acceptance probabilities, and therefore minimizing the overhead.

By repeatedly \textit{concatenating} distillation rounds with carefully moderated values of $m$, we can reach arbitrarily low target infidelities while keeping the total resource requirements controlled. An example of a concatenated protocol is shown in \cref{fig:basic_results}(b). Each of the boxes therein is run in a repeat-until-success fashion. Since each of the Cliffords is independently sampled uniformly at random, for a fixed layer $i$, while each of the $C_{i,j}$ acts on the same number of qubits, the Clifford unitary $C_{i,j}$ itself is in general different for different $j$. If one of the distillation rounds (i.e., one of the boxes) fails, all its dependencies must be rerun: for instance, if one instance of the $C_2$ protocol fails, then the $C_1$ protocol must be rerun (until each succeeds) twice.

The key challenge is choosing appropriate parameters $n$ and $m$ at each layer of concatenation. These parameters must satisfy two properties: first, the output infidelity $\bar{\varepsilon}$ must be appreciably smaller than the input infidelity $\varepsilon$, allowing each successive layer to further purify the state; second, the resource overhead of each layer must be bounded, such that the resulting overhead of the entire protocol remains constant. In practice, we optimize $m$ at each concatenation layer to minimize the expected overhead in that layer. Doing so allows us to achieve overall constant overhead, as described in our main result below.

\begin{theorem}
\label{thm:bilocal_Clifford_concat}
Let $\varepsilon_0 < 0.5$ be an initial infidelity and $\bar{\varepsilon}>0$ be any target infidelity. The concatenated form of the bilocal Clifford protocol (\cref{alg:random_bilocal_Clifford}) succeeds with probability at least $1 - \delta$ and requires:
\begin{enumerate}[itemsep=0.1ex]
    \item an overhead of $\mathcal{O} \leq O(\log \delta^{-1})$;
    \item at most $O((\log \bar{\varepsilon}^{-1})^{3/2} \log \delta^{-1})$ input Bell pairs;
    \item $(\log_2 \bar{\varepsilon}^{-1})^{3/2}$ qubits of memory for each Alice and Bob;
    \item $L \leq O(\log_2^*(\bar{\varepsilon}^{-1}))$ layers of concatenation, where $\log_2^*$ is the iterated logarithm; and
    \item $O(\poly \log \bar{\varepsilon}^{-1})$ gates, which can be organized to run in $O(\poly \log \log \bar{\varepsilon}^{-1})$ depth.
\end{enumerate}
\end{theorem}

Our numerical results (\cref{fig:results}) demonstrate the performance advantage over existing schemes in regimes of practical interest, such as initial infidelities $\varepsilon_0\approx 10^{-1}$.


\begin{figure}
    \centering
    \includegraphics[width=0.80\columnwidth]{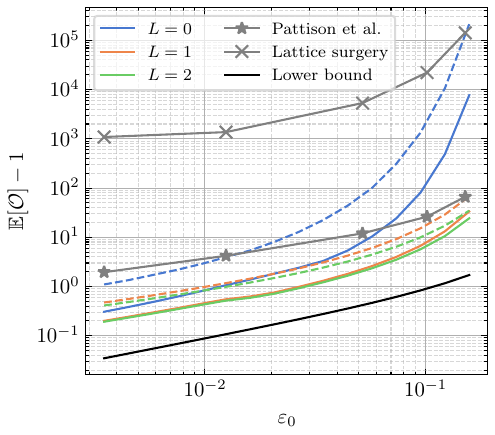}
    \caption{Overhead of our entanglement distillation protocols, with $L$ layers of concatenation, as a function of the initial single-pair infidelity $\varepsilon_0$. Dashed lines indicate the passive setting of our \cref{alg:random_bilocal_Clifford}. Solid lines indicate our protocol with the first layer having active error correction (with a maximum $\Err=3 \times 10^6$), assuming IID depolarizing noise. 
    The parameters $n$ and $m$ at each layer of the protocol are optimized to achieve the lowest expected overhead, subject to the constraint that the output infidelity $\bar{\varepsilon}$ is below $10^{-12}$. 
    We compare with \cite{pattison2024constoverheaddistillation} and protocols using lattice surgery~\cite{fowler2010surface,ramette2023fault,sinclair2024fault}. The solid black line indicates a lower bound for the best achievable overhead given by the Rains bound on distillable entanglement~\cite{rains1999bound,Rains01}. 
    }\label{fig:results}
\end{figure}

So far, we have assumed Alice and Bob's local operations are noiseless. 
With noisy local operations, the number of possible concatenation layers will be limited.
One approach to handling noisy local operations is to encode each half of every Bell pair into a quantum error-correcting code. A particularly suitable code is the 2D color code~\cite{bombin2006topological}, which admits transversal Clifford operations~\footnote{This encoding process typically proceeds through state injection, where a physical state is incrementally grown into a logical state through repeated application of Clifford operations. While this injection process is not fault-tolerant and introduces additional errors at a rate $\varepsilon_I$, these errors combine with the channel noise $\varepsilon$ to produce a new effective channel with error rate $\sim \varepsilon + \varepsilon_I$. Importantly, the injection error rate $\varepsilon_I$ typically scales with the local gate noise~\cite{zhang2024facilitating,li2015magic,fowler2009high,raussendorf2006fault,horsman2012surface}, which is generally much smaller than the channel noise $\varepsilon$ (typical numbers are $10^{-3} \leq \varepsilon_I \leq 10^{-2}$, while $\varepsilon \sim 10^{-1}$).}. Thus, when local quantum error correction is available, the conclusions of our analysis remain valid even under realistic noisy operations. When full error correction may not be feasible, we must directly confront the impact of noisy gates. We thus analyze random bilocal Clifford circuits with a finite number of noisy random gates (see the End Matter). We show that even with noisy local operations, our protocol can achieve output fidelities limited only by the local gate noise, which is typically several orders of magnitude better than the channel noise affecting the input Bell pairs.


\paragraph*{Active setting.}
Instead of discarding trials in which errors are detected, we can attempt to actively correct certain errors using maximum-likelihood decoding. Let $\{q_\ell\}_{\ell=1}^{4^n}$ be Pauli-error probabilities in decreasing order~\footnote{By Pauli twirling --- applying the same random $n$-qubit Pauli on both halves, any channel is equivalent to a mixture of Pauli errors $P_\ell$ with probabilities $q_\ell$~\cite{BDSW96}.}. An \emph{$\Err$-active} strategy attempts to correct only the $\Err+1$ most probable errors. We precompute a lookup table that maps the observed syndrome $\vec{s}$ to its most likely Pauli $P_\ell$ and, when $\vec{s}$ occurs, apply the recovery $C P_\ell C^\dagger$; if $\vec{s}$ is absent from the table, we declare failure and restart~\footnote{This lookup-based correction is likely optimal for random Clifford codes, as decoding random codes is widely believed to be computationally intractable~\cite{berlekamp1978intractability,mceliece1978public,alekhnovich2003more}.}. The case $\Err=0$ recovers the passive protocol, where we discard anything that has nonzero syndrome. In the general case, by restricting to the top $\Err+1$ errors, the classical preprocessing and storage cost of this decoding is $O(\Err)$. This approach is effective when probability mass is concentrated on a small set of (typically low-weight) Paulis, as in IID depolarizing noise, so the table can be built by enumerating increasing-weight errors. Since it succeeds whenever the true error lies among the tabulated ones, its benefit grows with the cumulative mass $\sum_{\ell\le \Err+1} q_\ell$. In practice, it is most useful only at the first concatenation layer: subsequent random Clifford rounds scramble the error distribution, reducing the concentration that makes maximum-likelihood corrections reliable.

\begin{theorem}
\label{thm:active-perform}
Under an $\Err$-active correction strategy for \cref{alg:random_bilocal_Clifford}, the fidelity of the output state, conditioned on acceptance, is bounded from below as follows:
\begin{equation}
    (1-\bar{\varepsilon})^k \geq 1 - 2^{-m} \cdot (\Err+1) (q^{-1}-1),
\end{equation}
where $q=f^n+\sum_{\ell=2}^{\Err+1}q_{\ell}$. With $\Err=0$, we recover \cref{eq:new-fid}.
\end{theorem}

Although active correction introduces an extra factor of $\Err+1$ in the fidelity bound (reducing output fidelity compared to the passive case), it greatly improves acceptance probability by correcting likely errors instead of discarding them. For structured noise models dominated by low-weight errors, acceptance can improve from $p_{\acc} \sim f^n$ to $p_{\acc} \sim f^n + \sum_{\ell=2}^{E+1} q_\ell$. The power of active correction thus lies in this trade-off: a modest decrease in output fidelity is more than compensated for by a marked increase in acceptance probability, leading to lower overhead. In \cref{fig:results}, we demonstrate this advantage, showing that optimized active correction protocols achieve significantly better overhead than both passive protocols and prior work~\cite{pattison2024constoverheaddistillation,fowler2010surface,ramette2023fault,sinclair2024fault}

\paragraph*{Application: quantum repeaters.}

A key challenge in quantum technologies is distributing high-quality entanglement over long distances. Photons carrying quantum information get lost or corrupted exponentially in distance~\cite{VLMN09}.
Quantum repeaters address this by breaking long distances into shorter segments, with repeater nodes that purify and relay entanglement between segments~\cite{VLMN09,AME11}. Our entanglement distillation protocols are particularly well-suited for quantum repeaters, as demonstrated by two applications that make use of high-fidelity entanglement.

The first application is long-baseline quantum-optical-enhanced interferometry~\cite{gottesman2012baseline}, where combining light from distant telescopes enables higher resolution imaging. 
The second application is quantum key distribution (QKD)~\cite{BB84,Eke91}, which generates provably secure cryptographic keys between distant parties. 
As we show in \cref{fig:applications_baseline}, we find that our protocols outperform existing practical schemes~\cite{pattison2024constoverheaddistillation,BBP96} and approach the performance of the 
hashing protocol. We refer to the End Matter for detailed performance analysis and metrics.

\begin{figure}
    \centering
\includegraphics[width=0.85\columnwidth]{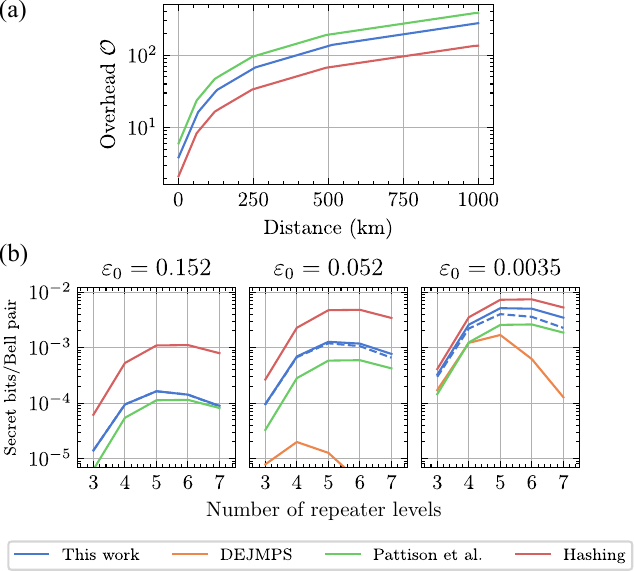}
\caption{Random bilocal Clifford protocols applied to: (a) overheads for long-baseline interferometry; and (b) secret-key rates for QKD. The initial infidelity in (a) is $\varepsilon_0=0.0035$. The target infidelity in both cases is $\bar{\varepsilon}=10^{-9}$. In (b), each level halves the communication distance via intermediate repeater stations. The dashed line indicates our protocol with $n \leq 12$.
}\label{fig:applications_baseline}
\end{figure}

\paragraph*{Summary and outlook.}

We have introduced a family of entanglement distillation protocols based on random bilocal Clifford operations. The key insight is that random operations, traditionally viewed as primarily theoretical tools, can achieve state-of-the-art practical performance through pure error detection. This sidesteps the complexity barrier that has historically limited the practical utility of random coding approaches.


Our protocol achieves constant overhead distillation through careful parameter selection at each concatenation layer, requiring only $O(\log_2^* \bar{\varepsilon}^{-1})$ layers to reach a target infidelity $\bar{\varepsilon}$. The implementation requirements are modest: we require shallow quantum circuits of depth $O(\poly \log \log \bar{\varepsilon}^{-1})$, and only $O(\poly \log \bar{\varepsilon}^{-1})$ qubits of local memory.

It would be interesting to generalize 
the random Clifford framework to multipartite settings. The connection between random operations and error detection might also find applications in other quantum protocols where decoding complexity has been a limiting factor. 
We also expect our results, particularly the analysis of finite-depth implementations with noisy gates, to inform the design and implementation of fault-tolerant interconnects for distributed quantum computing~\cite{awschalom2021interconnects,ramette2023fault,pattison2024constoverheaddistillation,sinclair2024fault}. 

A key open question is whether the resource requirements of our protocol, particularly the circuit depth and memory overhead, are optimal. Although qLDPC-based schemes require only constant depth, it is unclear whether this favorable scaling extends to error-detection-based protocols. Indeed, the $O(\poly \log \log \bar{\varepsilon}^{-1})$ depth scaling seems difficult to improve upon, given the need for information to propagate between qubits, but formal lower bounds remain elusive. 

While we analyzed the protocol using exact unitary 2-designs from random Clifford circuits, its key ingredient is the rapid spreading of Pauli errors, which also arises in generic chaotic dynamics that approximate such designs~\cite{roberts2017chaos,cotler2017chaos,hosur2016chaos,landsman2019verified}. This suggests our distillation scheme may work in naturally fast-scrambling systems --- such as trapped ions~\cite{bohnet2016quantum} or Rydberg atoms~\cite{bernien2017probing} --- potentially reducing the need for precise digital control.

\bigskip

\begin{acknowledgments}
A.G. thanks Pablo Bonilla for several insightful discussions on error correction and distillation. A.G. acknowledges support from the IBM PhD Fellowship. L.L. is funded through the Munich Quantum Valley project (MQV-K8) by Bayerisches Staatsministerium für Wissenschaft und Kunst. S.K. acknowledges financial support from the US Department of Energy, Office of Science, Advanced Scientific Computing Research program, under award number DE-SC0025430.
\end{acknowledgments}

\section*{Data availability}

The data that support the findings of this article are openly available~\cite{data}.


\bibliography{refs}

\begin{thebibliography}{134}%
\makeatletter
\providecommand \@ifxundefined [1]{%
 \@ifx{#1\undefined}
}%
\providecommand \@ifnum [1]{%
 \ifnum #1\expandafter \@firstoftwo
 \else \expandafter \@secondoftwo
 \fi
}%
\providecommand \@ifx [1]{%
 \ifx #1\expandafter \@firstoftwo
 \else \expandafter \@secondoftwo
 \fi
}%
\providecommand \natexlab [1]{#1}%
\providecommand \enquote  [1]{``#1''}%
\providecommand \bibnamefont  [1]{#1}%
\providecommand \bibfnamefont [1]{#1}%
\providecommand \citenamefont [1]{#1}%
\providecommand \href@noop [0]{\@secondoftwo}%
\providecommand \href [0]{\begingroup \@sanitize@url \@href}%
\providecommand \@href[1]{\@@startlink{#1}\@@href}%
\providecommand \@@href[1]{\endgroup#1\@@endlink}%
\providecommand \@sanitize@url [0]{\catcode `\\12\catcode `\$12\catcode `\&12\catcode `\#12\catcode `\^12\catcode `\_12\catcode `\%12\relax}%
\providecommand \@@startlink[1]{}%
\providecommand \@@endlink[0]{}%
\providecommand \url  [0]{\begingroup\@sanitize@url \@url }%
\providecommand \@url [1]{\endgroup\@href {#1}{\urlprefix }}%
\providecommand \urlprefix  [0]{URL }%
\providecommand \Eprint [0]{\href }%
\providecommand \doibase [0]{https://doi.org/}%
\providecommand \selectlanguage [0]{\@gobble}%
\providecommand \bibinfo  [0]{\@secondoftwo}%
\providecommand \bibfield  [0]{\@secondoftwo}%
\providecommand \translation [1]{[#1]}%
\providecommand \BibitemOpen [0]{}%
\providecommand \bibitemStop [0]{}%
\providecommand \bibitemNoStop [0]{.\EOS\space}%
\providecommand \EOS [0]{\spacefactor3000\relax}%
\providecommand \BibitemShut  [1]{\csname bibitem#1\endcsname}%
\let\auto@bib@innerbib\@empty
\bibitem [{\citenamefont {Bennett}\ and\ \citenamefont {Brassard}(1984)}]{BB84}%
  \BibitemOpen
  \bibfield  {author} {\bibinfo {author} {\bibfnamefont {C.~H.}\ \bibnamefont {Bennett}}\ and\ \bibinfo {author} {\bibfnamefont {G.}~\bibnamefont {Brassard}},\ }\bibfield  {title} {\bibinfo {title} {{Quantum cryptography: Public key distribution and coin tossing}},\ }in\ \href {https://doi.org/10.1016/j.tcs.2014.05.025} {\emph {\bibinfo {booktitle} {{International Conference on Computer System and Signal Processing, IEEE}}}}\ (\bibinfo {year} {1984})\ pp.\ \bibinfo {pages} {175--179}\BibitemShut {NoStop}%
\bibitem [{\citenamefont {Ekert}(1991)}]{Eke91}%
  \BibitemOpen
  \bibfield  {author} {\bibinfo {author} {\bibfnamefont {A.~K.}\ \bibnamefont {Ekert}},\ }\bibfield  {title} {\bibinfo {title} {{Quantum cryptography based on Bell's theorem}},\ }\href {https://doi.org/10.1103/PhysRevLett.67.661} {\bibfield  {journal} {\bibinfo  {journal} {Physical Review Letters}\ }\textbf {\bibinfo {volume} {67}},\ \bibinfo {pages} {661} (\bibinfo {year} {1991})}\BibitemShut {NoStop}%
\bibitem [{\citenamefont {Bennett}\ \emph {et~al.}(1993)\citenamefont {Bennett}, \citenamefont {Brassard}, \citenamefont {Cr\'epeau}, \citenamefont {Jozsa}, \citenamefont {Peres},\ and\ \citenamefont {Wootters}}]{BBC+93}%
  \BibitemOpen
  \bibfield  {author} {\bibinfo {author} {\bibfnamefont {C.~H.}\ \bibnamefont {Bennett}}, \bibinfo {author} {\bibfnamefont {G.}~\bibnamefont {Brassard}}, \bibinfo {author} {\bibfnamefont {C.}~\bibnamefont {Cr\'epeau}}, \bibinfo {author} {\bibfnamefont {R.}~\bibnamefont {Jozsa}}, \bibinfo {author} {\bibfnamefont {A.}~\bibnamefont {Peres}},\ and\ \bibinfo {author} {\bibfnamefont {W.~K.}\ \bibnamefont {Wootters}},\ }\bibfield  {title} {\bibinfo {title} {{Teleporting an unknown quantum state via dual classical and Einstein-Podolsky-Rosen channels}},\ }\href {https://doi.org/10.1103/PhysRevLett.70.1895} {\bibfield  {journal} {\bibinfo  {journal} {Physical Review Letters}\ }\textbf {\bibinfo {volume} {70}},\ \bibinfo {pages} {1895} (\bibinfo {year} {1993})}\BibitemShut {NoStop}%
\bibitem [{\citenamefont {Gottesman}\ and\ \citenamefont {Chuang}(1999)}]{gottesman1999gateteleportation}%
  \BibitemOpen
  \bibfield  {author} {\bibinfo {author} {\bibfnamefont {D.}~\bibnamefont {Gottesman}}\ and\ \bibinfo {author} {\bibfnamefont {I.~L.}\ \bibnamefont {Chuang}},\ }\bibfield  {title} {\bibinfo {title} {Demonstrating the viability of universal quantum computation using teleportation and single-qubit operations},\ }\href {https://doi.org/10.1038/46503} {\bibfield  {journal} {\bibinfo  {journal} {Nature}\ }\textbf {\bibinfo {volume} {402}},\ \bibinfo {pages} {390–393} (\bibinfo {year} {1999})}\BibitemShut {NoStop}%
\bibitem [{\citenamefont {Dankert}\ \emph {et~al.}(2009{\natexlab{a}})\citenamefont {Dankert}, \citenamefont {Cleve}, \citenamefont {Emerson},\ and\ \citenamefont {Livine}}]{eisert2000nonlocalgates}%
  \BibitemOpen
  \bibfield  {author} {\bibinfo {author} {\bibfnamefont {C.}~\bibnamefont {Dankert}}, \bibinfo {author} {\bibfnamefont {R.}~\bibnamefont {Cleve}}, \bibinfo {author} {\bibfnamefont {J.}~\bibnamefont {Emerson}},\ and\ \bibinfo {author} {\bibfnamefont {E.}~\bibnamefont {Livine}},\ }\bibfield  {title} {\bibinfo {title} {Exact and approximate unitary 2-designs and their application to fidelity estimation},\ }\href {https://doi.org/10.1103/PhysRevA.80.012304} {\bibfield  {journal} {\bibinfo  {journal} {Phys. Rev. A}\ }\textbf {\bibinfo {volume} {80}},\ \bibinfo {pages} {012304} (\bibinfo {year} {2009}{\natexlab{a}})}\BibitemShut {NoStop}%
\bibitem [{\citenamefont {Nielsen}(2003)}]{nielsen2003MBQC}%
  \BibitemOpen
  \bibfield  {author} {\bibinfo {author} {\bibfnamefont {M.~A.}\ \bibnamefont {Nielsen}},\ }\bibfield  {title} {\bibinfo {title} {Quantum computation by measurement and quantum memory},\ }\href {https://doi.org/10.1016/s0375-9601(02)01803-0} {\bibfield  {journal} {\bibinfo  {journal} {Physics Letters A}\ }\textbf {\bibinfo {volume} {308}},\ \bibinfo {pages} {96–100} (\bibinfo {year} {2003})}\BibitemShut {NoStop}%
\bibitem [{\citenamefont {Leung}(2004)}]{leung2004MBQC}%
  \BibitemOpen
  \bibfield  {author} {\bibinfo {author} {\bibfnamefont {D.~W.}\ \bibnamefont {Leung}},\ }\bibfield  {title} {\bibinfo {title} {Quantum computation by measurements},\ }\href {https://doi.org/10.1142/S0219749904000055} {\bibfield  {journal} {\bibinfo  {journal} {International Journal of Quantum Information}\ }\textbf {\bibinfo {volume} {02}},\ \bibinfo {pages} {33} (\bibinfo {year} {2004})}\BibitemShut {NoStop}%
\bibitem [{\citenamefont {T\'oth}(2012)}]{toth2012multipartitemetrology}%
  \BibitemOpen
  \bibfield  {author} {\bibinfo {author} {\bibfnamefont {G.}~\bibnamefont {T\'oth}},\ }\bibfield  {title} {\bibinfo {title} {Multipartite entanglement and high-precision metrology},\ }\href {https://doi.org/10.1103/PhysRevA.85.022322} {\bibfield  {journal} {\bibinfo  {journal} {Physical Review A}\ }\textbf {\bibinfo {volume} {85}},\ \bibinfo {pages} {022322} (\bibinfo {year} {2012})}\BibitemShut {NoStop}%
\bibitem [{\citenamefont {Hyllus}\ \emph {et~al.}(2012)\citenamefont {Hyllus}, \citenamefont {Laskowski}, \citenamefont {Krischek}, \citenamefont {Schwemmer}, \citenamefont {Wieczorek}, \citenamefont {Weinfurter}, \citenamefont {Pezz\'e},\ and\ \citenamefont {Smerzi}}]{hyllus2012multipartitemetrology}%
  \BibitemOpen
  \bibfield  {author} {\bibinfo {author} {\bibfnamefont {P.}~\bibnamefont {Hyllus}}, \bibinfo {author} {\bibfnamefont {W.}~\bibnamefont {Laskowski}}, \bibinfo {author} {\bibfnamefont {R.}~\bibnamefont {Krischek}}, \bibinfo {author} {\bibfnamefont {C.}~\bibnamefont {Schwemmer}}, \bibinfo {author} {\bibfnamefont {W.}~\bibnamefont {Wieczorek}}, \bibinfo {author} {\bibfnamefont {H.}~\bibnamefont {Weinfurter}}, \bibinfo {author} {\bibfnamefont {L.}~\bibnamefont {Pezz\'e}},\ and\ \bibinfo {author} {\bibfnamefont {A.}~\bibnamefont {Smerzi}},\ }\bibfield  {title} {\bibinfo {title} {Fisher information and multiparticle entanglement},\ }\href {https://doi.org/10.1103/PhysRevA.85.022321} {\bibfield  {journal} {\bibinfo  {journal} {Physical Review A}\ }\textbf {\bibinfo {volume} {85}},\ \bibinfo {pages} {022321} (\bibinfo {year} {2012})}\BibitemShut {NoStop}%
\bibitem [{\citenamefont {Zhuang}\ \emph {et~al.}(2018)\citenamefont {Zhuang}, \citenamefont {Zhang},\ and\ \citenamefont {Shapiro}}]{ZZS18}%
  \BibitemOpen
  \bibfield  {author} {\bibinfo {author} {\bibfnamefont {Q.}~\bibnamefont {Zhuang}}, \bibinfo {author} {\bibfnamefont {Z.}~\bibnamefont {Zhang}},\ and\ \bibinfo {author} {\bibfnamefont {J.~H.}\ \bibnamefont {Shapiro}},\ }\bibfield  {title} {\bibinfo {title} {{Distributed quantum sensing using continuous-variable multipartite entanglement}},\ }\href {https://doi.org/10.1103/PhysRevA.97.032329} {\bibfield  {journal} {\bibinfo  {journal} {Physical Review A}\ }\textbf {\bibinfo {volume} {97}},\ \bibinfo {pages} {032329} (\bibinfo {year} {2018})}\BibitemShut {NoStop}%
\bibitem [{\citenamefont {Xia}\ \emph {et~al.}(2019)\citenamefont {Xia}, \citenamefont {Zhuang}, \citenamefont {Clark},\ and\ \citenamefont {Zhang}}]{XZCZ19}%
  \BibitemOpen
  \bibfield  {author} {\bibinfo {author} {\bibfnamefont {Y.}~\bibnamefont {Xia}}, \bibinfo {author} {\bibfnamefont {Q.}~\bibnamefont {Zhuang}}, \bibinfo {author} {\bibfnamefont {W.}~\bibnamefont {Clark}},\ and\ \bibinfo {author} {\bibfnamefont {Z.}~\bibnamefont {Zhang}},\ }\bibfield  {title} {\bibinfo {title} {{Repeater-enhanced distributed quantum sensing based on continuous-variable multipartite entanglement}},\ }\href {https://doi.org/10.1103/PhysRevA.99.012328} {\bibfield  {journal} {\bibinfo  {journal} {Physical Review A}\ }\textbf {\bibinfo {volume} {99}},\ \bibinfo {pages} {012328} (\bibinfo {year} {2019})}\BibitemShut {NoStop}%
\bibitem [{\citenamefont {Guo}\ \emph {et~al.}(2019)\citenamefont {Guo}, \citenamefont {Breum}, \citenamefont {Borregaard}, \citenamefont {Izumi}, \citenamefont {Larsen}, \citenamefont {Gehring}, \citenamefont {Christandl}, \citenamefont {Neergaard-Nielsen},\ and\ \citenamefont {Andersen}}]{guo2020distributedsensing}%
  \BibitemOpen
  \bibfield  {author} {\bibinfo {author} {\bibfnamefont {X.}~\bibnamefont {Guo}}, \bibinfo {author} {\bibfnamefont {C.~R.}\ \bibnamefont {Breum}}, \bibinfo {author} {\bibfnamefont {J.}~\bibnamefont {Borregaard}}, \bibinfo {author} {\bibfnamefont {S.}~\bibnamefont {Izumi}}, \bibinfo {author} {\bibfnamefont {M.~V.}\ \bibnamefont {Larsen}}, \bibinfo {author} {\bibfnamefont {T.}~\bibnamefont {Gehring}}, \bibinfo {author} {\bibfnamefont {M.}~\bibnamefont {Christandl}}, \bibinfo {author} {\bibfnamefont {J.~S.}\ \bibnamefont {Neergaard-Nielsen}},\ and\ \bibinfo {author} {\bibfnamefont {U.~L.}\ \bibnamefont {Andersen}},\ }\bibfield  {title} {\bibinfo {title} {Distributed quantum sensing in a continuous-variable entangled network},\ }\href {https://doi.org/10.1038/s41567-019-0743-x} {\bibfield  {journal} {\bibinfo  {journal} {Nature Physics}\ }\textbf {\bibinfo {volume} {16}},\ \bibinfo {pages} {281–284} (\bibinfo {year} {2019})}\BibitemShut {NoStop}%
\bibitem [{\citenamefont {Gottesman}\ \emph {et~al.}(2012)\citenamefont {Gottesman}, \citenamefont {Jennewein},\ and\ \citenamefont {Croke}}]{gottesman2012baseline}%
  \BibitemOpen
  \bibfield  {author} {\bibinfo {author} {\bibfnamefont {D.}~\bibnamefont {Gottesman}}, \bibinfo {author} {\bibfnamefont {T.}~\bibnamefont {Jennewein}},\ and\ \bibinfo {author} {\bibfnamefont {S.}~\bibnamefont {Croke}},\ }\bibfield  {title} {\bibinfo {title} {{Longer-Baseline Telescopes Using Quantum Repeaters}},\ }\href {https://doi.org/10.1103/PhysRevLett.109.070503} {\bibfield  {journal} {\bibinfo  {journal} {Physical Review Letters}\ }\textbf {\bibinfo {volume} {109}},\ \bibinfo {pages} {070503} (\bibinfo {year} {2012})}\BibitemShut {NoStop}%
\bibitem [{\citenamefont {Kimble}(2008)}]{Kim08}%
  \BibitemOpen
  \bibfield  {author} {\bibinfo {author} {\bibfnamefont {H.~J.}\ \bibnamefont {Kimble}},\ }\bibfield  {title} {\bibinfo {title} {{The quantum internet}},\ }\href {https://doi.org/10.1038/nature07127} {\bibfield  {journal} {\bibinfo  {journal} {Nature}\ }\textbf {\bibinfo {volume} {453}} (\bibinfo {year} {2008})}\BibitemShut {NoStop}%
\bibitem [{\citenamefont {Bennett}\ \emph {et~al.}(1996{\natexlab{a}})\citenamefont {Bennett}, \citenamefont {DiVincenzo}, \citenamefont {Smolin},\ and\ \citenamefont {Wootters}}]{BDSW96}%
  \BibitemOpen
  \bibfield  {author} {\bibinfo {author} {\bibfnamefont {C.~H.}\ \bibnamefont {Bennett}}, \bibinfo {author} {\bibfnamefont {D.~P.}\ \bibnamefont {DiVincenzo}}, \bibinfo {author} {\bibfnamefont {J.~A.}\ \bibnamefont {Smolin}},\ and\ \bibinfo {author} {\bibfnamefont {W.~K.}\ \bibnamefont {Wootters}},\ }\bibfield  {title} {\bibinfo {title} {{Mixed-state entanglement and quantum error correction}},\ }\href {https://doi.org/10.1103/PhysRevA.54.3824} {\bibfield  {journal} {\bibinfo  {journal} {Physical Review A}\ }\textbf {\bibinfo {volume} {54}},\ \bibinfo {pages} {3824} (\bibinfo {year} {1996}{\natexlab{a}})}\BibitemShut {NoStop}%
\bibitem [{\citenamefont {Bennett}\ \emph {et~al.}(1996{\natexlab{b}})\citenamefont {Bennett}, \citenamefont {Brassard}, \citenamefont {Popescu}, \citenamefont {Schumacher}, \citenamefont {Smolin},\ and\ \citenamefont {Wootters}}]{BBP96}%
  \BibitemOpen
  \bibfield  {author} {\bibinfo {author} {\bibfnamefont {C.~H.}\ \bibnamefont {Bennett}}, \bibinfo {author} {\bibfnamefont {G.}~\bibnamefont {Brassard}}, \bibinfo {author} {\bibfnamefont {S.}~\bibnamefont {Popescu}}, \bibinfo {author} {\bibfnamefont {B.}~\bibnamefont {Schumacher}}, \bibinfo {author} {\bibfnamefont {J.~A.}\ \bibnamefont {Smolin}},\ and\ \bibinfo {author} {\bibfnamefont {W.~K.}\ \bibnamefont {Wootters}},\ }\bibfield  {title} {\bibinfo {title} {{Purification of Noisy Entanglement and Faithful Teleportation via Noisy Channels}},\ }\href {https://doi.org/10.1103/PhysRevLett.76.722} {\bibfield  {journal} {\bibinfo  {journal} {Physical Review Letters}\ }\textbf {\bibinfo {volume} {76}},\ \bibinfo {pages} {722} (\bibinfo {year} {1996}{\natexlab{b}})}\BibitemShut {NoStop}%
\bibitem [{\citenamefont {Bennett}\ \emph {et~al.}(1996{\natexlab{c}})\citenamefont {Bennett}, \citenamefont {Bernstein}, \citenamefont {Popescu},\ and\ \citenamefont {Schumacher}}]{bennett1996concentrating}%
  \BibitemOpen
  \bibfield  {author} {\bibinfo {author} {\bibfnamefont {C.~H.}\ \bibnamefont {Bennett}}, \bibinfo {author} {\bibfnamefont {H.~J.}\ \bibnamefont {Bernstein}}, \bibinfo {author} {\bibfnamefont {S.}~\bibnamefont {Popescu}},\ and\ \bibinfo {author} {\bibfnamefont {B.}~\bibnamefont {Schumacher}},\ }\bibfield  {title} {\bibinfo {title} {{Concentrating partial entanglement by local operations}},\ }\href {https://doi.org/10.1103/PhysRevA.53.2046} {\bibfield  {journal} {\bibinfo  {journal} {Physical Review A}\ }\textbf {\bibinfo {volume} {53}},\ \bibinfo {pages} {2046} (\bibinfo {year} {1996}{\natexlab{c}})}\BibitemShut {NoStop}%
\bibitem [{\citenamefont {Deutsch}\ \emph {et~al.}(1996)\citenamefont {Deutsch}, \citenamefont {Ekert}, \citenamefont {Jozsa}, \citenamefont {Macchiavello}, \citenamefont {Popescu},\ and\ \citenamefont {Sanpera}}]{DEJMPS96}%
  \BibitemOpen
  \bibfield  {author} {\bibinfo {author} {\bibfnamefont {D.}~\bibnamefont {Deutsch}}, \bibinfo {author} {\bibfnamefont {A.}~\bibnamefont {Ekert}}, \bibinfo {author} {\bibfnamefont {R.}~\bibnamefont {Jozsa}}, \bibinfo {author} {\bibfnamefont {C.}~\bibnamefont {Macchiavello}}, \bibinfo {author} {\bibfnamefont {S.}~\bibnamefont {Popescu}},\ and\ \bibinfo {author} {\bibfnamefont {A.}~\bibnamefont {Sanpera}},\ }\bibfield  {title} {\bibinfo {title} {{Quantum Privacy Amplification and the Security of Quantum Cryptography over Noisy Channels}},\ }\href {https://doi.org/10.1103/PhysRevLett.77.2818} {\bibfield  {journal} {\bibinfo  {journal} {Physical Review Letters}\ }\textbf {\bibinfo {volume} {77}},\ \bibinfo {pages} {2818} (\bibinfo {year} {1996})}\BibitemShut {NoStop}%
\bibitem [{\citenamefont {Murao}\ \emph {et~al.}(1998)\citenamefont {Murao}, \citenamefont {Plenio}, \citenamefont {Popescu}, \citenamefont {Vedral},\ and\ \citenamefont {Knight}}]{murao1998multiparticlepurification}%
  \BibitemOpen
  \bibfield  {author} {\bibinfo {author} {\bibfnamefont {M.}~\bibnamefont {Murao}}, \bibinfo {author} {\bibfnamefont {M.~B.}\ \bibnamefont {Plenio}}, \bibinfo {author} {\bibfnamefont {S.}~\bibnamefont {Popescu}}, \bibinfo {author} {\bibfnamefont {V.}~\bibnamefont {Vedral}},\ and\ \bibinfo {author} {\bibfnamefont {P.~L.}\ \bibnamefont {Knight}},\ }\bibfield  {title} {\bibinfo {title} {{Multiparticle entanglement purification protocols}},\ }\href {https://doi.org/10.1103/PhysRevA.57.R4075} {\bibfield  {journal} {\bibinfo  {journal} {Physical Review A}\ }\textbf {\bibinfo {volume} {57}},\ \bibinfo {pages} {R4075} (\bibinfo {year} {1998})}\BibitemShut {NoStop}%
\bibitem [{\citenamefont {Horodecki}\ and\ \citenamefont {Horodecki}(1999)}]{horodecki1999reduction}%
  \BibitemOpen
  \bibfield  {author} {\bibinfo {author} {\bibfnamefont {M.}~\bibnamefont {Horodecki}}\ and\ \bibinfo {author} {\bibfnamefont {P.}~\bibnamefont {Horodecki}},\ }\bibfield  {title} {\bibinfo {title} {{Reduction criterion of separability and limits for a class of distillation protocols}},\ }\href {https://doi.org/10.1103/PhysRevA.59.4206} {\bibfield  {journal} {\bibinfo  {journal} {Physical Review A}\ }\textbf {\bibinfo {volume} {59}},\ \bibinfo {pages} {4206} (\bibinfo {year} {1999})}\BibitemShut {NoStop}%
\bibitem [{\citenamefont {Chau}\ and\ \citenamefont {Ho}(2011)}]{chau2011distillationQLDPC}%
  \BibitemOpen
  \bibfield  {author} {\bibinfo {author} {\bibfnamefont {H.~F.}\ \bibnamefont {Chau}}\ and\ \bibinfo {author} {\bibfnamefont {K.~H.}\ \bibnamefont {Ho}},\ }\bibfield  {title} {\bibinfo {title} {Practical entanglement distillation scheme using recurrence method and quantum low density parity check codes},\ }\href {https://doi.org/10.1007/s11128-010-0190-1} {\bibfield  {journal} {\bibinfo  {journal} {Quantum Information Processing}\ }\textbf {\bibinfo {volume} {10}},\ \bibinfo {pages} {213} (\bibinfo {year} {2011})}\BibitemShut {NoStop}%
\bibitem [{\citenamefont {Roque}\ \emph {et~al.}(2024)\citenamefont {Roque}, \citenamefont {Cruz}, \citenamefont {Monteiro},\ and\ \citenamefont {Coutinho}}]{roque2024efficient}%
  \BibitemOpen
  \bibfield  {author} {\bibinfo {author} {\bibfnamefont {A.}~\bibnamefont {Roque}}, \bibinfo {author} {\bibfnamefont {D.}~\bibnamefont {Cruz}}, \bibinfo {author} {\bibfnamefont {F.~A.}\ \bibnamefont {Monteiro}},\ and\ \bibinfo {author} {\bibfnamefont {B.~C.}\ \bibnamefont {Coutinho}},\ }\bibfield  {title} {\bibinfo {title} {Efficient entanglement purification based on noise guessing decoding},\ }\href {https://doi.org/10.22331/q-2024-09-19-1476} {\bibfield  {journal} {\bibinfo  {journal} {{Quantum}}\ }\textbf {\bibinfo {volume} {8}},\ \bibinfo {pages} {1476} (\bibinfo {year} {2024})}\BibitemShut {NoStop}%
\bibitem [{\citenamefont {Pattison}\ \emph {et~al.}(2024)\citenamefont {Pattison}, \citenamefont {Baranes}, \citenamefont {Ataides}, \citenamefont {Lukin},\ and\ \citenamefont {Zhou}}]{pattison2024constoverheaddistillation}%
  \BibitemOpen
  \bibfield  {author} {\bibinfo {author} {\bibfnamefont {C.~A.}\ \bibnamefont {Pattison}}, \bibinfo {author} {\bibfnamefont {G.}~\bibnamefont {Baranes}}, \bibinfo {author} {\bibfnamefont {J.~P.~B.}\ \bibnamefont {Ataides}}, \bibinfo {author} {\bibfnamefont {M.~D.}\ \bibnamefont {Lukin}},\ and\ \bibinfo {author} {\bibfnamefont {H.}~\bibnamefont {Zhou}},\ }\href@noop {} {\bibinfo {title} {Fast quantum interconnects via constant-rate entanglement distillation}} (\bibinfo {year} {2024}),\ \Eprint {https://arxiv.org/abs/2408.15936} {arXiv:2408.15936 [quant-ph]} \BibitemShut {NoStop}%
\bibitem [{\citenamefont {Shi}\ \emph {et~al.}(2024)\citenamefont {Shi}, \citenamefont {Patil},\ and\ \citenamefont {Guha}}]{shi2024stabilizer}%
  \BibitemOpen
  \bibfield  {author} {\bibinfo {author} {\bibfnamefont {Y.}~\bibnamefont {Shi}}, \bibinfo {author} {\bibfnamefont {A.}~\bibnamefont {Patil}},\ and\ \bibinfo {author} {\bibfnamefont {S.}~\bibnamefont {Guha}},\ }\href@noop {} {\bibinfo {title} {Stabilizer entanglement distillation and efficient fault-tolerant encoder}} (\bibinfo {year} {2024}),\ \Eprint {https://arxiv.org/abs/2408.06299} {arXiv:2408.06299 [quant-ph]} \BibitemShut {NoStop}%
\bibitem [{\citenamefont {Ataides}\ \emph {et~al.}(2025)\citenamefont {Ataides}, \citenamefont {Zhou}, \citenamefont {Xu}, \citenamefont {Baranes}, \citenamefont {Li}, \citenamefont {Lukin},\ and\ \citenamefont {Jiang}}]{ataides2025constant}%
  \BibitemOpen
  \bibfield  {author} {\bibinfo {author} {\bibfnamefont {J.~P.~B.}\ \bibnamefont {Ataides}}, \bibinfo {author} {\bibfnamefont {H.}~\bibnamefont {Zhou}}, \bibinfo {author} {\bibfnamefont {Q.}~\bibnamefont {Xu}}, \bibinfo {author} {\bibfnamefont {G.}~\bibnamefont {Baranes}}, \bibinfo {author} {\bibfnamefont {B.}~\bibnamefont {Li}}, \bibinfo {author} {\bibfnamefont {M.~D.}\ \bibnamefont {Lukin}},\ and\ \bibinfo {author} {\bibfnamefont {L.}~\bibnamefont {Jiang}},\ }\href@noop {} {\bibinfo {title} {Constant-overhead fault-tolerant bell-pair distillation using high-rate codes}} (\bibinfo {year} {2025}),\ \Eprint {https://arxiv.org/abs/2502.09542} {arXiv:2502.09542 [quant-ph]} \BibitemShut {NoStop}%
\bibitem [{\citenamefont {Dehaene}\ \emph {et~al.}(2003)\citenamefont {Dehaene}, \citenamefont {Van~den Nest}, \citenamefont {De~Moor},\ and\ \citenamefont {Verstraete}}]{dehaene2003distillationpermutation}%
  \BibitemOpen
  \bibfield  {author} {\bibinfo {author} {\bibfnamefont {J.}~\bibnamefont {Dehaene}}, \bibinfo {author} {\bibfnamefont {M.}~\bibnamefont {Van~den Nest}}, \bibinfo {author} {\bibfnamefont {B.}~\bibnamefont {De~Moor}},\ and\ \bibinfo {author} {\bibfnamefont {F.}~\bibnamefont {Verstraete}},\ }\bibfield  {title} {\bibinfo {title} {{Local permutations of products of Bell states and entanglement distillation}},\ }\href {https://doi.org/10.1103/PhysRevA.67.022310} {\bibfield  {journal} {\bibinfo  {journal} {Physical Review A}\ }\textbf {\bibinfo {volume} {67}},\ \bibinfo {pages} {022310} (\bibinfo {year} {2003})}\BibitemShut {NoStop}%
\bibitem [{\citenamefont {Maneva}\ and\ \citenamefont {Smolin}(2002)}]{maneva2002purification}%
  \BibitemOpen
  \bibfield  {author} {\bibinfo {author} {\bibfnamefont {E.~N.}\ \bibnamefont {Maneva}}\ and\ \bibinfo {author} {\bibfnamefont {J.~A.}\ \bibnamefont {Smolin}},\ }\bibfield  {title} {\bibinfo {title} {Improved two-party and multi-party purification protocols},\ }in\ \href {https://doi.org/10.1090/conm/305/05220} {\emph {\bibinfo {booktitle} {Quantum computation and information ({W}ashington, {DC}, 2000)}}},\ \bibinfo {series} {Contemp. Math.}, Vol.\ \bibinfo {volume} {305}\ (\bibinfo  {publisher} {Amer. Math. Soc., Providence, RI},\ \bibinfo {year} {2002})\ pp.\ \bibinfo {pages} {203--212}\BibitemShut {NoStop}%
\bibitem [{\citenamefont {Bombin}\ and\ \citenamefont {Martin-Delgado}(2005)}]{BM05}%
  \BibitemOpen
  \bibfield  {author} {\bibinfo {author} {\bibfnamefont {H.}~\bibnamefont {Bombin}}\ and\ \bibinfo {author} {\bibfnamefont {M.~A.}\ \bibnamefont {Martin-Delgado}},\ }\bibfield  {title} {\bibinfo {title} {{Entanglement distillation protocols and number theory}},\ }\href {https://doi.org/10.1103/PhysRevA.72.032313} {\bibfield  {journal} {\bibinfo  {journal} {Physical Review A}\ }\textbf {\bibinfo {volume} {72}},\ \bibinfo {pages} {032313} (\bibinfo {year} {2005})}\BibitemShut {NoStop}%
\bibitem [{\citenamefont {Krastanov}\ \emph {et~al.}(2019)\citenamefont {Krastanov}, \citenamefont {Albert},\ and\ \citenamefont {Jiang}}]{KAJ19}%
  \BibitemOpen
  \bibfield  {author} {\bibinfo {author} {\bibfnamefont {S.}~\bibnamefont {Krastanov}}, \bibinfo {author} {\bibfnamefont {V.~V.}\ \bibnamefont {Albert}},\ and\ \bibinfo {author} {\bibfnamefont {L.}~\bibnamefont {Jiang}},\ }\bibfield  {title} {\bibinfo {title} {{Optimized Entanglement Purification}},\ }\href {https://doi.org/10.22331/q-2019-02-18-123} {\bibfield  {journal} {\bibinfo  {journal} {Quantum}\ }\textbf {\bibinfo {volume} {3}},\ \bibinfo {pages} {123} (\bibinfo {year} {2019})}\BibitemShut {NoStop}%
\bibitem [{\citenamefont {Jansen}\ \emph {et~al.}(2022)\citenamefont {Jansen}, \citenamefont {Goodenough}, \citenamefont {de~Bone}, \citenamefont {Gijswijt},\ and\ \citenamefont {Elkouss}}]{jansen2022enumeratingclifforddistillation}%
  \BibitemOpen
  \bibfield  {author} {\bibinfo {author} {\bibfnamefont {S.}~\bibnamefont {Jansen}}, \bibinfo {author} {\bibfnamefont {K.}~\bibnamefont {Goodenough}}, \bibinfo {author} {\bibfnamefont {S.}~\bibnamefont {de~Bone}}, \bibinfo {author} {\bibfnamefont {D.}~\bibnamefont {Gijswijt}},\ and\ \bibinfo {author} {\bibfnamefont {D.}~\bibnamefont {Elkouss}},\ }\bibfield  {title} {\bibinfo {title} {{Enumerating all bilocal Clifford distillation protocols through symmetry reduction}},\ }\href {https://doi.org/10.22331/q-2022-05-19-715} {\bibfield  {journal} {\bibinfo  {journal} {Quantum}\ }\textbf {\bibinfo {volume} {6}},\ \bibinfo {pages} {715} (\bibinfo {year} {2022})}\BibitemShut {NoStop}%
\bibitem [{\citenamefont {Addala}\ \emph {et~al.}(2023)\citenamefont {Addala}, \citenamefont {Ge},\ and\ \citenamefont {Krastanov}}]{addala2023optimizedpurification}%
  \BibitemOpen
  \bibfield  {author} {\bibinfo {author} {\bibfnamefont {V.~L.}\ \bibnamefont {Addala}}, \bibinfo {author} {\bibfnamefont {S.}~\bibnamefont {Ge}},\ and\ \bibinfo {author} {\bibfnamefont {S.}~\bibnamefont {Krastanov}},\ }\href@noop {} {\bibinfo {title} {{Faster-than-Clifford Simulations of Entanglement Purification Circuits and Their Full-stack Optimization}}} (\bibinfo {year} {2023}),\ \Eprint {https://arxiv.org/abs/2307.06354} {arXiv:2307.06354 [quant-ph]} \BibitemShut {NoStop}%
\bibitem [{\citenamefont {Goodenough}\ \emph {et~al.}(2024)\citenamefont {Goodenough}, \citenamefont {de~Bone}, \citenamefont {Addala}, \citenamefont {Krastanov}, \citenamefont {Jansen}, \citenamefont {Gijswijt},\ and\ \citenamefont {Elkouss}}]{goodenough2023ntokdistillation}%
  \BibitemOpen
  \bibfield  {author} {\bibinfo {author} {\bibfnamefont {K.}~\bibnamefont {Goodenough}}, \bibinfo {author} {\bibfnamefont {S.}~\bibnamefont {de~Bone}}, \bibinfo {author} {\bibfnamefont {V.}~\bibnamefont {Addala}}, \bibinfo {author} {\bibfnamefont {S.}~\bibnamefont {Krastanov}}, \bibinfo {author} {\bibfnamefont {S.}~\bibnamefont {Jansen}}, \bibinfo {author} {\bibfnamefont {D.}~\bibnamefont {Gijswijt}},\ and\ \bibinfo {author} {\bibfnamefont {D.}~\bibnamefont {Elkouss}},\ }\bibfield  {title} {\bibinfo {title} {{Near-Term $n$ to $k$ Distillation Protocols Using Graph Codes}},\ }\href {https://doi.org/10.1109/jsac.2024.3380094} {\bibfield  {journal} {\bibinfo  {journal} {IEEE Journal on Selected Areas in Communications}\ }\textbf {\bibinfo {volume} {42}},\ \bibinfo {pages} {1830–1849} (\bibinfo {year} {2024})}\BibitemShut {NoStop}%
\bibitem [{\citenamefont {Devetak}\ \emph {et~al.}(2004)\citenamefont {Devetak}, \citenamefont {Harrow},\ and\ \citenamefont {Winter}}]{devetak2004father}%
  \BibitemOpen
  \bibfield  {author} {\bibinfo {author} {\bibfnamefont {I.}~\bibnamefont {Devetak}}, \bibinfo {author} {\bibfnamefont {A.~W.}\ \bibnamefont {Harrow}},\ and\ \bibinfo {author} {\bibfnamefont {A.}~\bibnamefont {Winter}},\ }\bibfield  {title} {\bibinfo {title} {{A Family of Quantum Protocols}},\ }\href {https://doi.org/10.1103/PhysRevLett.93.230504} {\bibfield  {journal} {\bibinfo  {journal} {Physical Review Letters}\ }\textbf {\bibinfo {volume} {93}},\ \bibinfo {pages} {230504} (\bibinfo {year} {2004})}\BibitemShut {NoStop}%
\bibitem [{\citenamefont {Devetak}(2005)}]{devetak2005private}%
  \BibitemOpen
  \bibfield  {author} {\bibinfo {author} {\bibfnamefont {I.}~\bibnamefont {Devetak}},\ }\bibfield  {title} {\bibinfo {title} {The private classical capacity and quantum capacity of a quantum channel},\ }\href {https://doi.org/10.1109/TIT.2004.839515} {\bibfield  {journal} {\bibinfo  {journal} {IEEE Transactions on Information Theory}\ }\textbf {\bibinfo {volume} {51}},\ \bibinfo {pages} {44} (\bibinfo {year} {2005})}\BibitemShut {NoStop}%
\bibitem [{\citenamefont {Devetak}\ and\ \citenamefont {Winter}(2005)}]{devetak2005distillation}%
  \BibitemOpen
  \bibfield  {author} {\bibinfo {author} {\bibfnamefont {I.}~\bibnamefont {Devetak}}\ and\ \bibinfo {author} {\bibfnamefont {A.}~\bibnamefont {Winter}},\ }\bibfield  {title} {\bibinfo {title} {Distillation of secret key and entanglement from quantum states},\ }\href {https://doi.org/10.1098/rspa.2004.1372} {\bibfield  {journal} {\bibinfo  {journal} {Proceedings of the Royal Society A: Mathematical, Physical and Engineering Sciences}\ }\textbf {\bibinfo {volume} {461}},\ \bibinfo {pages} {207–235} (\bibinfo {year} {2005})}\BibitemShut {NoStop}%
\bibitem [{\citenamefont {Hayden}\ \emph {et~al.}(2008)\citenamefont {Hayden}, \citenamefont {Horodecki}, \citenamefont {Winter},\ and\ \citenamefont {Yard}}]{hayden2008decoupling}%
  \BibitemOpen
  \bibfield  {author} {\bibinfo {author} {\bibfnamefont {P.}~\bibnamefont {Hayden}}, \bibinfo {author} {\bibfnamefont {M.}~\bibnamefont {Horodecki}}, \bibinfo {author} {\bibfnamefont {A.}~\bibnamefont {Winter}},\ and\ \bibinfo {author} {\bibfnamefont {J.}~\bibnamefont {Yard}},\ }\bibfield  {title} {\bibinfo {title} {A decoupling approach to the quantum capacity},\ }\href {https://doi.org/10.1142/s1230161208000043} {\bibfield  {journal} {\bibinfo  {journal} {Open Systems \& Information Dynamics}\ }\textbf {\bibinfo {volume} {15}},\ \bibinfo {pages} {7–19} (\bibinfo {year} {2008})}\BibitemShut {NoStop}%
\bibitem [{\citenamefont {Abeyesinghe}\ \emph {et~al.}(2009)\citenamefont {Abeyesinghe}, \citenamefont {Devetak}, \citenamefont {Hayden},\ and\ \citenamefont {Winter}}]{abeyesinghe2009mother}%
  \BibitemOpen
  \bibfield  {author} {\bibinfo {author} {\bibfnamefont {A.}~\bibnamefont {Abeyesinghe}}, \bibinfo {author} {\bibfnamefont {I.}~\bibnamefont {Devetak}}, \bibinfo {author} {\bibfnamefont {P.}~\bibnamefont {Hayden}},\ and\ \bibinfo {author} {\bibfnamefont {A.}~\bibnamefont {Winter}},\ }\bibfield  {title} {\bibinfo {title} {The mother of all protocols: restructuring quantum information’s family tree},\ }\href {https://doi.org/10.1098/rspa.2009.0202} {\bibfield  {journal} {\bibinfo  {journal} {Proceedings of the Royal Society A: Mathematical, Physical and Engineering Sciences}\ }\textbf {\bibinfo {volume} {465}},\ \bibinfo {pages} {2537–2563} (\bibinfo {year} {2009})}\BibitemShut {NoStop}%
\bibitem [{\citenamefont {Buscemi}\ and\ \citenamefont {Datta}(2010)}]{buscemi2010distilling}%
  \BibitemOpen
  \bibfield  {author} {\bibinfo {author} {\bibfnamefont {F.}~\bibnamefont {Buscemi}}\ and\ \bibinfo {author} {\bibfnamefont {N.}~\bibnamefont {Datta}},\ }\bibfield  {title} {\bibinfo {title} {Distilling entanglement from arbitrary resources},\ }\href {https://doi.org/10.1063/1.3483717} {\bibfield  {journal} {\bibinfo  {journal} {Journal of Mathematical Physics}\ }\textbf {\bibinfo {volume} {51}},\ \bibinfo {pages} {102201} (\bibinfo {year} {2010})}\BibitemShut {NoStop}%
\bibitem [{\citenamefont {Leditzky}\ \emph {et~al.}(2018)\citenamefont {Leditzky}, \citenamefont {Datta},\ and\ \citenamefont {Smith}}]{leditzky2018useful}%
  \BibitemOpen
  \bibfield  {author} {\bibinfo {author} {\bibfnamefont {F.}~\bibnamefont {Leditzky}}, \bibinfo {author} {\bibfnamefont {N.}~\bibnamefont {Datta}},\ and\ \bibinfo {author} {\bibfnamefont {G.}~\bibnamefont {Smith}},\ }\bibfield  {title} {\bibinfo {title} {Useful states and entanglement distillation},\ }\href {https://doi.org/10.1109/tit.2017.2776907} {\bibfield  {journal} {\bibinfo  {journal} {IEEE Transactions on Information Theory}\ }\textbf {\bibinfo {volume} {64}},\ \bibinfo {pages} {4689–4708} (\bibinfo {year} {2018})}\BibitemShut {NoStop}%
\bibitem [{\citenamefont {Fang}\ \emph {et~al.}(2019)\citenamefont {Fang}, \citenamefont {Wang}, \citenamefont {Tomamichel},\ and\ \citenamefont {Duan}}]{FWTD19}%
  \BibitemOpen
  \bibfield  {author} {\bibinfo {author} {\bibfnamefont {K.}~\bibnamefont {Fang}}, \bibinfo {author} {\bibfnamefont {X.}~\bibnamefont {Wang}}, \bibinfo {author} {\bibfnamefont {M.}~\bibnamefont {Tomamichel}},\ and\ \bibinfo {author} {\bibfnamefont {R.}~\bibnamefont {Duan}},\ }\bibfield  {title} {\bibinfo {title} {Non-asymptotic entanglement distillation},\ }\href {https://doi.org/10.1109/tit.2019.2914688} {\bibfield  {journal} {\bibinfo  {journal} {IEEE Transactions on Information Theory}\ }\textbf {\bibinfo {volume} {65}},\ \bibinfo {pages} {6454–6465} (\bibinfo {year} {2019})}\BibitemShut {NoStop}%
\bibitem [{\citenamefont {Regula}\ \emph {et~al.}(2023)\citenamefont {Regula}, \citenamefont {Lami},\ and\ \citenamefont {Wilde}}]{regula2023probabilistic}%
  \BibitemOpen
  \bibfield  {author} {\bibinfo {author} {\bibfnamefont {B.}~\bibnamefont {Regula}}, \bibinfo {author} {\bibfnamefont {L.}~\bibnamefont {Lami}},\ and\ \bibinfo {author} {\bibfnamefont {M.~M.}\ \bibnamefont {Wilde}},\ }\bibfield  {title} {\bibinfo {title} {Overcoming entropic limitations on asymptotic state transformations through probabilistic protocols},\ }\href {https://doi.org/10.1103/PhysRevA.107.042401} {\bibfield  {journal} {\bibinfo  {journal} {Physical Review A}\ }\textbf {\bibinfo {volume} {107}},\ \bibinfo {pages} {042401} (\bibinfo {year} {2023})}\BibitemShut {NoStop}%
\bibitem [{\citenamefont {Siddhu}\ \emph {et~al.}(2024)\citenamefont {Siddhu}, \citenamefont {Abdelhadi}, \citenamefont {Jochym-O'Connor},\ and\ \citenamefont {Smolin}}]{siddhu2024entanglementsharing}%
  \BibitemOpen
  \bibfield  {author} {\bibinfo {author} {\bibfnamefont {V.}~\bibnamefont {Siddhu}}, \bibinfo {author} {\bibfnamefont {D.}~\bibnamefont {Abdelhadi}}, \bibinfo {author} {\bibfnamefont {T.}~\bibnamefont {Jochym-O'Connor}},\ and\ \bibinfo {author} {\bibfnamefont {J.}~\bibnamefont {Smolin}},\ }\bibfield  {title} {\bibinfo {title} {{Entanglement Sharing Across a Damping-Dephasing Channel}},\ }in\ \href {https://doi.org/10.1109/ISIT57864.2024.10619242} {\emph {\bibinfo {booktitle} {{2024 IEEE International Symposium on Information Theory (ISIT)}}}}\ (\bibinfo {year} {2024})\ pp.\ \bibinfo {pages} {1432--1437}\BibitemShut {NoStop}%
\bibitem [{\citenamefont {Abdelhadi}\ \emph {et~al.}(2024)\citenamefont {Abdelhadi}, \citenamefont {Jochym-O'Connor}, \citenamefont {Siddhu},\ and\ \citenamefont {Smolin}}]{abdelhadi2024adaptive}%
  \BibitemOpen
  \bibfield  {author} {\bibinfo {author} {\bibfnamefont {D.}~\bibnamefont {Abdelhadi}}, \bibinfo {author} {\bibfnamefont {T.}~\bibnamefont {Jochym-O'Connor}}, \bibinfo {author} {\bibfnamefont {V.}~\bibnamefont {Siddhu}},\ and\ \bibinfo {author} {\bibfnamefont {J.}~\bibnamefont {Smolin}},\ }\href@noop {} {\bibinfo {title} {{Adaptive Channel Reshaping for Improved Entanglement Distillation}}} (\bibinfo {year} {2024}),\ \Eprint {https://arxiv.org/abs/2410.22295} {arXiv:2410.22295 [quant-ph]} \BibitemShut {NoStop}%
\bibitem [{\citenamefont {Fowler}\ \emph {et~al.}(2010)\citenamefont {Fowler}, \citenamefont {Wang}, \citenamefont {Hill}, \citenamefont {Ladd}, \citenamefont {Van~Meter},\ and\ \citenamefont {Hollenberg}}]{fowler2010surface}%
  \BibitemOpen
  \bibfield  {author} {\bibinfo {author} {\bibfnamefont {A.~G.}\ \bibnamefont {Fowler}}, \bibinfo {author} {\bibfnamefont {D.~S.}\ \bibnamefont {Wang}}, \bibinfo {author} {\bibfnamefont {C.~D.}\ \bibnamefont {Hill}}, \bibinfo {author} {\bibfnamefont {T.~D.}\ \bibnamefont {Ladd}}, \bibinfo {author} {\bibfnamefont {R.}~\bibnamefont {Van~Meter}},\ and\ \bibinfo {author} {\bibfnamefont {L.~C.~L.}\ \bibnamefont {Hollenberg}},\ }\bibfield  {title} {\bibinfo {title} {{Surface Code Quantum Communication}},\ }\href {https://doi.org/10.1103/PhysRevLett.104.180503} {\bibfield  {journal} {\bibinfo  {journal} {Physical Review Letters}\ }\textbf {\bibinfo {volume} {104}},\ \bibinfo {pages} {180503} (\bibinfo {year} {2010})}\BibitemShut {NoStop}%
\bibitem [{\citenamefont {Ramette}\ \emph {et~al.}(2024)\citenamefont {Ramette}, \citenamefont {Sinclair}, \citenamefont {Breuckmann},\ and\ \citenamefont {Vuletić}}]{ramette2023fault}%
  \BibitemOpen
  \bibfield  {author} {\bibinfo {author} {\bibfnamefont {J.}~\bibnamefont {Ramette}}, \bibinfo {author} {\bibfnamefont {J.}~\bibnamefont {Sinclair}}, \bibinfo {author} {\bibfnamefont {N.~P.}\ \bibnamefont {Breuckmann}},\ and\ \bibinfo {author} {\bibfnamefont {V.}~\bibnamefont {Vuletić}},\ }\bibfield  {title} {\bibinfo {title} {Fault-tolerant connection of error-corrected qubits with noisy links},\ }\href {https://doi.org/10.1038/s41534-024-00855-4} {\bibfield  {journal} {\bibinfo  {journal} {npj Quantum Information}\ }\textbf {\bibinfo {volume} {10}},\ \bibinfo {pages} {58} (\bibinfo {year} {2024})}\BibitemShut {NoStop}%
\bibitem [{\citenamefont {Sinclair}\ \emph {et~al.}(2024)\citenamefont {Sinclair}, \citenamefont {Ramette}, \citenamefont {Grinkemeyer}, \citenamefont {Bluvstein}, \citenamefont {Lukin},\ and\ \citenamefont {Vuletić}}]{sinclair2024fault}%
  \BibitemOpen
  \bibfield  {author} {\bibinfo {author} {\bibfnamefont {J.}~\bibnamefont {Sinclair}}, \bibinfo {author} {\bibfnamefont {J.}~\bibnamefont {Ramette}}, \bibinfo {author} {\bibfnamefont {B.}~\bibnamefont {Grinkemeyer}}, \bibinfo {author} {\bibfnamefont {D.}~\bibnamefont {Bluvstein}}, \bibinfo {author} {\bibfnamefont {M.}~\bibnamefont {Lukin}},\ and\ \bibinfo {author} {\bibfnamefont {V.}~\bibnamefont {Vuletić}},\ }\href@noop {} {\bibinfo {title} {Fault-tolerant optical interconnects for neutral-atom arrays}} (\bibinfo {year} {2024}),\ \Eprint {https://arxiv.org/abs/2408.08955} {arXiv:2408.08955 [quant-ph]} \BibitemShut {NoStop}%
\bibitem [{\citenamefont {Shannon}(1948)}]{shannon1948mathematical}%
  \BibitemOpen
  \bibfield  {author} {\bibinfo {author} {\bibfnamefont {C.~E.}\ \bibnamefont {Shannon}},\ }\bibfield  {title} {\bibinfo {title} {A {Mathematical} {Theory} of {Communication}},\ }\href {https://doi.org/10.1002/j.1538-7305.1948.tb01338.x} {\bibfield  {journal} {\bibinfo  {journal} {Bell System Technical Journal}\ }\textbf {\bibinfo {volume} {27}},\ \bibinfo {pages} {379} (\bibinfo {year} {1948})}\BibitemShut {NoStop}%
\bibitem [{\citenamefont {Berlekamp}\ \emph {et~al.}(1978)\citenamefont {Berlekamp}, \citenamefont {McEliece},\ and\ \citenamefont {van Tilborg}}]{berlekamp1978intractability}%
  \BibitemOpen
  \bibfield  {author} {\bibinfo {author} {\bibfnamefont {E.}~\bibnamefont {Berlekamp}}, \bibinfo {author} {\bibfnamefont {R.}~\bibnamefont {McEliece}},\ and\ \bibinfo {author} {\bibfnamefont {H.}~\bibnamefont {van Tilborg}},\ }\bibfield  {title} {\bibinfo {title} {On the inherent intractability of certain coding problems},\ }\href {https://doi.org/10.1109/TIT.1978.1055873} {\bibfield  {journal} {\bibinfo  {journal} {IEEE Transactions on Information Theory}\ }\textbf {\bibinfo {volume} {24}},\ \bibinfo {pages} {384} (\bibinfo {year} {1978})}\BibitemShut {NoStop}%
\bibitem [{\citenamefont {Vardy}(1997)}]{vardy1997intractability}%
  \BibitemOpen
  \bibfield  {author} {\bibinfo {author} {\bibfnamefont {A.}~\bibnamefont {Vardy}},\ }\bibfield  {title} {\bibinfo {title} {The intractability of computing the minimum distance of a code},\ }\href {https://doi.org/10.1109/18.641542} {\bibfield  {journal} {\bibinfo  {journal} {IEEE Transactions on Information Theory}\ }\textbf {\bibinfo {volume} {43}},\ \bibinfo {pages} {1757} (\bibinfo {year} {1997})}\BibitemShut {NoStop}%
\bibitem [{\citenamefont {Hsieh}\ and\ \citenamefont {Le~Gall}(2011)}]{hsieh2011NPhardquantumdecoding}%
  \BibitemOpen
  \bibfield  {author} {\bibinfo {author} {\bibfnamefont {M.-H.}\ \bibnamefont {Hsieh}}\ and\ \bibinfo {author} {\bibfnamefont {F.}~\bibnamefont {Le~Gall}},\ }\bibfield  {title} {\bibinfo {title} {{NP-hardness of decoding quantum error-correction codes}},\ }\href {https://doi.org/10.1103/PhysRevA.83.052331} {\bibfield  {journal} {\bibinfo  {journal} {Physical Review A}\ }\textbf {\bibinfo {volume} {83}},\ \bibinfo {pages} {052331} (\bibinfo {year} {2011})}\BibitemShut {NoStop}%
\bibitem [{\citenamefont {Kuo}\ and\ \citenamefont {Lu}(2012)}]{kuo2012hardnessquantumdecoding}%
  \BibitemOpen
  \bibfield  {author} {\bibinfo {author} {\bibfnamefont {K.-Y.}\ \bibnamefont {Kuo}}\ and\ \bibinfo {author} {\bibfnamefont {C.-C.}\ \bibnamefont {Lu}},\ }\bibfield  {title} {\bibinfo {title} {On the hardness of decoding quantum stabilizer codes under the depolarizing channel},\ }in\ \href {https://ieeexplore.ieee.org/document/6400919} {\emph {\bibinfo {booktitle} {{2012 International Symposium on Information Theory and its Applications}}}}\ (\bibinfo {year} {2012})\ pp.\ \bibinfo {pages} {208--211}\BibitemShut {NoStop}%
\bibitem [{\citenamefont {Iyer}\ and\ \citenamefont {Poulin}(2015)}]{iyer2015hardnessquantumdecoding}%
  \BibitemOpen
  \bibfield  {author} {\bibinfo {author} {\bibfnamefont {P.}~\bibnamefont {Iyer}}\ and\ \bibinfo {author} {\bibfnamefont {D.}~\bibnamefont {Poulin}},\ }\bibfield  {title} {\bibinfo {title} {{Hardness of Decoding Quantum Stabilizer Codes}},\ }\href {https://doi.org/10.1109/TIT.2015.2422294} {\bibfield  {journal} {\bibinfo  {journal} {IEEE Transactions on Information Theory}\ }\textbf {\bibinfo {volume} {61}},\ \bibinfo {pages} {5209} (\bibinfo {year} {2015})}\BibitemShut {NoStop}%
\bibitem [{\citenamefont {McEliece}(1978)}]{mceliece1978public}%
  \BibitemOpen
  \bibfield  {author} {\bibinfo {author} {\bibfnamefont {R.~J.}\ \bibnamefont {McEliece}},\ }\href {https://ipnpr.jpl.nasa.gov/progress_report/42-44/44N.PDF} {\emph {\bibinfo {title} {{A Public-Key Cryptosystem Based on Algebraic Coding Theory}}}},\ \bibinfo {type} {Tech. Rep.}\ \bibinfo {number} {DSN Progress Report 42--44}\ (\bibinfo  {institution} {Jet Propulsion Laboratory, California Institute of Technology},\ \bibinfo {year} {1978})\BibitemShut {NoStop}%
\bibitem [{\citenamefont {Alekhnovich}(2003)}]{alekhnovich2003more}%
  \BibitemOpen
  \bibfield  {author} {\bibinfo {author} {\bibfnamefont {M.}~\bibnamefont {Alekhnovich}},\ }\bibfield  {title} {\bibinfo {title} {More on average case vs approximation complexity},\ }in\ \href {https://doi.org/10.1109/SFCS.2003.1238204} {\emph {\bibinfo {booktitle} {44th Annual IEEE Symposium on Foundations of Computer Science, 2003. Proceedings.}}}\ (\bibinfo {year} {2003})\ pp.\ \bibinfo {pages} {298--307}\BibitemShut {NoStop}%
\bibitem [{\citenamefont {Cleve}\ \emph {et~al.}(2016)\citenamefont {Cleve}, \citenamefont {Leung}, \citenamefont {Liu},\ and\ \citenamefont {Wang}}]{cleve2016near}%
  \BibitemOpen
  \bibfield  {author} {\bibinfo {author} {\bibfnamefont {R.}~\bibnamefont {Cleve}}, \bibinfo {author} {\bibfnamefont {D.}~\bibnamefont {Leung}}, \bibinfo {author} {\bibfnamefont {L.}~\bibnamefont {Liu}},\ and\ \bibinfo {author} {\bibfnamefont {C.}~\bibnamefont {Wang}},\ }\bibfield  {title} {\bibinfo {title} {Near-linear constructions of exact unitary 2-designs},\ }\href {https://doi.org/10.5555/3179473.3179474} {\bibfield  {journal} {\bibinfo  {journal} {Quantum Information and Computation}\ }\textbf {\bibinfo {volume} {16}},\ \bibinfo {pages} {721–756} (\bibinfo {year} {2016})}\BibitemShut {NoStop}%
\bibitem [{Note1()}]{Note1}%
  \BibitemOpen
  \bibinfo {note} {The definition of single-pair fidelity comes from the notion of logical error rate in quantum error correction, where it makes sense to define a logical error rate per logical qubit of an $[[n,k,d]]$ quantum error-correction code. Specifically, if $p_L$ is the overall logical error rate of the $[[n,k,d]]$ code, then the per-qubit logical error rate is formulated as $1-(1-p_L)^{1/k}$ in both theoretical and experimental works~\cite {roffe2023biastailoredquantum,berent2023softwaretools,wang2025demonstration}.}\BibitemShut {Stop}%
\bibitem [{\citenamefont {Rains}(1999)}]{rains1999bound}%
  \BibitemOpen
  \bibfield  {author} {\bibinfo {author} {\bibfnamefont {E.~M.}\ \bibnamefont {Rains}},\ }\bibfield  {title} {\bibinfo {title} {Bound on distillable entanglement},\ }\href {https://doi.org/10.1103/PhysRevA.60.179} {\bibfield  {journal} {\bibinfo  {journal} {Physical Review A}\ }\textbf {\bibinfo {volume} {60}},\ \bibinfo {pages} {179} (\bibinfo {year} {1999})}\BibitemShut {NoStop}%
\bibitem [{\citenamefont {{Rains}}(2001)}]{Rains01}%
  \BibitemOpen
  \bibfield  {author} {\bibinfo {author} {\bibfnamefont {E.~M.}\ \bibnamefont {{Rains}}},\ }\bibfield  {title} {\bibinfo {title} {{A semidefinite program for distillable entanglement}},\ }\href {https://doi.org/10.1109/18.959270} {\bibfield  {journal} {\bibinfo  {journal} {IEEE Transactions on Information Theory}\ }\textbf {\bibinfo {volume} {47}},\ \bibinfo {pages} {2921} (\bibinfo {year} {2001})}\BibitemShut {NoStop}%
\bibitem [{\citenamefont {Bombin}\ and\ \citenamefont {Martin-Delgado}(2006)}]{bombin2006topological}%
  \BibitemOpen
  \bibfield  {author} {\bibinfo {author} {\bibfnamefont {H.}~\bibnamefont {Bombin}}\ and\ \bibinfo {author} {\bibfnamefont {M.~A.}\ \bibnamefont {Martin-Delgado}},\ }\bibfield  {title} {\bibinfo {title} {{Topological Quantum Distillation}},\ }\href {https://doi.org/10.1103/PhysRevLett.97.180501} {\bibfield  {journal} {\bibinfo  {journal} {Physical Review Letters}\ }\textbf {\bibinfo {volume} {97}},\ \bibinfo {pages} {180501} (\bibinfo {year} {2006})}\BibitemShut {NoStop}%
\bibitem [{Note2()}]{Note2}%
  \BibitemOpen
  \bibinfo {note} {This encoding process typically proceeds through state injection, where a physical state is incrementally grown into a logical state through repeated application of Clifford operations. While this injection process is not fault-tolerant and introduces additional errors at a rate $\varepsilon _I$, these errors combine with the channel noise $\varepsilon $ to produce a new effective channel with error rate $\sim \varepsilon + \varepsilon _I$. Importantly, the injection error rate $\varepsilon _I$ typically scales with the local gate noise~\cite {zhang2024facilitating,li2015magic,fowler2009high,raussendorf2006fault,horsman2012surface}, which is generally much smaller than the channel noise $\varepsilon $ (typical numbers are $10^{-3} \leq \varepsilon _I \leq 10^{-2}$, while $\varepsilon \sim 10^{-1}$).}\BibitemShut {Stop}%
\bibitem [{Note3()}]{Note3}%
  \BibitemOpen
  \bibinfo {note} {By Pauli twirling --- applying the same random $n$-qubit Pauli on both halves, any channel is equivalent to a mixture of Pauli errors $P_\ell $ with probabilities $q_\ell $~\cite {BDSW96}.}\BibitemShut {Stop}%
\bibitem [{Note4()}]{Note4}%
  \BibitemOpen
  \bibinfo {note} {This lookup-based correction is likely optimal for random Clifford codes, as decoding random codes is widely believed to be computationally intractable~\cite {berlekamp1978intractability,mceliece1978public,alekhnovich2003more}.}\BibitemShut {Stop}%
\bibitem [{\citenamefont {Van~Meter}\ \emph {et~al.}(2009)\citenamefont {Van~Meter}, \citenamefont {Ladd}, \citenamefont {Munro},\ and\ \citenamefont {Nemoto}}]{VLMN09}%
  \BibitemOpen
  \bibfield  {author} {\bibinfo {author} {\bibfnamefont {R.}~\bibnamefont {Van~Meter}}, \bibinfo {author} {\bibfnamefont {T.}~\bibnamefont {Ladd}}, \bibinfo {author} {\bibfnamefont {W.}~\bibnamefont {Munro}},\ and\ \bibinfo {author} {\bibfnamefont {K.}~\bibnamefont {Nemoto}},\ }\bibfield  {title} {\bibinfo {title} {System design for a long-line quantum repeater},\ }\href {https://doi.org/10.1109/tnet.2008.927260} {\bibfield  {journal} {\bibinfo  {journal} {IEEE/ACM Transactions on Networking}\ }\textbf {\bibinfo {volume} {17}},\ \bibinfo {pages} {1002–1013} (\bibinfo {year} {2009})}\BibitemShut {NoStop}%
\bibitem [{\citenamefont {Aparicio}\ \emph {et~al.}(2011)\citenamefont {Aparicio}, \citenamefont {Van~Meter},\ and\ \citenamefont {Esaki}}]{AME11}%
  \BibitemOpen
  \bibfield  {author} {\bibinfo {author} {\bibfnamefont {L.}~\bibnamefont {Aparicio}}, \bibinfo {author} {\bibfnamefont {R.}~\bibnamefont {Van~Meter}},\ and\ \bibinfo {author} {\bibfnamefont {H.}~\bibnamefont {Esaki}},\ }\bibfield  {title} {\bibinfo {title} {{Protocol Design for Quantum Repeater Networks}},\ }in\ \href {https://doi.org/10.1145/2089016.2089029} {\emph {\bibinfo {booktitle} {Proceedings of the 7th Asian Internet Engineering Conference}}},\ \bibinfo {series and number} {AINTEC '11}\ (\bibinfo  {publisher} {Association for Computing Machinery},\ \bibinfo {address} {New York, NY, USA},\ \bibinfo {year} {2011})\ pp.\ \bibinfo {pages} {73--80}\BibitemShut {NoStop}%
\bibitem [{\citenamefont {Awschalom}\ \emph {et~al.}(2021)\citenamefont {Awschalom}, \citenamefont {Berggren}, \citenamefont {Bernien}, \citenamefont {Bhave}, \citenamefont {Carr}, \citenamefont {Davids}, \citenamefont {Economou}, \citenamefont {Englund}, \citenamefont {Faraon}, \citenamefont {Fejer}, \citenamefont {Guha}, \citenamefont {Gustafsson}, \citenamefont {Hu}, \citenamefont {Jiang}, \citenamefont {Kim}, \citenamefont {Korzh}, \citenamefont {Kumar}, \citenamefont {Kwiat}, \citenamefont {Lon\ifmmode~\check{c}\else \v{c}\fi{}ar}, \citenamefont {Lukin}, \citenamefont {Miller}, \citenamefont {Monroe}, \citenamefont {Nam}, \citenamefont {Narang}, \citenamefont {Orcutt}, \citenamefont {Raymer}, \citenamefont {Safavi-Naeini}, \citenamefont {Spiropulu}, \citenamefont {Srinivasan}, \citenamefont {Sun}, \citenamefont {Vu\ifmmode \check{c}\else \v{c}\fi{}kovi\ifmmode~\acute{c}\else \'{c}\fi{}}, \citenamefont {Waks}, \citenamefont {Walsworth}, \citenamefont {Weiner},\ and\ \citenamefont
  {Zhang}}]{awschalom2021interconnects}%
  \BibitemOpen
  \bibfield  {author} {\bibinfo {author} {\bibfnamefont {D.}~\bibnamefont {Awschalom}}, \bibinfo {author} {\bibfnamefont {K.~K.}\ \bibnamefont {Berggren}}, \bibinfo {author} {\bibfnamefont {H.}~\bibnamefont {Bernien}}, \bibinfo {author} {\bibfnamefont {S.}~\bibnamefont {Bhave}}, \bibinfo {author} {\bibfnamefont {L.~D.}\ \bibnamefont {Carr}}, \bibinfo {author} {\bibfnamefont {P.}~\bibnamefont {Davids}}, \bibinfo {author} {\bibfnamefont {S.~E.}\ \bibnamefont {Economou}}, \bibinfo {author} {\bibfnamefont {D.}~\bibnamefont {Englund}}, \bibinfo {author} {\bibfnamefont {A.}~\bibnamefont {Faraon}}, \bibinfo {author} {\bibfnamefont {M.}~\bibnamefont {Fejer}}, \bibinfo {author} {\bibfnamefont {S.}~\bibnamefont {Guha}}, \bibinfo {author} {\bibfnamefont {M.~V.}\ \bibnamefont {Gustafsson}}, \bibinfo {author} {\bibfnamefont {E.}~\bibnamefont {Hu}}, \bibinfo {author} {\bibfnamefont {L.}~\bibnamefont {Jiang}}, \bibinfo {author} {\bibfnamefont {J.}~\bibnamefont {Kim}}, \bibinfo {author} {\bibfnamefont {B.}~\bibnamefont
  {Korzh}}, \bibinfo {author} {\bibfnamefont {P.}~\bibnamefont {Kumar}}, \bibinfo {author} {\bibfnamefont {P.~G.}\ \bibnamefont {Kwiat}}, \bibinfo {author} {\bibfnamefont {M.}~\bibnamefont {Lon\ifmmode~\check{c}\else \v{c}\fi{}ar}}, \bibinfo {author} {\bibfnamefont {M.~D.}\ \bibnamefont {Lukin}}, \bibinfo {author} {\bibfnamefont {D.~A.}\ \bibnamefont {Miller}}, \bibinfo {author} {\bibfnamefont {C.}~\bibnamefont {Monroe}}, \bibinfo {author} {\bibfnamefont {S.~W.}\ \bibnamefont {Nam}}, \bibinfo {author} {\bibfnamefont {P.}~\bibnamefont {Narang}}, \bibinfo {author} {\bibfnamefont {J.~S.}\ \bibnamefont {Orcutt}}, \bibinfo {author} {\bibfnamefont {M.~G.}\ \bibnamefont {Raymer}}, \bibinfo {author} {\bibfnamefont {A.~H.}\ \bibnamefont {Safavi-Naeini}}, \bibinfo {author} {\bibfnamefont {M.}~\bibnamefont {Spiropulu}}, \bibinfo {author} {\bibfnamefont {K.}~\bibnamefont {Srinivasan}}, \bibinfo {author} {\bibfnamefont {S.}~\bibnamefont {Sun}}, \bibinfo {author} {\bibfnamefont {J.}~\bibnamefont {Vu\ifmmode \check{c}\else
  \v{c}\fi{}kovi\ifmmode~\acute{c}\else \'{c}\fi{}}}, \bibinfo {author} {\bibfnamefont {E.}~\bibnamefont {Waks}}, \bibinfo {author} {\bibfnamefont {R.}~\bibnamefont {Walsworth}}, \bibinfo {author} {\bibfnamefont {A.~M.}\ \bibnamefont {Weiner}},\ and\ \bibinfo {author} {\bibfnamefont {Z.}~\bibnamefont {Zhang}},\ }\bibfield  {title} {\bibinfo {title} {{Development of Quantum Interconnects (QuICs) for Next-Generation Information Technologies}},\ }\href {https://doi.org/10.1103/PRXQuantum.2.017002} {\bibfield  {journal} {\bibinfo  {journal} {PRX Quantum}\ }\textbf {\bibinfo {volume} {2}},\ \bibinfo {pages} {017002} (\bibinfo {year} {2021})}\BibitemShut {NoStop}%
\bibitem [{\citenamefont {Roberts}\ and\ \citenamefont {Yoshida}(2017)}]{roberts2017chaos}%
  \BibitemOpen
  \bibfield  {author} {\bibinfo {author} {\bibfnamefont {D.~A.}\ \bibnamefont {Roberts}}\ and\ \bibinfo {author} {\bibfnamefont {B.}~\bibnamefont {Yoshida}},\ }\bibfield  {title} {\bibinfo {title} {Chaos and complexity by design},\ }\href {https://doi.org/10.1007/JHEP04(2017)121} {\bibfield  {journal} {\bibinfo  {journal} {Journal of High Energy Physics}\ }\textbf {\bibinfo {volume} {2017}},\ \bibinfo {pages} {121} (\bibinfo {year} {2017})}\BibitemShut {NoStop}%
\bibitem [{\citenamefont {Cotler}\ \emph {et~al.}(2017)\citenamefont {Cotler}, \citenamefont {Hunter-Jones}, \citenamefont {Liu},\ and\ \citenamefont {Yoshida}}]{cotler2017chaos}%
  \BibitemOpen
  \bibfield  {author} {\bibinfo {author} {\bibfnamefont {J.}~\bibnamefont {Cotler}}, \bibinfo {author} {\bibfnamefont {N.}~\bibnamefont {Hunter-Jones}}, \bibinfo {author} {\bibfnamefont {J.}~\bibnamefont {Liu}},\ and\ \bibinfo {author} {\bibfnamefont {B.}~\bibnamefont {Yoshida}},\ }\bibfield  {title} {\bibinfo {title} {Chaos, complexity, and random matrices},\ }\href {https://doi.org/10.1007/JHEP11(2017)048} {\bibfield  {journal} {\bibinfo  {journal} {Journal of High Energy Physics}\ }\textbf {\bibinfo {volume} {2017}},\ \bibinfo {pages} {48} (\bibinfo {year} {2017})}\BibitemShut {NoStop}%
\bibitem [{\citenamefont {Hosur}\ \emph {et~al.}(2016)\citenamefont {Hosur}, \citenamefont {Qi}, \citenamefont {Roberts},\ and\ \citenamefont {Yoshida}}]{hosur2016chaos}%
  \BibitemOpen
  \bibfield  {author} {\bibinfo {author} {\bibfnamefont {P.}~\bibnamefont {Hosur}}, \bibinfo {author} {\bibfnamefont {X.-L.}\ \bibnamefont {Qi}}, \bibinfo {author} {\bibfnamefont {D.~A.}\ \bibnamefont {Roberts}},\ and\ \bibinfo {author} {\bibfnamefont {B.}~\bibnamefont {Yoshida}},\ }\bibfield  {title} {\bibinfo {title} {Chaos in quantum channels},\ }\href {https://doi.org/10.1007/JHEP02(2016)004} {\bibfield  {journal} {\bibinfo  {journal} {Journal of High Energy Physics}\ }\textbf {\bibinfo {volume} {2016}},\ \bibinfo {pages} {4} (\bibinfo {year} {2016})}\BibitemShut {NoStop}%
\bibitem [{\citenamefont {Landsman}\ \emph {et~al.}(2019)\citenamefont {Landsman}, \citenamefont {Figgatt}, \citenamefont {Schuster}, \citenamefont {Linke}, \citenamefont {Yoshida}, \citenamefont {Yao},\ and\ \citenamefont {Monroe}}]{landsman2019verified}%
  \BibitemOpen
  \bibfield  {author} {\bibinfo {author} {\bibfnamefont {K.~A.}\ \bibnamefont {Landsman}}, \bibinfo {author} {\bibfnamefont {C.}~\bibnamefont {Figgatt}}, \bibinfo {author} {\bibfnamefont {T.}~\bibnamefont {Schuster}}, \bibinfo {author} {\bibfnamefont {N.~M.}\ \bibnamefont {Linke}}, \bibinfo {author} {\bibfnamefont {B.}~\bibnamefont {Yoshida}}, \bibinfo {author} {\bibfnamefont {N.~Y.}\ \bibnamefont {Yao}},\ and\ \bibinfo {author} {\bibfnamefont {C.}~\bibnamefont {Monroe}},\ }\bibfield  {title} {\bibinfo {title} {Verified quantum information scrambling},\ }\href {https://doi.org/10.1038/s41586-019-0952-6} {\bibfield  {journal} {\bibinfo  {journal} {Nature}\ }\textbf {\bibinfo {volume} {567}},\ \bibinfo {pages} {61} (\bibinfo {year} {2019})}\BibitemShut {NoStop}%
\bibitem [{\citenamefont {Bohnet}\ \emph {et~al.}(2016)\citenamefont {Bohnet}, \citenamefont {Sawyer}, \citenamefont {Britton}, \citenamefont {Wall}, \citenamefont {Rey}, \citenamefont {Foss-Feig},\ and\ \citenamefont {Bollinger}}]{bohnet2016quantum}%
  \BibitemOpen
  \bibfield  {author} {\bibinfo {author} {\bibfnamefont {J.~G.}\ \bibnamefont {Bohnet}}, \bibinfo {author} {\bibfnamefont {B.~C.}\ \bibnamefont {Sawyer}}, \bibinfo {author} {\bibfnamefont {J.~W.}\ \bibnamefont {Britton}}, \bibinfo {author} {\bibfnamefont {M.~L.}\ \bibnamefont {Wall}}, \bibinfo {author} {\bibfnamefont {A.~M.}\ \bibnamefont {Rey}}, \bibinfo {author} {\bibfnamefont {M.}~\bibnamefont {Foss-Feig}},\ and\ \bibinfo {author} {\bibfnamefont {J.~J.}\ \bibnamefont {Bollinger}},\ }\bibfield  {title} {\bibinfo {title} {Quantum spin dynamics and entanglement generation with hundreds of trapped ions},\ }\href {https://doi.org/10.1126/science.aad9958} {\bibfield  {journal} {\bibinfo  {journal} {Science}\ }\textbf {\bibinfo {volume} {352}},\ \bibinfo {pages} {1297} (\bibinfo {year} {2016})}\BibitemShut {NoStop}%
\bibitem [{\citenamefont {Bernien}\ \emph {et~al.}(2017)\citenamefont {Bernien}, \citenamefont {Schwartz}, \citenamefont {Keesling}, \citenamefont {Levine}, \citenamefont {Omran}, \citenamefont {Pichler}, \citenamefont {Choi}, \citenamefont {Zibrov}, \citenamefont {Endres}, \citenamefont {Greiner}, \citenamefont {Vuletić},\ and\ \citenamefont {Lukin}}]{bernien2017probing}%
  \BibitemOpen
  \bibfield  {author} {\bibinfo {author} {\bibfnamefont {H.}~\bibnamefont {Bernien}}, \bibinfo {author} {\bibfnamefont {S.}~\bibnamefont {Schwartz}}, \bibinfo {author} {\bibfnamefont {A.}~\bibnamefont {Keesling}}, \bibinfo {author} {\bibfnamefont {H.}~\bibnamefont {Levine}}, \bibinfo {author} {\bibfnamefont {A.}~\bibnamefont {Omran}}, \bibinfo {author} {\bibfnamefont {H.}~\bibnamefont {Pichler}}, \bibinfo {author} {\bibfnamefont {S.}~\bibnamefont {Choi}}, \bibinfo {author} {\bibfnamefont {A.~S.}\ \bibnamefont {Zibrov}}, \bibinfo {author} {\bibfnamefont {M.}~\bibnamefont {Endres}}, \bibinfo {author} {\bibfnamefont {M.}~\bibnamefont {Greiner}}, \bibinfo {author} {\bibfnamefont {V.}~\bibnamefont {Vuletić}},\ and\ \bibinfo {author} {\bibfnamefont {M.~D.}\ \bibnamefont {Lukin}},\ }\bibfield  {title} {\bibinfo {title} {Probing many-body dynamics on a 51-atom quantum simulator},\ }\href {https://doi.org/10.1038/nature24622} {\bibfield  {journal} {\bibinfo  {journal} {Nature}\ }\textbf {\bibinfo {volume} {551}},\
  \bibinfo {pages} {579} (\bibinfo {year} {2017})}\BibitemShut {NoStop}%
\bibitem [{\citenamefont {{Andi Gu}}(2025)}]{data}%
  \BibitemOpen
  \bibfield  {author} {\bibinfo {author} {\bibnamefont {{Andi Gu}}},\ }\href {https://github.com/andigu/bilocal-distillation} {\bibinfo {title} {{Code repository for generating figures in ``Constant Overhead Entanglement Distillation via Scrambling''}}} (\bibinfo {year} {{2025}})\BibitemShut {NoStop}%
\bibitem [{\citenamefont {Roffe}\ \emph {et~al.}(2023)\citenamefont {Roffe}, \citenamefont {Cohen}, \citenamefont {Quintavalle}, \citenamefont {Chandra},\ and\ \citenamefont {Campbell}}]{roffe2023biastailoredquantum}%
  \BibitemOpen
  \bibfield  {author} {\bibinfo {author} {\bibfnamefont {J.}~\bibnamefont {Roffe}}, \bibinfo {author} {\bibfnamefont {L.~Z.}\ \bibnamefont {Cohen}}, \bibinfo {author} {\bibfnamefont {A.~O.}\ \bibnamefont {Quintavalle}}, \bibinfo {author} {\bibfnamefont {D.}~\bibnamefont {Chandra}},\ and\ \bibinfo {author} {\bibfnamefont {E.~T.}\ \bibnamefont {Campbell}},\ }\bibfield  {title} {\bibinfo {title} {Bias-tailored quantum {LDPC} codes},\ }\href {https://doi.org/10.22331/q-2023-05-15-1005} {\bibfield  {journal} {\bibinfo  {journal} {{Quantum}}\ }\textbf {\bibinfo {volume} {7}},\ \bibinfo {pages} {1005} (\bibinfo {year} {2023})}\BibitemShut {NoStop}%
\bibitem [{\citenamefont {Berent}\ \emph {et~al.}(2023)\citenamefont {Berent}, \citenamefont {Burgholzer},\ and\ \citenamefont {Wille}}]{berent2023softwaretools}%
  \BibitemOpen
  \bibfield  {author} {\bibinfo {author} {\bibfnamefont {L.}~\bibnamefont {Berent}}, \bibinfo {author} {\bibfnamefont {L.}~\bibnamefont {Burgholzer}},\ and\ \bibinfo {author} {\bibfnamefont {R.}~\bibnamefont {Wille}},\ }\bibfield  {title} {\bibinfo {title} {{Software Tools for Decoding Quantum Low-Density Parity-Check Codes}},\ }in\ \href {https://doi.org/10.1145/3566097.3567934} {\emph {\bibinfo {booktitle} {{Proceedings of the 28th Asia and South Pacific Design Automation Conference}}}},\ \bibinfo {series and number} {{ASPDAC '23}}\ (\bibinfo  {publisher} {Association for Computing Machinery},\ \bibinfo {address} {New York, NY, USA},\ \bibinfo {year} {2023})\ pp.\ \bibinfo {pages} {709--714}\BibitemShut {NoStop}%
\bibitem [{\citenamefont {Wang}\ \emph {et~al.}(2025)\citenamefont {Wang}, \citenamefont {Lu}, \citenamefont {Zhang}, \citenamefont {Liu}, \citenamefont {Chen}, \citenamefont {Wang}, \citenamefont {Wu}, \citenamefont {Xu}, \citenamefont {Zhu}, \citenamefont {Jin}, \citenamefont {Gao}, \citenamefont {Tan}, \citenamefont {Cui}, \citenamefont {Wang}, \citenamefont {Zou}, \citenamefont {Zhang}, \citenamefont {Li}, \citenamefont {Shen}, \citenamefont {Zhong}, \citenamefont {Bao}, \citenamefont {Zhu}, \citenamefont {Han}, \citenamefont {He}, \citenamefont {Shen}, \citenamefont {Wang}, \citenamefont {Yang}, \citenamefont {Song}, \citenamefont {Deng}, \citenamefont {Dong}, \citenamefont {Sun}, \citenamefont {Li}, \citenamefont {Ye}, \citenamefont {Jiang}, \citenamefont {Ma}, \citenamefont {Shen}, \citenamefont {Zhang}, \citenamefont {Li}, \citenamefont {Guo}, \citenamefont {Wang}, \citenamefont {Song}, \citenamefont {Wang},\ and\ \citenamefont {Deng}}]{wang2025demonstration}%
  \BibitemOpen
  \bibfield  {author} {\bibinfo {author} {\bibfnamefont {K.}~\bibnamefont {Wang}}, \bibinfo {author} {\bibfnamefont {Z.}~\bibnamefont {Lu}}, \bibinfo {author} {\bibfnamefont {C.}~\bibnamefont {Zhang}}, \bibinfo {author} {\bibfnamefont {G.}~\bibnamefont {Liu}}, \bibinfo {author} {\bibfnamefont {J.}~\bibnamefont {Chen}}, \bibinfo {author} {\bibfnamefont {Y.}~\bibnamefont {Wang}}, \bibinfo {author} {\bibfnamefont {Y.}~\bibnamefont {Wu}}, \bibinfo {author} {\bibfnamefont {S.}~\bibnamefont {Xu}}, \bibinfo {author} {\bibfnamefont {X.}~\bibnamefont {Zhu}}, \bibinfo {author} {\bibfnamefont {F.}~\bibnamefont {Jin}}, \bibinfo {author} {\bibfnamefont {Y.}~\bibnamefont {Gao}}, \bibinfo {author} {\bibfnamefont {Z.}~\bibnamefont {Tan}}, \bibinfo {author} {\bibfnamefont {Z.}~\bibnamefont {Cui}}, \bibinfo {author} {\bibfnamefont {N.}~\bibnamefont {Wang}}, \bibinfo {author} {\bibfnamefont {Y.}~\bibnamefont {Zou}}, \bibinfo {author} {\bibfnamefont {A.}~\bibnamefont {Zhang}}, \bibinfo {author} {\bibfnamefont {T.}~\bibnamefont
  {Li}}, \bibinfo {author} {\bibfnamefont {F.}~\bibnamefont {Shen}}, \bibinfo {author} {\bibfnamefont {J.}~\bibnamefont {Zhong}}, \bibinfo {author} {\bibfnamefont {Z.}~\bibnamefont {Bao}}, \bibinfo {author} {\bibfnamefont {Z.}~\bibnamefont {Zhu}}, \bibinfo {author} {\bibfnamefont {Y.}~\bibnamefont {Han}}, \bibinfo {author} {\bibfnamefont {Y.}~\bibnamefont {He}}, \bibinfo {author} {\bibfnamefont {J.}~\bibnamefont {Shen}}, \bibinfo {author} {\bibfnamefont {H.}~\bibnamefont {Wang}}, \bibinfo {author} {\bibfnamefont {J.-N.}\ \bibnamefont {Yang}}, \bibinfo {author} {\bibfnamefont {Z.}~\bibnamefont {Song}}, \bibinfo {author} {\bibfnamefont {J.}~\bibnamefont {Deng}}, \bibinfo {author} {\bibfnamefont {H.}~\bibnamefont {Dong}}, \bibinfo {author} {\bibfnamefont {Z.-Z.}\ \bibnamefont {Sun}}, \bibinfo {author} {\bibfnamefont {W.}~\bibnamefont {Li}}, \bibinfo {author} {\bibfnamefont {Q.}~\bibnamefont {Ye}}, \bibinfo {author} {\bibfnamefont {S.}~\bibnamefont {Jiang}}, \bibinfo {author} {\bibfnamefont {Y.}~\bibnamefont
  {Ma}}, \bibinfo {author} {\bibfnamefont {P.-X.}\ \bibnamefont {Shen}}, \bibinfo {author} {\bibfnamefont {P.}~\bibnamefont {Zhang}}, \bibinfo {author} {\bibfnamefont {H.}~\bibnamefont {Li}}, \bibinfo {author} {\bibfnamefont {Q.}~\bibnamefont {Guo}}, \bibinfo {author} {\bibfnamefont {Z.}~\bibnamefont {Wang}}, \bibinfo {author} {\bibfnamefont {C.}~\bibnamefont {Song}}, \bibinfo {author} {\bibfnamefont {H.}~\bibnamefont {Wang}},\ and\ \bibinfo {author} {\bibfnamefont {D.-L.}\ \bibnamefont {Deng}},\ }\href@noop {} {\bibinfo {title} {{Demonstration of low-overhead quantum error correction codes}}} (\bibinfo {year} {2025}),\ \Eprint {https://arxiv.org/abs/2505.09684} {arXiv:2505.09684 [quant-ph]} \BibitemShut {NoStop}%
\bibitem [{\citenamefont {Zhang}\ \emph {et~al.}(2024)\citenamefont {Zhang}, \citenamefont {Wu},\ and\ \citenamefont {Guo}}]{zhang2024facilitating}%
  \BibitemOpen
  \bibfield  {author} {\bibinfo {author} {\bibfnamefont {J.}~\bibnamefont {Zhang}}, \bibinfo {author} {\bibfnamefont {Y.-C.}\ \bibnamefont {Wu}},\ and\ \bibinfo {author} {\bibfnamefont {G.-P.}\ \bibnamefont {Guo}},\ }\bibfield  {title} {\bibinfo {title} {Facilitating practical fault-tolerant quantum computing based on color codes},\ }\href {https://doi.org/10.1103/PhysRevResearch.6.033086} {\bibfield  {journal} {\bibinfo  {journal} {Physical Review Research}\ }\textbf {\bibinfo {volume} {6}},\ \bibinfo {pages} {033086} (\bibinfo {year} {2024})}\BibitemShut {NoStop}%
\bibitem [{\citenamefont {Li}(2015)}]{li2015magic}%
  \BibitemOpen
  \bibfield  {author} {\bibinfo {author} {\bibfnamefont {Y.}~\bibnamefont {Li}},\ }\bibfield  {title} {\bibinfo {title} {A magic state’s fidelity can be superior to the operations that created it},\ }\href {https://doi.org/10.1088/1367-2630/17/2/023037} {\bibfield  {journal} {\bibinfo  {journal} {New Journal of Physics}\ }\textbf {\bibinfo {volume} {17}},\ \bibinfo {pages} {023037} (\bibinfo {year} {2015})}\BibitemShut {NoStop}%
\bibitem [{\citenamefont {Fowler}\ \emph {et~al.}(2009)\citenamefont {Fowler}, \citenamefont {Stephens},\ and\ \citenamefont {Groszkowski}}]{fowler2009high}%
  \BibitemOpen
  \bibfield  {author} {\bibinfo {author} {\bibfnamefont {A.~G.}\ \bibnamefont {Fowler}}, \bibinfo {author} {\bibfnamefont {A.~M.}\ \bibnamefont {Stephens}},\ and\ \bibinfo {author} {\bibfnamefont {P.}~\bibnamefont {Groszkowski}},\ }\bibfield  {title} {\bibinfo {title} {High-threshold universal quantum computation on the surface code},\ }\href {https://doi.org/10.1103/PhysRevA.80.052312} {\bibfield  {journal} {\bibinfo  {journal} {Physical Review A}\ }\textbf {\bibinfo {volume} {80}},\ \bibinfo {pages} {052312} (\bibinfo {year} {2009})}\BibitemShut {NoStop}%
\bibitem [{\citenamefont {Raussendorf}\ \emph {et~al.}(2006)\citenamefont {Raussendorf}, \citenamefont {Harrington},\ and\ \citenamefont {Goyal}}]{raussendorf2006fault}%
  \BibitemOpen
  \bibfield  {author} {\bibinfo {author} {\bibfnamefont {R.}~\bibnamefont {Raussendorf}}, \bibinfo {author} {\bibfnamefont {J.}~\bibnamefont {Harrington}},\ and\ \bibinfo {author} {\bibfnamefont {K.}~\bibnamefont {Goyal}},\ }\bibfield  {title} {\bibinfo {title} {A fault-tolerant one-way quantum computer},\ }\href {https://doi.org/https://doi.org/10.1016/j.aop.2006.01.012} {\bibfield  {journal} {\bibinfo  {journal} {Annals of Physics}\ }\textbf {\bibinfo {volume} {321}},\ \bibinfo {pages} {2242} (\bibinfo {year} {2006})}\BibitemShut {NoStop}%
\bibitem [{\citenamefont {Horsman}\ \emph {et~al.}(2012)\citenamefont {Horsman}, \citenamefont {Fowler}, \citenamefont {Devitt},\ and\ \citenamefont {Meter}}]{horsman2012surface}%
  \BibitemOpen
  \bibfield  {author} {\bibinfo {author} {\bibfnamefont {D.}~\bibnamefont {Horsman}}, \bibinfo {author} {\bibfnamefont {A.~G.}\ \bibnamefont {Fowler}}, \bibinfo {author} {\bibfnamefont {S.}~\bibnamefont {Devitt}},\ and\ \bibinfo {author} {\bibfnamefont {R.~V.}\ \bibnamefont {Meter}},\ }\bibfield  {title} {\bibinfo {title} {Surface code quantum computing by lattice surgery},\ }\href {https://doi.org/10.1088/1367-2630/14/12/123011} {\bibfield  {journal} {\bibinfo  {journal} {New Journal of Physics}\ }\textbf {\bibinfo {volume} {14}},\ \bibinfo {pages} {123011} (\bibinfo {year} {2012})}\BibitemShut {NoStop}%
\bibitem [{\citenamefont {Dankert}\ \emph {et~al.}(2009{\natexlab{b}})\citenamefont {Dankert}, \citenamefont {Cleve}, \citenamefont {Emerson},\ and\ \citenamefont {Livine}}]{DCEL09}%
  \BibitemOpen
  \bibfield  {author} {\bibinfo {author} {\bibfnamefont {C.}~\bibnamefont {Dankert}}, \bibinfo {author} {\bibfnamefont {R.}~\bibnamefont {Cleve}}, \bibinfo {author} {\bibfnamefont {J.}~\bibnamefont {Emerson}},\ and\ \bibinfo {author} {\bibfnamefont {E.}~\bibnamefont {Livine}},\ }\bibfield  {title} {\bibinfo {title} {{Exact and approximate unitary 2-designs and their application to fidelity estimation}},\ }\href {https://doi.org/10.1103/PhysRevA.80.012304} {\bibfield  {journal} {\bibinfo  {journal} {Physical Review A}\ }\textbf {\bibinfo {volume} {80}},\ \bibinfo {pages} {012304} (\bibinfo {year} {2009}{\natexlab{b}})}\BibitemShut {NoStop}%
\bibitem [{\citenamefont {Gross}\ \emph {et~al.}(2007)\citenamefont {Gross}, \citenamefont {Audenaert},\ and\ \citenamefont {Eisert}}]{DAE07}%
  \BibitemOpen
  \bibfield  {author} {\bibinfo {author} {\bibfnamefont {D.}~\bibnamefont {Gross}}, \bibinfo {author} {\bibfnamefont {K.}~\bibnamefont {Audenaert}},\ and\ \bibinfo {author} {\bibfnamefont {J.}~\bibnamefont {Eisert}},\ }\bibfield  {title} {\bibinfo {title} {{Evenly distributed unitaries: On the structure of unitary designs}},\ }\href {https://doi.org/10.1063/1.2716992} {\bibfield  {journal} {\bibinfo  {journal} {Journal of Mathematical Physics}\ }\textbf {\bibinfo {volume} {48}},\ \bibinfo {pages} {052104} (\bibinfo {year} {2007})}\BibitemShut {NoStop}%
\bibitem [{\citenamefont {Watrous}(2018)}]{Wat18_book}%
  \BibitemOpen
  \bibfield  {author} {\bibinfo {author} {\bibfnamefont {J.}~\bibnamefont {Watrous}},\ }\href {https://doi.org/10.1017/9781316848142} {\emph {\bibinfo {title} {{The Theory of Quantum Information}}}}\ (\bibinfo  {publisher} {Cambridge University Press},\ \bibinfo {year} {2018})\BibitemShut {NoStop}%
\bibitem [{\citenamefont {Zwerger}\ \emph {et~al.}(2014)\citenamefont {Zwerger}, \citenamefont {Briegel},\ and\ \citenamefont {D\"ur}}]{zwerger2014robustness}%
  \BibitemOpen
  \bibfield  {author} {\bibinfo {author} {\bibfnamefont {M.}~\bibnamefont {Zwerger}}, \bibinfo {author} {\bibfnamefont {H.~J.}\ \bibnamefont {Briegel}},\ and\ \bibinfo {author} {\bibfnamefont {W.}~\bibnamefont {D\"ur}},\ }\bibfield  {title} {\bibinfo {title} {Robustness of hashing protocols for entanglement purification},\ }\href {https://doi.org/10.1103/PhysRevA.90.012314} {\bibfield  {journal} {\bibinfo  {journal} {Physical Review A}\ }\textbf {\bibinfo {volume} {90}},\ \bibinfo {pages} {012314} (\bibinfo {year} {2014})}\BibitemShut {NoStop}%
\bibitem [{\citenamefont {Dupuis}\ \emph {et~al.}(2014)\citenamefont {Dupuis}, \citenamefont {Berta}, \citenamefont {Wullschleger},\ and\ \citenamefont {Renner}}]{dupuis2014decoupling}%
  \BibitemOpen
  \bibfield  {author} {\bibinfo {author} {\bibfnamefont {F.}~\bibnamefont {Dupuis}}, \bibinfo {author} {\bibfnamefont {M.}~\bibnamefont {Berta}}, \bibinfo {author} {\bibfnamefont {J.}~\bibnamefont {Wullschleger}},\ and\ \bibinfo {author} {\bibfnamefont {R.}~\bibnamefont {Renner}},\ }\bibfield  {title} {\bibinfo {title} {One-{Shot} {Decoupling}},\ }\href {https://doi.org/10.1007/s00220-014-1990-4} {\bibfield  {journal} {\bibinfo  {journal} {Communications in Mathematical Physics}\ }\textbf {\bibinfo {volume} {328}},\ \bibinfo {pages} {251} (\bibinfo {year} {2014})}\BibitemShut {NoStop}%
\bibitem [{\citenamefont {D\"ur}\ \emph {et~al.}(1999)\citenamefont {D\"ur}, \citenamefont {Briegel}, \citenamefont {Cirac},\ and\ \citenamefont {Zoller}}]{DBC99}%
  \BibitemOpen
  \bibfield  {author} {\bibinfo {author} {\bibfnamefont {W.}~\bibnamefont {D\"ur}}, \bibinfo {author} {\bibfnamefont {H.-J.}\ \bibnamefont {Briegel}}, \bibinfo {author} {\bibfnamefont {J.~I.}\ \bibnamefont {Cirac}},\ and\ \bibinfo {author} {\bibfnamefont {P.}~\bibnamefont {Zoller}},\ }\bibfield  {title} {\bibinfo {title} {{Quantum repeaters based on entanglement purification}},\ }\href {https://doi.org/10.1103/PhysRevA.59.169} {\bibfield  {journal} {\bibinfo  {journal} {Physical Review A}\ }\textbf {\bibinfo {volume} {59}},\ \bibinfo {pages} {169} (\bibinfo {year} {1999})}\BibitemShut {NoStop}%
\bibitem [{\citenamefont {D\"ur}\ and\ \citenamefont {Briegel}(2003)}]{dur2003entanglement}%
  \BibitemOpen
  \bibfield  {author} {\bibinfo {author} {\bibfnamefont {W.}~\bibnamefont {D\"ur}}\ and\ \bibinfo {author} {\bibfnamefont {H.-J.}\ \bibnamefont {Briegel}},\ }\bibfield  {title} {\bibinfo {title} {{Entanglement Purification for Quantum Computation}},\ }\href {https://doi.org/10.1103/PhysRevLett.90.067901} {\bibfield  {journal} {\bibinfo  {journal} {Physical Review Letters}\ }\textbf {\bibinfo {volume} {90}},\ \bibinfo {pages} {067901} (\bibinfo {year} {2003})}\BibitemShut {NoStop}%
\bibitem [{\citenamefont {Sangouard}\ \emph {et~al.}(2011)\citenamefont {Sangouard}, \citenamefont {Simon}, \citenamefont {de~Riedmatten},\ and\ \citenamefont {Gisin}}]{SSR+11}%
  \BibitemOpen
  \bibfield  {author} {\bibinfo {author} {\bibfnamefont {N.}~\bibnamefont {Sangouard}}, \bibinfo {author} {\bibfnamefont {C.}~\bibnamefont {Simon}}, \bibinfo {author} {\bibfnamefont {H.}~\bibnamefont {de~Riedmatten}},\ and\ \bibinfo {author} {\bibfnamefont {N.}~\bibnamefont {Gisin}},\ }\bibfield  {title} {\bibinfo {title} {{Quantum repeaters based on atomic ensembles and linear optics}},\ }\href {https://doi.org/10.1103/RevModPhys.83.33} {\bibfield  {journal} {\bibinfo  {journal} {Reviews of Modern Physics}\ }\textbf {\bibinfo {volume} {83}},\ \bibinfo {pages} {33} (\bibinfo {year} {2011})}\BibitemShut {NoStop}%
\bibitem [{\citenamefont {Azuma}\ \emph {et~al.}(2023)\citenamefont {Azuma}, \citenamefont {Economou}, \citenamefont {Elkouss}, \citenamefont {Hilaire}, \citenamefont {Jiang}, \citenamefont {Lo},\ and\ \citenamefont {Tzitrin}}]{azuma2023repeatersRMP}%
  \BibitemOpen
  \bibfield  {author} {\bibinfo {author} {\bibfnamefont {K.}~\bibnamefont {Azuma}}, \bibinfo {author} {\bibfnamefont {S.~E.}\ \bibnamefont {Economou}}, \bibinfo {author} {\bibfnamefont {D.}~\bibnamefont {Elkouss}}, \bibinfo {author} {\bibfnamefont {P.}~\bibnamefont {Hilaire}}, \bibinfo {author} {\bibfnamefont {L.}~\bibnamefont {Jiang}}, \bibinfo {author} {\bibfnamefont {H.-K.}\ \bibnamefont {Lo}},\ and\ \bibinfo {author} {\bibfnamefont {I.}~\bibnamefont {Tzitrin}},\ }\bibfield  {title} {\bibinfo {title} {{Quantum repeaters: From quantum networks to the quantum internet}},\ }\href {https://doi.org/10.1103/RevModPhys.95.045006} {\bibfield  {journal} {\bibinfo  {journal} {Reviews of Modern Physics}\ }\textbf {\bibinfo {volume} {95}},\ \bibinfo {pages} {045006} (\bibinfo {year} {2023})}\BibitemShut {NoStop}%
\bibitem [{\citenamefont {Shchukin}\ and\ \citenamefont {Van~Loock}(2022)}]{shchukin2022optimal}%
  \BibitemOpen
  \bibfield  {author} {\bibinfo {author} {\bibfnamefont {E.}~\bibnamefont {Shchukin}}\ and\ \bibinfo {author} {\bibfnamefont {P.}~\bibnamefont {Van~Loock}},\ }\bibfield  {title} {\bibinfo {title} {Optimal {Entanglement} {Swapping} in {Quantum} {Repeaters}},\ }\href {https://doi.org/10.1103/PhysRevLett.128.150502} {\bibfield  {journal} {\bibinfo  {journal} {Physical Review Letters}\ }\textbf {\bibinfo {volume} {128}},\ \bibinfo {pages} {150502} (\bibinfo {year} {2022})}\BibitemShut {NoStop}%
\bibitem [{\citenamefont {Haldar}\ \emph {et~al.}(2025)\citenamefont {Haldar}, \citenamefont {Barge}, \citenamefont {Cheng}, \citenamefont {Chang}, \citenamefont {Kirby}, \citenamefont {Khatri}, \citenamefont {Wong},\ and\ \citenamefont {Lee}}]{haldar2024reducing}%
  \BibitemOpen
  \bibfield  {author} {\bibinfo {author} {\bibfnamefont {S.}~\bibnamefont {Haldar}}, \bibinfo {author} {\bibfnamefont {P.~J.}\ \bibnamefont {Barge}}, \bibinfo {author} {\bibfnamefont {X.}~\bibnamefont {Cheng}}, \bibinfo {author} {\bibfnamefont {K.-C.}\ \bibnamefont {Chang}}, \bibinfo {author} {\bibfnamefont {B.~T.}\ \bibnamefont {Kirby}}, \bibinfo {author} {\bibfnamefont {S.}~\bibnamefont {Khatri}}, \bibinfo {author} {\bibfnamefont {C.~W.}\ \bibnamefont {Wong}},\ and\ \bibinfo {author} {\bibfnamefont {H.}~\bibnamefont {Lee}},\ }\bibfield  {title} {\bibinfo {title} {Reducing classical communication costs in multiplexed quantum repeaters using hardware-aware quasi-local policies},\ }\href {https://doi.org/10.1038/s42005-025-02029-w} {\bibfield  {journal} {\bibinfo  {journal} {Communications Physics}\ }\textbf {\bibinfo {volume} {8}},\ \bibinfo {pages} {132} (\bibinfo {year} {2025})}\BibitemShut {NoStop}%
\bibitem [{\citenamefont {Khatri}(2022)}]{khatri2022networkMDP}%
  \BibitemOpen
  \bibfield  {author} {\bibinfo {author} {\bibfnamefont {S.}~\bibnamefont {Khatri}},\ }\bibfield  {title} {\bibinfo {title} {{On the design and analysis of near-term quantum network protocols using Markov decision processes}},\ }\href {https://doi.org/10.1116/5.0084653} {\bibfield  {journal} {\bibinfo  {journal} {AVS Quantum Science}\ }\textbf {\bibinfo {volume} {4}},\ \bibinfo {pages} {030501} (\bibinfo {year} {2022})}\BibitemShut {NoStop}%
\bibitem [{\citenamefont {Gottesman}(2004)}]{gottesman1997stabilizer}%
  \BibitemOpen
  \bibfield  {author} {\bibinfo {author} {\bibfnamefont {D.~E.}\ \bibnamefont {Gottesman}},\ }\emph {\bibinfo {title} {Stabilizer {Codes} and {Quantum} {Error} {Correction}}},\ \href {https://doi.org/10.7907/RZR7-DT72} {Ph.D. thesis},\ \bibinfo  {school} {California Institute of Technology} (\bibinfo {year} {2004})\BibitemShut {NoStop}%
\bibitem [{\citenamefont {Steane}(1996)}]{steane1997error}%
  \BibitemOpen
  \bibfield  {author} {\bibinfo {author} {\bibfnamefont {A.~M.}\ \bibnamefont {Steane}},\ }\bibfield  {title} {\bibinfo {title} {{Error Correcting Codes in Quantum Theory}},\ }\href {https://doi.org/10.1103/PhysRevLett.77.793} {\bibfield  {journal} {\bibinfo  {journal} {Physical Review Letters}\ }\textbf {\bibinfo {volume} {77}},\ \bibinfo {pages} {793} (\bibinfo {year} {1996})}\BibitemShut {NoStop}%
\bibitem [{\citenamefont {Schumacher}(1995)}]{schumacher1995quantum}%
  \BibitemOpen
  \bibfield  {author} {\bibinfo {author} {\bibfnamefont {B.}~\bibnamefont {Schumacher}},\ }\bibfield  {title} {\bibinfo {title} {Quantum coding},\ }\href {https://doi.org/10.1103/PhysRevA.51.2738} {\bibfield  {journal} {\bibinfo  {journal} {Physical Review A}\ }\textbf {\bibinfo {volume} {51}},\ \bibinfo {pages} {2738} (\bibinfo {year} {1995})}\BibitemShut {NoStop}%
\bibitem [{\citenamefont {Rozp\k{e}dek}\ \emph {et~al.}(2018)\citenamefont {Rozp\k{e}dek}, \citenamefont {Schiet}, \citenamefont {Thinh}, \citenamefont {Elkouss}, \citenamefont {Doherty},\ and\ \citenamefont {Wehner}}]{RST+18}%
  \BibitemOpen
  \bibfield  {author} {\bibinfo {author} {\bibfnamefont {F.}~\bibnamefont {Rozp\k{e}dek}}, \bibinfo {author} {\bibfnamefont {T.}~\bibnamefont {Schiet}}, \bibinfo {author} {\bibfnamefont {L.~P.}\ \bibnamefont {Thinh}}, \bibinfo {author} {\bibfnamefont {D.}~\bibnamefont {Elkouss}}, \bibinfo {author} {\bibfnamefont {A.~C.}\ \bibnamefont {Doherty}},\ and\ \bibinfo {author} {\bibfnamefont {S.}~\bibnamefont {Wehner}},\ }\bibfield  {title} {\bibinfo {title} {{Optimizing practical entanglement distillation}},\ }\href {https://doi.org/10.1103/PhysRevA.97.062333} {\bibfield  {journal} {\bibinfo  {journal} {Physical Review A}\ }\textbf {\bibinfo {volume} {97}},\ \bibinfo {pages} {062333} (\bibinfo {year} {2018})}\BibitemShut {NoStop}%
\bibitem [{\citenamefont {D\"{u}r}\ and\ \citenamefont {Briegel}(2007)}]{dur2007purificationQECreview}%
  \BibitemOpen
  \bibfield  {author} {\bibinfo {author} {\bibfnamefont {W.}~\bibnamefont {D\"{u}r}}\ and\ \bibinfo {author} {\bibfnamefont {H.~J.}\ \bibnamefont {Briegel}},\ }\bibfield  {title} {\bibinfo {title} {{Entanglement purification and quantum error correction}},\ }\href {https://doi.org/10.1088/0034-4885/70/8/R03} {\bibfield  {journal} {\bibinfo  {journal} {Reports on Progress in Physics}\ }\textbf {\bibinfo {volume} {70}},\ \bibinfo {pages} {1381} (\bibinfo {year} {2007})}\BibitemShut {NoStop}%
\bibitem [{\citenamefont {D\"{u}r}\ and\ \citenamefont {Briegel}(2016)}]{dur2016distillationreview}%
  \BibitemOpen
  \bibfield  {author} {\bibinfo {author} {\bibfnamefont {W.}~\bibnamefont {D\"{u}r}}\ and\ \bibinfo {author} {\bibfnamefont {H.-J.}\ \bibnamefont {Briegel}},\ }\bibinfo {title} {{Purification and Distillation}},\ in\ \href {https://doi.org/https://doi.org/10.1002/9783527805785.ch12} {\emph {\bibinfo {booktitle} {Quantum Information}}}\ (\bibinfo  {publisher} {John Wiley \& Sons, Ltd},\ \bibinfo {year} {2016})\ Chap.~\bibinfo {chapter} {12}, pp.\ \bibinfo {pages} {231--263}\BibitemShut {NoStop}%
\bibitem [{\citenamefont {Calderbank}\ and\ \citenamefont {Shor}(1996)}]{calderbank1996good}%
  \BibitemOpen
  \bibfield  {author} {\bibinfo {author} {\bibfnamefont {A.~R.}\ \bibnamefont {Calderbank}}\ and\ \bibinfo {author} {\bibfnamefont {P.~W.}\ \bibnamefont {Shor}},\ }\bibfield  {title} {\bibinfo {title} {Good quantum error-correcting codes exist},\ }\href {https://doi.org/10.1103/physreva.54.1098} {\bibfield  {journal} {\bibinfo  {journal} {Physical Review A}\ }\textbf {\bibinfo {volume} {54}},\ \bibinfo {pages} {1098–1105} (\bibinfo {year} {1996})}\BibitemShut {NoStop}%
\bibitem [{\citenamefont {Alber}\ \emph {et~al.}(2001)\citenamefont {Alber}, \citenamefont {Delgado}, \citenamefont {Gisin},\ and\ \citenamefont {Jex}}]{alber2001efficient}%
  \BibitemOpen
  \bibfield  {author} {\bibinfo {author} {\bibfnamefont {G.}~\bibnamefont {Alber}}, \bibinfo {author} {\bibfnamefont {A.}~\bibnamefont {Delgado}}, \bibinfo {author} {\bibfnamefont {N.}~\bibnamefont {Gisin}},\ and\ \bibinfo {author} {\bibfnamefont {I.}~\bibnamefont {Jex}},\ }\bibfield  {title} {\bibinfo {title} {{Efficient bipartite quantum state purification in arbitrary dimensional Hilbert spaces}},\ }\href {https://doi.org/10.1088/0305-4470/34/42/307} {\bibfield  {journal} {\bibinfo  {journal} {Journal of Physics A: Mathematical and General}\ }\textbf {\bibinfo {volume} {34}},\ \bibinfo {pages} {8821} (\bibinfo {year} {2001})}\BibitemShut {NoStop}%
\bibitem [{\citenamefont {Matsumoto}(2003)}]{matsumoto2003conversion}%
  \BibitemOpen
  \bibfield  {author} {\bibinfo {author} {\bibfnamefont {R.}~\bibnamefont {Matsumoto}},\ }\bibfield  {title} {\bibinfo {title} {Conversion of a general quantum stabilizer code to an entanglement distillation protocol},\ }\href {https://doi.org/10.1088/0305-4470/36/29/316} {\bibfield  {journal} {\bibinfo  {journal} {Journal of Physics A: Mathematical and General}\ }\textbf {\bibinfo {volume} {36}},\ \bibinfo {pages} {8113–8127} (\bibinfo {year} {2003})}\BibitemShut {NoStop}%
\bibitem [{\citenamefont {Hostens}\ \emph {et~al.}(2005)\citenamefont {Hostens}, \citenamefont {Dehaene},\ and\ \citenamefont {De~Moor}}]{hostens2004stabilizer}%
  \BibitemOpen
  \bibfield  {author} {\bibinfo {author} {\bibfnamefont {E.}~\bibnamefont {Hostens}}, \bibinfo {author} {\bibfnamefont {J.}~\bibnamefont {Dehaene}},\ and\ \bibinfo {author} {\bibfnamefont {B.}~\bibnamefont {De~Moor}},\ }\bibfield  {title} {\bibinfo {title} {{Stabilizer states and Clifford operations for systems of arbitrary dimensions and modular arithmetic}},\ }\href {https://doi.org/10.1103/PhysRevA.71.042315} {\bibfield  {journal} {\bibinfo  {journal} {Physical Review A}\ }\textbf {\bibinfo {volume} {71}},\ \bibinfo {pages} {042315} (\bibinfo {year} {2005})}\BibitemShut {NoStop}%
\bibitem [{\citenamefont {Aschauer}(2005)}]{aschauer2005quantum}%
  \BibitemOpen
  \bibfield  {author} {\bibinfo {author} {\bibfnamefont {H.}~\bibnamefont {Aschauer}},\ }\emph {\bibinfo {title} {Quantum communication in noisy environments}},\ \href {https://doi.org/10.5282/EDOC.3588} {Ph.D. thesis},\ \bibinfo  {school} {Ludwig-Maximilians-Universität München} (\bibinfo {year} {2005})\BibitemShut {NoStop}%
\bibitem [{\citenamefont {Fujii}\ and\ \citenamefont {Yamamoto}(2009)}]{fujii2009entpurifdoubleselection}%
  \BibitemOpen
  \bibfield  {author} {\bibinfo {author} {\bibfnamefont {K.}~\bibnamefont {Fujii}}\ and\ \bibinfo {author} {\bibfnamefont {K.}~\bibnamefont {Yamamoto}},\ }\bibfield  {title} {\bibinfo {title} {{Entanglement purification with double selection}},\ }\href {https://doi.org/10.1103/PhysRevA.80.042308} {\bibfield  {journal} {\bibinfo  {journal} {Physical Review A}\ }\textbf {\bibinfo {volume} {80}},\ \bibinfo {pages} {042308} (\bibinfo {year} {2009})}\BibitemShut {NoStop}%
\bibitem [{\citenamefont {Gidney}(2023)}]{gidney2023tetrationally}%
  \BibitemOpen
  \bibfield  {author} {\bibinfo {author} {\bibfnamefont {C.}~\bibnamefont {Gidney}},\ }\href@noop {} {\bibinfo {title} {Tetrationally compact entanglement purification}} (\bibinfo {year} {2023}),\ \Eprint {https://arxiv.org/abs/2311.10971} {arXiv:2311.10971 [quant-ph]} \BibitemShut {NoStop}%
\bibitem [{\citenamefont {Brown}\ and\ \citenamefont {Fawzi}(2013)}]{brown2023short}%
  \BibitemOpen
  \bibfield  {author} {\bibinfo {author} {\bibfnamefont {W.}~\bibnamefont {Brown}}\ and\ \bibinfo {author} {\bibfnamefont {O.}~\bibnamefont {Fawzi}},\ }\bibfield  {title} {\bibinfo {title} {Short random circuits define good quantum error correcting codes},\ }in\ \href {https://doi.org/10.1109/isit.2013.6620245} {\emph {\bibinfo {booktitle} {2013 IEEE International Symposium on Information Theory}}}\ (\bibinfo  {publisher} {IEEE},\ \bibinfo {year} {2013})\BibitemShut {NoStop}%
\bibitem [{\citenamefont {Gullans}\ \emph {et~al.}(2021)\citenamefont {Gullans}, \citenamefont {Krastanov}, \citenamefont {Huse}, \citenamefont {Jiang},\ and\ \citenamefont {Flammia}}]{gullan2021quantum}%
  \BibitemOpen
  \bibfield  {author} {\bibinfo {author} {\bibfnamefont {M.~J.}\ \bibnamefont {Gullans}}, \bibinfo {author} {\bibfnamefont {S.}~\bibnamefont {Krastanov}}, \bibinfo {author} {\bibfnamefont {D.~A.}\ \bibnamefont {Huse}}, \bibinfo {author} {\bibfnamefont {L.}~\bibnamefont {Jiang}},\ and\ \bibinfo {author} {\bibfnamefont {S.~T.}\ \bibnamefont {Flammia}},\ }\bibfield  {title} {\bibinfo {title} {Quantum {Coding} with {Low}-{Depth} {Random} {Circuits}},\ }\href {https://doi.org/10.1103/PhysRevX.11.031066} {\bibfield  {journal} {\bibinfo  {journal} {Physical Review X}\ }\textbf {\bibinfo {volume} {11}},\ \bibinfo {pages} {031066} (\bibinfo {year} {2021})}\BibitemShut {NoStop}%
\bibitem [{\citenamefont {Darmawan}\ \emph {et~al.}(2024)\citenamefont {Darmawan}, \citenamefont {Nakata}, \citenamefont {Tamiya},\ and\ \citenamefont {Yamasaki}}]{darmawan2024low}%
  \BibitemOpen
  \bibfield  {author} {\bibinfo {author} {\bibfnamefont {A.~S.}\ \bibnamefont {Darmawan}}, \bibinfo {author} {\bibfnamefont {Y.}~\bibnamefont {Nakata}}, \bibinfo {author} {\bibfnamefont {S.}~\bibnamefont {Tamiya}},\ and\ \bibinfo {author} {\bibfnamefont {H.}~\bibnamefont {Yamasaki}},\ }\bibfield  {title} {\bibinfo {title} {Low-depth random {Clifford} circuits for quantum coding against {Pauli} noise using a tensor-network decoder},\ }\href {https://doi.org/10.1103/PhysRevResearch.6.023055} {\bibfield  {journal} {\bibinfo  {journal} {Physical Review Research}\ }\textbf {\bibinfo {volume} {6}},\ \bibinfo {pages} {023055} (\bibinfo {year} {2024})}\BibitemShut {NoStop}%
\bibitem [{\citenamefont {Nakata}\ \emph {et~al.}(2021)\citenamefont {Nakata}, \citenamefont {Wakakuwa},\ and\ \citenamefont {Yamasaki}}]{nakata2021one}%
  \BibitemOpen
  \bibfield  {author} {\bibinfo {author} {\bibfnamefont {Y.}~\bibnamefont {Nakata}}, \bibinfo {author} {\bibfnamefont {E.}~\bibnamefont {Wakakuwa}},\ and\ \bibinfo {author} {\bibfnamefont {H.}~\bibnamefont {Yamasaki}},\ }\bibfield  {title} {\bibinfo {title} {One-shot quantum error correction of classical and quantum information},\ }\href {https://doi.org/10.1103/PhysRevA.104.012408} {\bibfield  {journal} {\bibinfo  {journal} {Physical Review A}\ }\textbf {\bibinfo {volume} {104}},\ \bibinfo {pages} {012408} (\bibinfo {year} {2021})}\BibitemShut {NoStop}%
\bibitem [{\citenamefont {Nelson}\ \emph {et~al.}(2023)\citenamefont {Nelson}, \citenamefont {Bentsen}, \citenamefont {Flammia},\ and\ \citenamefont {Gullans}}]{nelson2023faulttolerantrandom}%
  \BibitemOpen
  \bibfield  {author} {\bibinfo {author} {\bibfnamefont {J.}~\bibnamefont {Nelson}}, \bibinfo {author} {\bibfnamefont {G.}~\bibnamefont {Bentsen}}, \bibinfo {author} {\bibfnamefont {S.~T.}\ \bibnamefont {Flammia}},\ and\ \bibinfo {author} {\bibfnamefont {M.~J.}\ \bibnamefont {Gullans}},\ }\href@noop {} {\bibinfo {title} {Fault-tolerant quantum memory using low-depth random circuit codes}} (\bibinfo {year} {2023}),\ \Eprint {https://arxiv.org/abs/2311.17985} {arXiv:2311.17985 [quant-ph]} \BibitemShut {NoStop}%
\bibitem [{\citenamefont {Calderbank}\ \emph {et~al.}(1997)\citenamefont {Calderbank}, \citenamefont {Rains}, \citenamefont {Shor},\ and\ \citenamefont {Sloane}}]{calderbank1997quantum}%
  \BibitemOpen
  \bibfield  {author} {\bibinfo {author} {\bibfnamefont {A.~R.}\ \bibnamefont {Calderbank}}, \bibinfo {author} {\bibfnamefont {E.~M.}\ \bibnamefont {Rains}}, \bibinfo {author} {\bibfnamefont {P.~W.}\ \bibnamefont {Shor}},\ and\ \bibinfo {author} {\bibfnamefont {N.~J.~A.}\ \bibnamefont {Sloane}},\ }\bibfield  {title} {\bibinfo {title} {Quantum error correction and orthogonal geometry},\ }\href {https://doi.org/10.1103/physrevlett.78.405} {\bibfield  {journal} {\bibinfo  {journal} {Physical Review Letters}\ }\textbf {\bibinfo {volume} {78}},\ \bibinfo {pages} {405–408} (\bibinfo {year} {1997})}\BibitemShut {NoStop}%
\bibitem [{\citenamefont {Gallager}(1972)}]{gallager1986information}%
  \BibitemOpen
  \bibfield  {author} {\bibinfo {author} {\bibfnamefont {R.}~\bibnamefont {Gallager}},\ }\href {https://doi.org/10.1007/978-3-7091-2945-6} {\emph {\bibinfo {title} {Information {Theory} and {Reliable} {Communication}}}}\ (\bibinfo  {publisher} {Springer Vienna},\ \bibinfo {address} {Vienna},\ \bibinfo {year} {1972})\BibitemShut {NoStop}%
\bibitem [{\citenamefont {McEliece}(2002)}]{mceliece2003theory}%
  \BibitemOpen
  \bibfield  {author} {\bibinfo {author} {\bibfnamefont {R.}~\bibnamefont {McEliece}},\ }\href {https://doi.org/10.1017/CBO9780511606267} {\emph {\bibinfo {title} {The {Theory} of {Information} and {Coding}}}},\ \bibinfo {edition} {2nd}\ ed.\ (\bibinfo  {publisher} {Cambridge University Press},\ \bibinfo {year} {2002})\BibitemShut {NoStop}%
\bibitem [{\citenamefont {Cover}\ and\ \citenamefont {Thomas}(2006)}]{cover_thomas_book}%
  \BibitemOpen
  \bibfield  {author} {\bibinfo {author} {\bibfnamefont {T.~M.}\ \bibnamefont {Cover}}\ and\ \bibinfo {author} {\bibfnamefont {J.~A.}\ \bibnamefont {Thomas}},\ }\href {https://doi.org/10.1002/047174882X} {\emph {\bibinfo {title} {{Elements of Information Theory}}}},\ \bibinfo {edition} {2nd}\ ed.\ (\bibinfo  {publisher} {John Wiley \& Sons, Inc.},\ \bibinfo {year} {2006})\BibitemShut {NoStop}%
\bibitem [{\citenamefont {Vollbrecht}\ and\ \citenamefont {Verstraete}(2005)}]{vollbrecht2005interpolation}%
  \BibitemOpen
  \bibfield  {author} {\bibinfo {author} {\bibfnamefont {K.}~\bibnamefont {Vollbrecht}}\ and\ \bibinfo {author} {\bibfnamefont {F.}~\bibnamefont {Verstraete}},\ }\bibfield  {title} {\bibinfo {title} {Interpolation of recurrence and hashing entanglement distillation protocols},\ }\href {https://doi.org/10.1103/PhysRevA.71.062325} {\bibfield  {journal} {\bibinfo  {journal} {Physical Review A}\ }\textbf {\bibinfo {volume} {71}},\ \bibinfo {pages} {062325} (\bibinfo {year} {2005})}\BibitemShut {NoStop}%
\bibitem [{\citenamefont {Shor}(1995)}]{Shor95errorcorrection}%
  \BibitemOpen
  \bibfield  {author} {\bibinfo {author} {\bibfnamefont {P.~W.}\ \bibnamefont {Shor}},\ }\bibfield  {title} {\bibinfo {title} {{Scheme for reducing decoherence in quantum computer memory}},\ }\href {https://doi.org/10.1103/PhysRevA.52.R2493} {\bibfield  {journal} {\bibinfo  {journal} {Physical Review A}\ }\textbf {\bibinfo {volume} {52}},\ \bibinfo {pages} {R2493} (\bibinfo {year} {1995})}\BibitemShut {NoStop}%
\bibitem [{\citenamefont {Schumacher}(1996)}]{schumacher1996entanglementnoisychannels}%
  \BibitemOpen
  \bibfield  {author} {\bibinfo {author} {\bibfnamefont {B.}~\bibnamefont {Schumacher}},\ }\bibfield  {title} {\bibinfo {title} {Sending entanglement through noisy quantum channels},\ }\href {https://doi.org/10.1103/PhysRevA.54.2614} {\bibfield  {journal} {\bibinfo  {journal} {Physical Review A}\ }\textbf {\bibinfo {volume} {54}},\ \bibinfo {pages} {2614} (\bibinfo {year} {1996})}\BibitemShut {NoStop}%
\bibitem [{\citenamefont {Lloyd}(1997)}]{lloyd1997capacity}%
  \BibitemOpen
  \bibfield  {author} {\bibinfo {author} {\bibfnamefont {S.}~\bibnamefont {Lloyd}},\ }\bibfield  {title} {\bibinfo {title} {Capacity of the noisy quantum channel},\ }\href {https://doi.org/10.1103/PhysRevA.55.1613} {\bibfield  {journal} {\bibinfo  {journal} {Physical Review A}\ }\textbf {\bibinfo {volume} {55}},\ \bibinfo {pages} {1613} (\bibinfo {year} {1997})}\BibitemShut {NoStop}%
\bibitem [{\citenamefont {Shor}(2003)}]{shor2002quantum}%
  \BibitemOpen
  \bibfield  {author} {\bibinfo {author} {\bibfnamefont {P.~W.}\ \bibnamefont {Shor}},\ }\bibfield  {title} {\bibinfo {title} {Capacities of quantum channels and how to find them},\ }\href {https://doi.org/10.1007/s10107-003-0446-y} {\bibfield  {journal} {\bibinfo  {journal} {Mathematical Programming}\ }\textbf {\bibinfo {volume} {97}},\ \bibinfo {pages} {311–335} (\bibinfo {year} {2003})}\BibitemShut {NoStop}%
\bibitem [{\citenamefont {Terhal}(2004)}]{terhal2004entanglement}%
  \BibitemOpen
  \bibfield  {author} {\bibinfo {author} {\bibfnamefont {B.~M.}\ \bibnamefont {Terhal}},\ }\bibfield  {title} {\bibinfo {title} {Is entanglement monogamous?},\ }\href {https://doi.org/10.1147/rd.481.0071} {\bibfield  {journal} {\bibinfo  {journal} {IBM Journal of Research and Development}\ }\textbf {\bibinfo {volume} {48}},\ \bibinfo {pages} {71} (\bibinfo {year} {2004})}\BibitemShut {NoStop}%
\bibitem [{\citenamefont {Yang}(2006)}]{yang2006simple}%
  \BibitemOpen
  \bibfield  {author} {\bibinfo {author} {\bibfnamefont {D.}~\bibnamefont {Yang}},\ }\bibfield  {title} {\bibinfo {title} {A simple proof of monogamy of entanglement},\ }\href {https://doi.org/https://doi.org/10.1016/j.physleta.2006.08.027} {\bibfield  {journal} {\bibinfo  {journal} {Physics Letters A}\ }\textbf {\bibinfo {volume} {360}},\ \bibinfo {pages} {249} (\bibinfo {year} {2006})}\BibitemShut {NoStop}%
\bibitem [{\citenamefont {Arikan}(2009)}]{arikan2009polarcodes}%
  \BibitemOpen
  \bibfield  {author} {\bibinfo {author} {\bibfnamefont {E.}~\bibnamefont {Arikan}},\ }\bibfield  {title} {\bibinfo {title} {Channel polarization: A method for constructing capacity-achieving codes for symmetric binary-input memoryless channels},\ }\href {https://doi.org/10.1109/tit.2009.2021379} {\bibfield  {journal} {\bibinfo  {journal} {IEEE Transactions on Information Theory}\ }\textbf {\bibinfo {volume} {55}},\ \bibinfo {pages} {3051–3073} (\bibinfo {year} {2009})}\BibitemShut {NoStop}%
\bibitem [{\citenamefont {Şaşoğlu}\ \emph {et~al.}(2009)\citenamefont {Şaşoğlu}, \citenamefont {Telatar},\ and\ \citenamefont {Arikan}}]{sasoglu2009polarcodes}%
  \BibitemOpen
  \bibfield  {author} {\bibinfo {author} {\bibfnamefont {E.}~\bibnamefont {Şaşoğlu}}, \bibinfo {author} {\bibfnamefont {E.}~\bibnamefont {Telatar}},\ and\ \bibinfo {author} {\bibfnamefont {E.}~\bibnamefont {Arikan}},\ }\bibfield  {title} {\bibinfo {title} {Polarization for arbitrary discrete memoryless channels},\ }in\ \href {https://doi.org/10.1109/ITW.2009.5351487} {\emph {\bibinfo {booktitle} {2009 IEEE Information Theory Workshop}}}\ (\bibinfo {year} {2009})\ pp.\ \bibinfo {pages} {144--148}\BibitemShut {NoStop}%
\bibitem [{\citenamefont {Sutter}\ \emph {et~al.}(2012)\citenamefont {Sutter}, \citenamefont {Renes}, \citenamefont {Dupuis},\ and\ \citenamefont {Renner}}]{sutter2012polarcodesDMC}%
  \BibitemOpen
  \bibfield  {author} {\bibinfo {author} {\bibfnamefont {D.}~\bibnamefont {Sutter}}, \bibinfo {author} {\bibfnamefont {J.~M.}\ \bibnamefont {Renes}}, \bibinfo {author} {\bibfnamefont {F.}~\bibnamefont {Dupuis}},\ and\ \bibinfo {author} {\bibfnamefont {R.}~\bibnamefont {Renner}},\ }\bibfield  {title} {\bibinfo {title} {{Achieving the capacity of any DMC using only polar codes}},\ }in\ \href {https://doi.org/10.1109/itw.2012.6404638} {\emph {\bibinfo {booktitle} {2012 IEEE Information Theory Workshop}}}\ (\bibinfo  {publisher} {IEEE},\ \bibinfo {year} {2012})\ p.\ \bibinfo {pages} {114–118}\BibitemShut {NoStop}%
\bibitem [{\citenamefont {Renes}\ \emph {et~al.}(2012)\citenamefont {Renes}, \citenamefont {Dupuis},\ and\ \citenamefont {Renner}}]{renes2012quantumpolarcoding}%
  \BibitemOpen
  \bibfield  {author} {\bibinfo {author} {\bibfnamefont {J.~M.}\ \bibnamefont {Renes}}, \bibinfo {author} {\bibfnamefont {F.}~\bibnamefont {Dupuis}},\ and\ \bibinfo {author} {\bibfnamefont {R.}~\bibnamefont {Renner}},\ }\bibfield  {title} {\bibinfo {title} {Efficient polar coding of quantum information},\ }\href {https://doi.org/10.1103/PhysRevLett.109.050504} {\bibfield  {journal} {\bibinfo  {journal} {Phys. Rev. Lett.}\ }\textbf {\bibinfo {volume} {109}},\ \bibinfo {pages} {050504} (\bibinfo {year} {2012})}\BibitemShut {NoStop}%
\bibitem [{\citenamefont {Renes}\ and\ \citenamefont {Wilde}(2014)}]{renes2014quantumpolarcoding}%
  \BibitemOpen
  \bibfield  {author} {\bibinfo {author} {\bibfnamefont {J.~M.}\ \bibnamefont {Renes}}\ and\ \bibinfo {author} {\bibfnamefont {M.~M.}\ \bibnamefont {Wilde}},\ }\bibfield  {title} {\bibinfo {title} {Polar codes for private and quantum communication over arbitrary channels},\ }\href {https://doi.org/10.1109/tit.2014.2314463} {\bibfield  {journal} {\bibinfo  {journal} {IEEE Transactions on Information Theory}\ }\textbf {\bibinfo {volume} {60}},\ \bibinfo {pages} {3090–3103} (\bibinfo {year} {2014})}\BibitemShut {NoStop}%
\bibitem [{\citenamefont {Sheng}\ and\ \citenamefont {Deng}(2010{\natexlab{a}})}]{sheng2010deterministicpurification}%
  \BibitemOpen
  \bibfield  {author} {\bibinfo {author} {\bibfnamefont {Y.-B.}\ \bibnamefont {Sheng}}\ and\ \bibinfo {author} {\bibfnamefont {F.-G.}\ \bibnamefont {Deng}},\ }\bibfield  {title} {\bibinfo {title} {{Deterministic entanglement purification and complete nonlocal Bell-state analysis with hyperentanglement}},\ }\href {https://doi.org/10.1103/PhysRevA.81.032307} {\bibfield  {journal} {\bibinfo  {journal} {Physical Review A}\ }\textbf {\bibinfo {volume} {81}},\ \bibinfo {pages} {032307} (\bibinfo {year} {2010}{\natexlab{a}})}\BibitemShut {NoStop}%
\bibitem [{\citenamefont {Sheng}\ and\ \citenamefont {Deng}(2010{\natexlab{b}})}]{sheng2010deterministicpurification_b}%
  \BibitemOpen
  \bibfield  {author} {\bibinfo {author} {\bibfnamefont {Y.-B.}\ \bibnamefont {Sheng}}\ and\ \bibinfo {author} {\bibfnamefont {F.-G.}\ \bibnamefont {Deng}},\ }\bibfield  {title} {\bibinfo {title} {{One-step deterministic polarization-entanglement purification using spatial entanglement}},\ }\href {https://doi.org/10.1103/PhysRevA.82.044305} {\bibfield  {journal} {\bibinfo  {journal} {Physical Review A}\ }\textbf {\bibinfo {volume} {82}},\ \bibinfo {pages} {044305} (\bibinfo {year} {2010}{\natexlab{b}})}\BibitemShut {NoStop}%
\bibitem [{\citenamefont {Huang}\ \emph {et~al.}(2022)\citenamefont {Huang}, \citenamefont {Hu}, \citenamefont {Liu}, \citenamefont {Zhou}, \citenamefont {Sheng}, \citenamefont {Li},\ and\ \citenamefont {Guo}}]{huang2022deterministicpurification}%
  \BibitemOpen
  \bibfield  {author} {\bibinfo {author} {\bibfnamefont {C.-X.}\ \bibnamefont {Huang}}, \bibinfo {author} {\bibfnamefont {X.-M.}\ \bibnamefont {Hu}}, \bibinfo {author} {\bibfnamefont {B.-H.}\ \bibnamefont {Liu}}, \bibinfo {author} {\bibfnamefont {L.}~\bibnamefont {Zhou}}, \bibinfo {author} {\bibfnamefont {Y.-B.}\ \bibnamefont {Sheng}}, \bibinfo {author} {\bibfnamefont {C.-F.}\ \bibnamefont {Li}},\ and\ \bibinfo {author} {\bibfnamefont {G.-C.}\ \bibnamefont {Guo}},\ }\bibfield  {title} {\bibinfo {title} {Experimental one-step deterministic polarization entanglement purification},\ }\href {https://doi.org/10.1016/j.scib.2021.12.018} {\bibfield  {journal} {\bibinfo  {journal} {Science Bulletin}\ }\textbf {\bibinfo {volume} {67}},\ \bibinfo {pages} {593–597} (\bibinfo {year} {2022})}\BibitemShut {NoStop}%
\bibitem [{\citenamefont {Watrous}(2004)}]{watrous2004distillation}%
  \BibitemOpen
  \bibfield  {author} {\bibinfo {author} {\bibfnamefont {J.}~\bibnamefont {Watrous}},\ }\bibfield  {title} {\bibinfo {title} {Many {Copies} {May} {Be} {Required} for {Entanglement} {Distillation}},\ }\href {https://doi.org/10.1103/PhysRevLett.93.010502} {\bibfield  {journal} {\bibinfo  {journal} {Physical Review Letters}\ }\textbf {\bibinfo {volume} {93}},\ \bibinfo {pages} {010502} (\bibinfo {year} {2004})}\BibitemShut {NoStop}%
\bibitem [{\citenamefont {Brandão}\ and\ \citenamefont {Eisert}(2008)}]{brandao2008distillation}%
  \BibitemOpen
  \bibfield  {author} {\bibinfo {author} {\bibfnamefont {F.~G. S.~L.}\ \bibnamefont {Brandão}}\ and\ \bibinfo {author} {\bibfnamefont {J.}~\bibnamefont {Eisert}},\ }\bibfield  {title} {\bibinfo {title} {Correlated entanglement distillation and the structure of the set of undistillable states},\ }\href {https://doi.org/10.1063/1.2888925} {\bibfield  {journal} {\bibinfo  {journal} {Journal of Mathematical Physics}\ }\textbf {\bibinfo {volume} {49}},\ \bibinfo {pages} {042102} (\bibinfo {year} {2008})}\BibitemShut {NoStop}%
\bibitem [{\citenamefont {Waeldchen}\ \emph {et~al.}(2016)\citenamefont {Waeldchen}, \citenamefont {Gertis}, \citenamefont {Campbell},\ and\ \citenamefont {Eisert}}]{waeldchen2016renormalizingdistillation}%
  \BibitemOpen
  \bibfield  {author} {\bibinfo {author} {\bibfnamefont {S.}~\bibnamefont {Waeldchen}}, \bibinfo {author} {\bibfnamefont {J.}~\bibnamefont {Gertis}}, \bibinfo {author} {\bibfnamefont {E.~T.}\ \bibnamefont {Campbell}},\ and\ \bibinfo {author} {\bibfnamefont {J.}~\bibnamefont {Eisert}},\ }\bibfield  {title} {\bibinfo {title} {Renormalizing {Entanglement} {Distillation}},\ }\href {https://doi.org/10.1103/PhysRevLett.116.020502} {\bibfield  {journal} {\bibinfo  {journal} {Physical Review Letters}\ }\textbf {\bibinfo {volume} {116}},\ \bibinfo {pages} {020502} (\bibinfo {year} {2016})}\BibitemShut {NoStop}%
\bibitem [{\citenamefont {Webb}(2016)}]{webb2016clifford}%
  \BibitemOpen
  \bibfield  {author} {\bibinfo {author} {\bibfnamefont {Z.}~\bibnamefont {Webb}},\ }\bibfield  {title} {\bibinfo {title} {{The Clifford group forms a unitary 3-design}},\ }\href {https://doi.org/10.26421/QIC16.15-16-8} {\bibfield  {journal} {\bibinfo  {journal} {Quantum Information and Computation}\ }\textbf {\bibinfo {volume} {16}},\ \bibinfo {pages} {1379} (\bibinfo {year} {2016})}\BibitemShut {NoStop}%
\bibitem [{\citenamefont {Karlin}\ and\ \citenamefont {McGregor}(1957)}]{karlin1957classification}%
  \BibitemOpen
  \bibfield  {author} {\bibinfo {author} {\bibfnamefont {S.}~\bibnamefont {Karlin}}\ and\ \bibinfo {author} {\bibfnamefont {J.}~\bibnamefont {McGregor}},\ }\bibfield  {title} {\bibinfo {title} {The classification of birth and death processes},\ }\href {https://doi.org/10.1090/S0002-9947-1957-0094854-8} {\bibfield  {journal} {\bibinfo  {journal} {Transactions of the American Mathematical Society}\ }\textbf {\bibinfo {volume} {86}},\ \bibinfo {pages} {366} (\bibinfo {year} {1957})}\BibitemShut {NoStop}%
\end{thebibliography}%

\onecolumngrid
\clearpage\newpage


\begin{center}
    {\large\bfseries End Matter}
\end{center}
\bigskip
\twocolumngrid

\begin{lemma}[Acceptance probabilities of \cref{alg:random_bilocal_Clifford}, passive setting]\label{lem:probs}
For an input state with initial fidelity $f^n$, the acceptance probability of \cref{alg:random_bilocal_Clifford} in the passive setting is
\begin{equation}\label{eq:pr_accept_2}
p_{\acc}=\frac{f^n-4^{-n}}{1-4^{-n}} + \frac{1-f^n}{1-4^{-n}} \frac{2^m \cdot 4^k}{4^n}.
\end{equation}
The probability of acceptance \emph{and} obtaining the target state is 
\begin{equation}\label{eq:pr_accept_and_phi_2}
p_{\acc \land \Phi}=\frac{f^n-4^{-n}}{1-4^{-n}}+\frac{1-f^n}{1-4^{-n}}\frac{2^{m}}{4^n}.
\end{equation}
\end{lemma}
\begin{proof}
    Using the fact that the Clifford group forms a unitary 2-design~\cite{DCEL09,DAE07}, it is known that~\cite[Theorem~7.25]{Wat18_book}
    \begin{equation}\label{eq:Clifford_twirl_2}
        \E_{C}[C^{1,1}\rho_{A^nB^n}C^{1,1\dag}]=f^n\Phi_{AB}^{\otimes n}+\frac{1-f^n}{4^n-1}(\id_{AB}^{\otimes n}-\Phi_{AB}^{\otimes n}),
    \end{equation}
    where $C^{1,1}\equiv C\otimes C^{*}$. The first term in the expressions in \eqref{eq:pr_accept_2} and \eqref{eq:pr_accept_and_phi_2} represents the contribution from the maximally entangled component, while the second term captures the probability of accidentally accepting despite having errors. For details, we refer to the Supplemental Material.
\end{proof}
The above expressions lead to a simple but powerful lower bound on the output fidelity, $(1-\bar{\varepsilon})^k=\frac{p_{\acc\wedge\Phi}}{p_{\acc}}$.

\begin{corollary}[Output fidelity of \cref{alg:random_bilocal_Clifford}, passive setting]\label{cor:simple-f}
For every input state with fidelity at least $f^n \coloneqq (1-\varepsilon)^n$, the output fidelity when the protocol succeeds satisfies
\begin{equation}\label{eq:new-fid_EM}
    (1-\bar{\varepsilon})^k \geq 1 - 2^{-m}(f^{-n}-1).
\end{equation}
\end{corollary} 

\begin{proof}
 The proof follows from the straightforward lower bound $p_{\acc \land \Phi} \geq f^n$ and the upper bound $p_{\acc} \leq f^n + (1 - f^n) 2^{-m}$. The second addend in $p_{\acc}$ represents a \textit{false positive} event, where the input state contains some error, yet the protocol fails to detect it; see the Supplemental Material for details.
\end{proof}

\paragraph*{Finite-depth protocol with noisy gates.} Consider the following protocol. Alice and Bob share $n$ qubit pairs. They repeatedly choose a random pair of qubit positions, $(i,j)$, and Alice applies a random two-qubit Clifford unitary $C$ to qubits $i$ and $j$, while Bob applies its complex conjugate $C^{\ast}$ to qubits $i$ and $j$. After some number $G$ of gates, they both measure $m$ of their qubits in the computational basis and accept if they have matching outcomes, just like the passive setting of \cref{alg:random_bilocal_Clifford}. We model the noise in the system as acting on the qubits to which each gate is applied. To evaluate the protocol's performance, we compute both the success probability and the fidelity of the output state when the protocol succeeds. 
This involves the construction of a Markov chain based on the Clifford transformation of Pauli probabilities; see the Supplemental Material for details. 

Our analysis reveals an important trade-off in noisy finite-depth protocols. While more gates are necessary for scrambling, they also introduce more noise into the system. In \cref{fig:finite_depth}, we see that there is a value of $G$ beyond which the infidelity does not improve. This is the point at which noise accumulates to an extent that cannot be further counteracted with our distillation protocol. This optimal value of $G$ depends on both the noise parameters and the desired output fidelity.

\begin{table*}[t]
\begin{tabular}{@{}p{3.0cm} p{2.7cm} @{\hskip 0.1in} p{2.0cm} p{2.8cm} p{3.0cm} @{\hskip 0.15in} p{3.0cm}@{}}
\toprule
\hfil \multirow{2}{*}{\textbf{Protocol}} \hfill & \hfil \multirow{2}{*}{\textbf{Main idea}} \hfill & \hfil \multirow{2}{*}{\textbf{Overhead}} \hfill & \hfil \multirow{2}{*}{\textbf{Robustness}} \hfill & \multicolumn{2}{c}{\textbf{Resource requirements}} \\ \cmidrule(lr){5-6} 
& & &                                     & \multicolumn{1}{c}{\emph{Quantum}}             &  \multicolumn{1}{c}{\emph{Classical}}            \\ \cmidrule{1-6}\morecmidrules\cmidrule{1-6}

BBPSSW-1EPP~\cite{BDSW96}    & Hashing based on random classical linear codes                              & $O(1)$                                   &  Robust with modifications \cite{zwerger2014robustness}                                    & $O(n^2)$ CNOT gates                & One-way CC, exponentially costly decoding. \\ 
Decoupling \cite{hayden2008decoupling,dupuis2014decoupling}    & Random quantum codes                              & $O(1)$                                   &  Requires injection                                    & $O(n \poly\log n)$ Clifford gates                & One-way CC, exponentially costly decoding.               \\ 
Lattice surgery~\cite{sinclair2024fault,ramette2023fault,fowler2010surface} & Lattice surgery to distribute logical Bell pairs & $O(\poly \log \bar{\varepsilon}^{-1})$ (threshold $\varepsilon_0 \lesssim 10\%$) & Natively fault-tolerant & $O(\poly \log \bar{\varepsilon}^{-1})$ space-time & One-way CC, $O(\poly n)$ decoding. \\
qLDPC~\cite{ataides2025constant} & High-rate qLDPC codes & $O(1)$ (threshold $\varepsilon_0 \lesssim 5\%$) & Natively fault-tolerant & $O(\log \bar{\varepsilon}^{-1})$ space and $O(1)$ time & One-way CC and $O(\poly n)$ decoding. \\
\hline
BBPSSW-2EPP \cite{BDSW96}  &  $2 \to 1$ recurrence & $O(\poly \bar{\varepsilon}^{-1})$                                   & Requires injection & $O(\poly \bar{\varepsilon}^{-1})$ space-time & Two-way CC                \\
DEJMPS \cite{DEJMPS96}     & $2 \to 1$ recurrence                            & $O(\poly \bar{\varepsilon}^{-1})$                                   & Requires injection & $O(\poly \bar{\varepsilon}^{-1})$  space-time & Two-way CC                \\
Entanglement pumping \cite{DBC99,dur2003entanglement} & $2 \to 1$ iterative pumping & $O(\poly \bar{\varepsilon}^{-1})$ & Requires injection & $O(\poly \bar{\varepsilon}^{-1})$ space-time & Two-way CC \\
Bell basis permutations \cite{dehaene2003distillationpermutation,maneva2002purification,BM05,KAJ19,jansen2022enumeratingclifforddistillation,addala2023optimizedpurification,goodenough2023ntokdistillation} & $n \to k$ recurrence & $O(\poly \bar{\varepsilon}^{-1})$ & Requires injection & $O(\poly \bar{\varepsilon}^{-1})$ space-time & Two-way CC \\
\citet{pattison2024constoverheaddistillation} & Concatenating small error detecting codes & $O(1)$ & Requires injection & $O(\poly \log \varepsilon^{-1})$ space-time & Two-way CC \\
\hline
\textbf{Ours} & \textbf{Random bilocal Cliffords} & $\boldsymbol{O(1)}$ & \textbf{Requires injection} &  $\boldsymbol{O(\poly \log \bar{\varepsilon}^{-1})}$ \textbf{space,} $\boldsymbol{O(\poly \log \log \bar{\varepsilon}^{-1})}$ \textbf{time} & \textbf{Two-way CC} \\
\bottomrule
\end{tabular}
\caption{Overview of prominent entanglement distillation protocols.
The overhead column describes the number of input Bell pairs required to achieve a target infidelity of $\bar{\varepsilon}$. For error-correction-based protocols, a threshold input infidelity $\varepsilon_0$ is indicated; running these protocols on Bell pairs with an infidelity higher than $\varepsilon_0$ will only \emph{increase} the infidelity. In contrast, general error-detection-based protocols always output Bell pairs with improved infidelity compared to the input. The robustness column indicates tolerance to noisy \emph{local} operations. Some protocols provide native fault-tolerance by using the same code to protect against both network and local noise. However, others are designed only to protect against network noise; protecting against local noise requires that the physical Bell pairs be injected (non-fault-tolerantly) into some secondary code, on top of which the distillation protocol is run.
Scalings in $n$ instead of $\varepsilon$ appear for asymptotic protocols whose performance guarantees rely on large input sizes to leverage concentration of measure effects. Such scalings are valid when $n$ is sufficiently large, with the relationship between $n$ and the target infidelity $\bar{\varepsilon}$ often implicit.}
\label{tab:comparisons}
\end{table*}

\begin{figure}
    \centering
    \includegraphics[width=0.75\columnwidth]{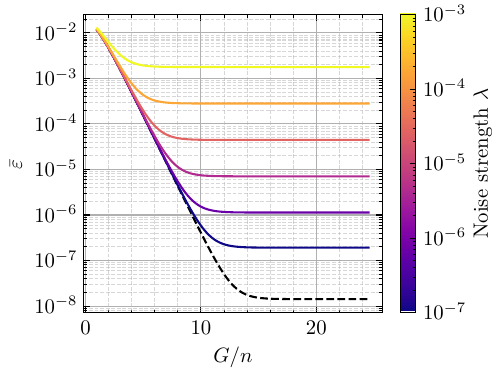}
    \caption{The achieved infidelity of our random bilocal Clifford protocol with a finite number $G$ of gates and varying two-qubit gate depolarizing strengths $\lambda$, with the noiseless case shown in the dashed line. We set the number of noisy Bell pairs to $n=30$ and the initial infidelity to $\varepsilon=0.02$. Note the infidelity for local noise rates $\lambda$ asymptotes to~$\sim\!\lambda$.}
    \label{fig:finite_depth}
\end{figure}

 \paragraph*{Details on the application to quantum repeaters.}
For simplicity, we consider a repeater chain using a \emph{nested scheme}; see, e.g., Ref.~\cite{DBC99}. In a (homogeneous) nested repeater scheme, two parties are connected by $2^T$ (where $T\in \mathbb{N}_0$) equidistant segments of optical fiber, connected to each other by processing nodes, i.e., repeaters, capable of generating and distilling entanglement. It is possible to create nearest-neighbor entanglement in the chain, and then perform entanglement swapping at the intermediate nodes to create longer-range entanglement~\cite{SSR+11,azuma2023repeatersRMP}. Distillation is then performed to counteract noise, which is caused by either the initial entanglement generation, or noise accrued from swapping noisy states. Nested repeater schemes perform swaps and distillations symmetrically. Specifically, at each \emph{level} of the protocol, two states spanning an identical distance are swapped. Furthermore, for a given level the same distillation protocol is applied to each of the states. This significantly narrows down the protocol space and is known to perform well for homogeneous repeater chains (although not always optimal; see, e.g., Refs.~\cite{shchukin2022optimal,haldar2024reducing}). Even with nested repeater schemes, there are still many ways to apply our distillation protocols. Here, we only consider a single round of distillation at each level. This choice is motivated not only out of simplicity, but also to showcase the capabilities of random bilocal Clifford protocols without needing to perform complicated optimizations. As such, the parameters $(n_i, k_i)$ at each level $i \in \left[T\right]$ specify a nested repeater scheme using random bilocal Clifford distillation protocols. In the following, we further narrow down the search space. Finally, we restrict ourselves to the passive setting and to at most $T=9$ levels.

Motivated by the observation that distillation is most effective when applied to weakly entangled states prior to swapping (see, e.g., Ref.~\cite{haldar2024reducing}), we employ the following heuristic to determine the values of $(n_i, k_i)$. First, for each $2\leq n\leq 100$, where $n$ is the number of entangled qubit pairs, we find the smallest value of $k$ that achieves an infidelity of $10^{-9}$ after distillation. 
For each such pair $(n, k)$, we numerically find the one that minimizes the overhead (over a single segment), which gives us the first distillation protocol in our sequence. Subsequent swaps will roughly double the infidelity~\cite{khatri2022networkMDP}, and as such should be counteracted by distillation. Since the states are already of high quality, we select the remaining parameters as $(n', n'-1)$ for some $n'$, which allows for rapid optimization. Our heuristic protocol is thus only specified by three parameters: $n$, $k$, and $n'$. We obtain optimized values of $n=93$, $k=68$ and $n' = 40$. We note here that an optimization using the protocols from~\citet{pattison2024constoverheaddistillation}, where further distillation between levels is allowed, would be a more fair comparison.

For QKD, we brute-force search over all protocols that at the first level distill with some value of $4\leq n\leq 30$ ($4\leq n \leq 12$ for the dashed lines in \cref{fig:applications_baseline}) and $1 \leq k\leq n-1$. For the remaining levels we distill with parameters set to $n, k=n-1$ or $n, k=n-2$. Similar to the heuristic for long-baseline optical interferometry, this is far from an exhaustive search, but still yields a high secret-key rate.

\clearpage
\newpage
\let\addcontentsline\oldaddcontentsline
\title{Supplemental Material: Constant Overhead Entanglement Distillation via Scrambling}
\maketitle

\onecolumngrid


\tableofcontents

\section{Prior Art}
Entanglement distillation has historically been approached from two complementary perspectives, each providing distinct insights into the fundamental nature of the problem. The first perspective, rooted in quantum error correction~\cite{BDSW96,gottesman1997stabilizer,steane1997error}, views distillation as a distributed error correction task where Alice and Bob cooperate to extract high-fidelity entangled states from noisy ones. This approach has led to concrete, practical protocols based on quantum codes and stabilizer measurements. The second perspective, emerging from quantum Shannon theory~\cite{schumacher1995quantum,BDSW96,devetak2005private}, treats distillation as a quantum information processing task and focuses on establishing fundamental limits on achievable rates and fidelities. While these perspectives were initially developed independently, their unification has provided deep insights into both the theoretical limits and practical implementations of entanglement distillation protocols. Parallel to these developments, significant progress has emerged from analyzing the mathematical structure of local operations on Bell states~\cite{dehaene2003distillationpermutation,maneva2002purification,BM05}, leading to systematic frameworks for optimizing practical protocols~\cite{RST+18,KAJ19,jansen2022enumeratingclifforddistillation,addala2023optimizedpurification,goodenough2023ntokdistillation}. Many of the earlier developments are thoroughly covered in comprehensive reviews~\cite{dur2007purificationQECreview,dur2016distillationreview}, which trace the evolution of the field from early concatenated protocols through to modern implementations. Our work draws on both traditions, using tools from quantum error correction to construct explicit protocols while leveraging information-theoretic techniques to analyze their performance.

\paragraph*{Quantum error correction.} The deep connection between entanglement distillation and quantum error correction was first highlighted by~\citet{BDSW96}, who introduced a unified perspective on quantum coding and distillation. They demonstrated that any quantum error-correcting code could be used to derive an entanglement distillation protocol, and vice versa. This equivalence arises from their shared objective: countering the effects of noise while preserving quantum information. Quantum error correction encodes quantum information into a larger Hilbert space to protect it from local noise, while distillation extracts high-fidelity entangled states from noisy copies by selectively discarding errors.

This connection was further developed through the stabilizer formalism. \citet{gottesman1997stabilizer} showed that stabilizer codes naturally give rise to entanglement distillation protocols in which the encoding and decoding operations become local operations performed by Alice and Bob. The CSS code construction proved particularly fruitful for distillation, as it allowed separate handling of bit and phase errors~\cite{steane1997error,calderbank1996good}. These ideas led to early practical protocols like the BBPSSW protocol~\cite{BBP96}, which implements a simple concatenated scheme alternating between bit and phase error correction using a classical $[\![2,1,2]\!]$ repetition code. Similar concatenated approaches using small quantum codes were developed~\cite{murao1998multiparticlepurification,DEJMPS96,horodecki1999reduction,alber2001efficient}, but these protocols achieve asymptotically zero rates --- the number of input pairs needed per output pair grows without bound as the target infidelity $\varepsilon$ decreases. More sophisticated schemes were subsequently developed~\cite{matsumoto2003conversion,hostens2004stabilizer}, culminating in recent works that achieve asymptotically constant rates~\cite{pattison2024constoverheaddistillation,shi2024stabilizer} --- meaning the expected number of input Bell pairs required per output pair remains bounded even as $\varepsilon$ approaches zero. 

While early work focused on error correction, subsequent research explored probabilistic approaches that trade success probability for improved performance. Error detecting codes, which only flag the presence of errors rather than attempting to correct them, proved particularly valuable when adapted to entanglement distillation. Indeed, it has been shown that protocols based on error detection can achieve significantly better fidelity than traditional error-correcting schemes in high-noise regimes~\cite{aschauer2005quantum,fujii2009entpurifdoubleselection}. Recent work has also shown that error detection-based protocols can be highly space-efficient: \citet{gidney2023tetrationally} demonstrated that entanglement can be purified to infidelity $\bar{\varepsilon}$ using only $O(\log^* (1/\bar{\varepsilon}))$ qubits of storage (where $\log^*$ is the iterated logarithm) through careful staging of error detection rounds. These improvements, however, come at a key cost, which is that two-way classical communication become necessary. This is because Bob must transmit his measurement outcomes to Alice, who then decides whether to accept or reject the state before communicating this decision back to Bob. Despite this, classical communication is generally considered to be cheap, so the ability to tolerate higher rates of transmission error has made error detection a useful technique for realizing practical entanglement distillation protocols.

The use of randomness in error correction has a rich history dating back to Shannon's original work on classical error correction~\cite{shannon1948mathematical}. In the quantum setting, random stabilizer codes are known to saturate various quantum capacity bounds~\cite{gottesman1997stabilizer}, and recent work has shown that even shallow random quantum circuits can define effective error correcting codes~\cite{brown2023short,gullan2021quantum,darmawan2024low,nakata2021one,nelson2023faulttolerantrandom}. Our work connects to this tradition, but with a crucial difference: rather than attempting to decode random stabilizer codes (which can be computationally intensive~\cite{mceliece1978public,alekhnovich2003more}), we use them purely for error detection. The random Clifford operations in our protocol effectively implement random stabilizer codes, spreading errors in a way that makes them detectable through local measurements. While this error detection approach means our protocol succeeds only probabilistically, we show that it can nonetheless achieve asymptotically constant rates in expectation; furthermore, limiting the failure probability to $\delta$ only incurs a multiplicative $O(\log \delta^{-1})$ overhead.

The error correction perspective helps explain several features of our protocol. The improvement in fidelity can be understood as arising from the distance properties of random stabilizer codes --- with high probability, these codes have good minimum distance, allowing them to detect most error patterns. The trade-off between success probability and output fidelity mirrors similar trade-offs in error detection codes. Perhaps most importantly, the constant overhead we achieve connects to the fact that random stabilizer codes can approach the quantum Gilbert--Varshamov bound~\cite{calderbank1997quantum,gottesman1997stabilizer,calderbank1996good}. 

\paragraph*{Random unitaries and decoupling.} The use of randomness in information theory has been foundational in understanding the limits of communication, both classical and quantum. In classical Shannon theory, random coding was first introduced to demonstrate the achievability of channel capacity, marking a major breakthrough in the study of communication~\cite{shannon1948mathematical}. By generating random codes, Shannon proved that reliable communication over noisy classical channels was possible, and this idea became a cornerstone of classical information theory, ultimately inspiring the development of practical capacity-achieving codes~\cite{gallager1986information, mceliece2003theory,cover_thomas_book}. 

The power of randomness in entanglement distillation was recognized early on by~\citet{BDSW96}, who introduced the hashing protocol; see also Refs.~\cite{vollbrecht2005interpolation,zwerger2014robustness}. This protocol uses random classical linear codes to extract maximally entangled states from multiple copies of a partially entangled state, much as classical hashing functions can extract pure random bits from partially random sources. The key insight was that measuring random combinations of Bell pairs could reveal information about their joint error syndrome, allowing Alice and Bob to identify and correct errors. Like many protocols based on random coding arguments, however, the hashing protocol requires a decoding procedure which has an exponential computational cost. This limitation highlights a recurring theme in entanglement distillation: while random coding arguments can establish achievable rates in principle, translating these ideas into practical protocols often requires significant modifications.

This interplay between classical and quantum randomness has proved remarkably fertile. The quantum generalization of these ideas has led to deep insights about both quantum communication and entanglement distillation~\cite{Shor95errorcorrection,schumacher1995quantum, schumacher1996entanglementnoisychannels, lloyd1997capacity,shor2002quantum, devetak2005private, devetak2004father, abeyesinghe2009mother}. While these tasks might appear different --- one involving the transmission of quantum states through noisy channels, the other the purification of shared entanglement --- they are fundamentally connected through the role of quantum correlations. Quantum communication naturally creates entanglement between sender and receiver, while shared entanglement enables quantum communication through teleportation~\cite{BBC+93,BBP96}. The power of randomness in analyzing these quantum protocols mirrors its classical counterpart: just as random codes reveal the capacity of classical channels, random quantum operations illuminate the fundamental limits of quantum information processing.

A major breakthrough in understanding this connection came from \citet{hayden2008decoupling}, who showed how random unitaries could provide a unified framework for both tasks. An important idea used in this approach was to describe the effects of a noisy channel as interactions with an auxiliary system called the ``environment''. This perspective is a key component of quantum Shannon theory, where the goal is often to suppress these correlations so that the transmitted information is reliably transferred to the receiver without being lost to the environment. The key insight of the decoupling approach is that if Alice applies a \emph{random} unitary to her system, she effectively scrambles any correlations between her system and the environment. More precisely, after applying a random unitary, the reduced state of the environment becomes nearly maximally mixed --- it retains essentially no information about Alice's original state. This ``decoupling'' from the environment implies, by the monogamy of entanglement~\cite{terhal2004entanglement, yang2006simple}, that any quantum correlations must now reside between Alice and Bob's systems. By leveraging this property, the decoupling approach enables the distillation of entanglement at rates approaching the coherent information, a quantity that is both a lower bound for quantum capacity and a lower bound for distillable entanglement~\cite{BDSW96,lloyd1997capacity,devetak2005private,devetak2005distillation,leditzky2018useful}.

However, achieving these rates in practice requires projecting onto the typical subspaces of Alice's and Bob's systems. These subspaces are defined as the regions where the state's entropic properties concentrate around their asymptotic values. Following this projection, random unitaries and recovery operations can distill maximally entangled states at the coherent information rate. The decoupling approach thus establishes the coherent information as not just a bound but a rate that can be (asymptotically) achieved in principle. However, the operations required for these projections --- such as identifying the subspaces and performing spectral decompositions --- are computationally and experimentally prohibitive. While theoretically elegant, such operations lack practical feasibility, creating a gap between theory and implementation that has been addressed through strategies such as polar codes~\cite{arikan2009polarcodes,sasoglu2009polarcodes,sutter2012polarcodesDMC,renes2012quantumpolarcoding,renes2014quantumpolarcoding}.

Our work shows that the typical subspace measurements of random coding procedures can be replaced with much simpler operations while still achieving asymptotically constant rates of entanglement distillation. Like the decoupling approach, random unitaries feature prominently in our protocol. However, rather than requiring typical subspace measurements, we show that random unitary operations followed by simple computational basis measurements suffice. While we cannot perform the optimal measurements required by the decoupling approach, random Clifford operations spread errors in a way that makes them easily detectable through local measurements, effectively decoupling the preserved information from the noise. Our random bilocal Clifford operations show achievability of favorable (and asymptotically constant) entanglement distillation rates, while maintaining efficiency of implementation through the structure of the Clifford group. This provides a more direct route to achieving practically feasible rates compared to protocols requiring typical subspace measurements~\cite{schumacher1995quantum,devetak2005distillation,hayden2008decoupling}.

\section{Entanglement distillation formalism}

An entanglement distillation protocol is a bipartite channel $\mathcal{L}$ consisting of local quantum operations between two parties, Alice and Bob, and classical communication between them (LOCC), which maps noisy copies of the maximally entangled state vector $\ket{\Phi}\coloneqq\frac{1}{\sqrt{d}}\sum_{k=0}^{d-1}\ket{k,k}$ to fewer, less noisy copies; we refer to Ref.~\cite{dur2016distillationreview} for a pedagogical introduction. We also let $\Phi\equiv\ketbra{\Phi}{\Phi}$ throughout. Intuitively, the pairs of entangled qubits they share have been corrupted by noise during transmission through a quantum channel. Their goal is to sacrifice some of these pairs to increase the quality of the remaining ones, using only LOCC.

The key challenge in designing entanglement distillation protocols is that quantum measurements are inherently probabilistic and irreversible. Unlike classical error correction, where we can simply measure a syndrome to detect and correct errors, quantum protocols must carefully balance information gain against state disturbance. This leads naturally to probabilistic protocols --- ones that might fail and need to be repeated, but which preserve quantum coherence when they succeed. To make this concrete, imagine Alice and Bob share $n$ noisy Bell pairs. They each perform some local operations and measurements on their qubits, then compare their results over a classical channel. Based on these results, they either declare success and keep $k < n$ purified pairs, or declare failure and start over with fresh pairs.

This probabilistic nature means that we need to carefully track both the quality of the output states when we succeed and how many input states we expect to consume before succeeding. In this work, we are particularly interested in LOCC algorithms that have some \emph{failure probability}, and therefore may need to be run arbitrarily many times in order to guarantee success. To understand the performance of such protocols, we need to account for this randomness due to the possibility of failure, which leads to the following more formal definition of an entanglement distillation protocol.

\begin{definition}[Entanglement distillation]\label{def:ent-distill}
Entanglement distillation protocols are parameterized by $5$ pieces of data: $(\varepsilon,\bar{\varepsilon},n,k,\mathcal{L})$, where $n,k \in \mathbb{Z}_{\geq 1}$, $\varepsilon, \bar{\varepsilon} \in (0,1)$, and $\mathcal{L}:\Lin(\mathcal{H}_A^{\otimes n}\otimes\mathcal{H}_B^{\otimes n})\to\Lin(\mathcal{H}_A^{\otimes k}\otimes\mathcal{H}_B^{\otimes k})$ is an LOCC channel that accepts as input an arbitrary bipartite $2n$-qubit state $\rho_{A^nB^n}$, which can be assumed to satisfy $\bra{\Phi}_{AB}^{\otimes n}\rho_{A^nB^n} \ket{\Phi}_{AB}^{\otimes n} \geq (1-\varepsilon)^{n}$, and outputs a bipartite $2k$-qubit state. This protocol has some acceptance probability $p_\acc \in [0,1]$, such that:
\begin{itemize}[noitemsep]
    \item With probability $p_\acc$, the protocol outputs a state $\bar{\rho}_{A^kB^k} \in \Density(\mathcal{H}_A^{\otimes k}\otimes\mathcal{H}_B^{\otimes k})$, which has fidelity
    \begin{equation}\label{eq:condition}
        \bra{\Phi}_{AB}^{\otimes k}\bar{\rho}_{A^kB^k}\ket{\Phi}_{AB}^{\otimes k}\geq (1-\bar{\varepsilon})^k.
    \end{equation}
    \item With probability $1-p_\acc$, the protocol outputs some fiducial state $\ket{\mathsf{FAIL}}$, indicating that the protocol has failed.
\end{itemize}
\end{definition}

\begin{remark}
The above definition explicitly allows for probabilistic protocols that can output a fiducial failure state. This is in contrast to deterministic protocols, which do not allow for the possibility of such a fiducial state. Probabilistic protocols make an important distinction between two types of errors: heralded failures, where the protocol explicitly outputs $\ket{\mathsf{FAIL}}$, and unheralded errors that contribute to the infidelity of the output state $\rho'$ with respect to the target state. This distinction is well-motivated practically: heralded failures can be addressed by simply rerunning the protocol, as Bell pairs are relatively cheap and fungible resources. In contrast, unheralded errors that escape detection are far more dangerous, particularly in applications like distributed quantum computing, where extremely low infidelities ($\sim 10^{-12}$) are required. Deterministic protocols (see, e.g., Refs.~\cite{RST+18,sheng2010deterministicpurification,sheng2010deterministicpurification_b,huang2022deterministicpurification}), by absorbing all errors directly into the output state, necessarily increase the unheralded error rate. Indeed, any probabilistic protocol can be understood to be a deterministic one that simply outputs the mixed state $\mathcal{L}(\rho) = p_\acc \rho' + (1-p_\acc) \ketbra{\mathsf{FAIL}}$. This new state necessarily has lower fidelity with respect to the desired state compared to $\rho'$.
\end{remark}

One can understand the states resulting from an entanglement distillation protocol as `error corrected' (or `error suppressed') entangled states. Indeed, this is more than a heuristic analogy. An early seminal work~\cite{BDSW96} established that this connection goes in both directions, in the sense that error correction protocols can be used to define entanglement distillation protocols, and vice versa. In keeping with this tradition, we will often use the overbar notation $\bar{x}$ to denote the `logical' version of $x$. Put another way, in the language of entanglement distillation, if $x$ is some property of some noisy input Bell pairs, $\bar{x}$ denotes  the value of property $x$ for the output (distilled) Bell pairs.

\cref{def:ent-distill} includes several important concepts. The key performance metrics that we care about in practice are:
\begin{enumerate}[noitemsep]
    \item \emph{Fidelity improvement}: The protocol takes states of fidelity $(1-\varepsilon)^{n}$ and produces states of fidelity $(1-\bar{\varepsilon})^k$. The relative decrease of $\bar{\varepsilon}$ compared to $\varepsilon$ tells us how much the fidelity of the input improves.
    \item \emph{Resource efficiency}: The ratio $n/k$ represents how many input pairs we sacrifice to produce each output pair in a successful round.
    \item \emph{Success probability}: The parameter $p_\acc$ bounds how often we need to repeat the protocol before succeeding. This directly impacts the total number of input pairs needed.
\end{enumerate}
These figures of merit are not independent --- there are fundamental trade-offs between them that any protocol must navigate. For instance, we can often improve output fidelity by measuring more qubits for error detection, but this reduces both $k$ and the success probability. The art of protocol design lies in finding sweet spots in these trade-offs that are practically useful. To capture these tradeoffs, we define the `overhead' of a protocol.

\begin{definition}[Overhead of an entanglement distillation protocol]\label{def:overhead}
Consider an entanglement distillation protocol parameterized by $(\varepsilon,\bar{\varepsilon},n,k,\mathcal{L})$. Since the protocol is probabilistic, it must be rerun $R$ times to achieve success, where $R$ is a random variable distributed according to the geometric distribution with parameter $p_\acc$. We define a random variable $N = R \cdot n$, which represents the total number of input Bell pairs required to achieve success. We then define the \emph{overhead} of the entanglement distillation protocol as the random variable $\mathcal{O} \coloneqq \frac{N}{k}$, and the rate as $r \coloneqq \mathcal{O}^{-1}$. The expected overhead of the protocol is
\begin{equation}
    \E[\mathcal{O}] = \frac{n}{k \cdot p_\acc}.
\end{equation}
\end{definition}

The random variable $N$ in our definition captures the interplay between success probability and resource costs. It represents the total number of input pairs we need before finally succeeding, accounting for all failed attempts. The expected overhead $\E[\mathcal{O}]$ then gives us a single figure of merit that combines all of these factors --- it tells us the expected number of input pairs needed per output pair, including any pairs consumed in failed attempts. We contrast this with the work of~\citet{regula2023probabilistic}, which considers probabilistic entanglement distillation protocols in the information-theoretic (asymptotic) setting. In this setting, the protocol can be repeated arbitrarily many times until success, without incurring any cost. Consequently, their definition of the rate~\cite[Eq.~(3)]{regula2023probabilistic} does not take the success probability explicitly into account. On the other hand, we explicitly take the success probability into account in our definition of the rate, and thus account for the cost of repeating the protocol. We also comment that Ref.~\cite{watrous2004distillation} has previously considered the number of copies of an input state needed for distilling entanglement in the setting of deterministic protocols.

\begin{remark}[Non-IID states]
    In both \cref{def:ent-distill} and for our main results, the input states need not be an ``independent and identically distributed'' (IID) tensor-product state $\rho_{AB}^{\otimes n}$, which is what one might think of when encountering the phrase ``$n$ noisy Bell pairs''. Nor do the states need to be a tensor product of different bipartite states. The input can be an arbitrary bipartite $2n$-qubit state $\rho_{A^nB^n}$. In particular, this state could result from arbitrary noise applied to $n$ (noiseless) Bell pairs $\Phi_{AB}^{\otimes n}$. This noise could be correlated among the copies, and thus the resulting state need not have a tensor-product form. In this sense, when we speak of ``$n$ noisy Bell pairs'' in the context of the distillation task, we mean an arbitrary bipartite $2n$-qubit state. This is particularly relevant when considering concatenation of our protocol, because of the fact that the output states of previous layers of concatenation---which are the inputs to the next layer of concatenation---will in general have correlations between the qubit pairs. Our work thus follows prior works~\cite{buscemi2010distilling,brandao2008distillation,waeldchen2016renormalizingdistillation}, which have also considered entanglement distillation with correlated input states.
\end{remark}


\subsection{Concatenated protocols} 

While a single round of entanglement distillation can improve fidelity, achieving very high fidelities requires multiple rounds. The natural approach is to \textit{concatenate} the protocol --- that is, to repeatedly apply distillation to the output of previous distillation rounds. In this section, we develop a simple framework for analyzing such concatenated protocols.

Consider a distillation protocol that transforms $n$ noisy Bell pairs with infidelity $\varepsilon$ into $k$ purified pairs with infidelity $\bar{\varepsilon}$, succeeding with probability $p_{\acc}$. As discussed in \cref{def:overhead}, the expected overhead $\E[\mathcal{O}]$ of this protocol is defined as the expected number of input pairs needed per output pair: $\E[\mathcal{O}] = \frac{n}{p_{\acc} \cdot k}$. 


Now, consider concatenating this protocol $L$ times. Each layer $i\in\{1,2,\dotsc,L\}$ implements its own instance of the protocol, potentially with different parameters. The first layer consumes $n_1$ noisy pairs to produce $k_1$ less noisy pairs with probability $p_{\acc,1}$. The second layer then takes $n_2$ of these partially-purified pairs and produces $k_2$ pairs of even higher fidelity with probability $p_{\acc,2}$, and so on. At each stage, the output infidelity $\bar{\varepsilon}_i$ becomes the input infidelity $\varepsilon_{i+1}$ for the next layer. At the very end of this concatenated protocol, we aim to produce $k_L$ pairs with a desired infidelity of $\bar{\varepsilon}_L$. 
Let us now show that the overhead of such a concatenated protocol can be expressed simply in terms of the overheads of the individual layers.

\begin{lemma}
    The expected overhead of a concatenated protocol with $L\in\{1,2,\dotsc\}$ layers is given by the product of the overheads of the individual layers:
    \begin{equation}\label{eq:expected_overhead_concatenated_protocol}
        \E[\mathcal{O}] = \prod_{i=1}^L \E[\mathcal{O}_i] = \prod_{i=1}^L \frac{n_i}{k_ip_{\acc,i}}.
    \end{equation}
\end{lemma}
\begin{proof}
To analyze the total overhead of a concatenated protocol, we need to consider two factors. First, each layer $i\in\{1,2,\dotsc,L\}$ has its own overhead, $\mathcal{O}_i = \frac{R_i\cdot n_i}{k_i}$, representing the resource cost of that particular stage of distillation. Second, we need enough successful outputs from each layer to supply the required inputs for the subsequent layer.

To start with, for concreteness, let us work through a two-layer example. Suppose we want to end up with $k_2$ purified pairs at the end of the protocol. By definition, the second layer needs $n_2$ input pairs, and we need to run it $R_2$ times to succeed and produce the required $k_2$ purified pairs. This means the first layer needs to produce $R_2\cdot n_2$ pairs. Since each successful run of the first layer produces $k_1$ pairs, we need to run the first layer $R_2\cdot n_2/k_1$ times. Each one of these runs requires $R_1\cdot n_1$ pairs. Therefore, the total number of pairs required is $R_1\cdot n_1\cdot R_2\cdot n_2/k_1$. Altogether, the overhead is 
\begin{equation}
    \mathcal{O}=\frac{\text{number of input pairs}}{\text{number of output pairs}}=\frac{R_1\cdot n_1\cdot R_2\cdot n_2/k_1}{k_2}=\frac{R_1\cdot n_1}{k_1}\cdot\frac{R_2\cdot n_2}{k_2}=\mathcal{O}_1\cdot\mathcal{O}_2.
\end{equation}
Then, because the random variables $R_1$ and $R_2$ are independent (indeed, the number of runs needed for success of the first layer is independent of the number of runs needed for success of the second layer), we obtain \eqref{eq:expected_overhead_concatenated_protocol} with $L=2$.

The reasoning above for $L=2$ can be straightforwardly generalized to arbitrary $L$. Indeed, the number of required pairs at the very end is $k_L$. The $(L-1)$-st layer should therefore produce $R_L\cdot n_L$ pairs. This in turn means that the $(L-1)$-st layer should be run $R_L\cdot n_L/k_{L-1}$ times, with each successful run consuming $R_{L-1}\cdot n_{L-1}$ pairs, meaning that the number of pairs required in the $(L-1)$-st layer is $R_{L-1}\cdot n_{L-1}\cdot R_L\cdot n_L/k_{L-1}$ --- this is now the number of pairs needed to be produced by the $(L-2)$-st layer. In order for this to happen, the $(L-2)$-st layer must be run $R_{L-1}\cdot n_{L-1}\cdot R_L\cdot n_L/(k_{L-1}\cdot k_{L-2})$ times, with each successful run consuming $R_{L-2}\cdot n_{L-2}$ pairs. The number of pairs required in the $(L-2)$-st layer is therefore $R_{L-2}\cdot n_{L-2}\cdot R_{L-1}\cdot n_{L-1}\cdot R_L\cdot n_L/(k_{L-1}\cdot k_{L-2})$. Carrying on in this manner, we conclude that the overhead is
\begin{equation}
    \mathcal{O}=\frac{\text{number of input pairs}}{\text{number of output pairs}}=\frac{1}{k_L}\frac{\prod_{i=1}^L R_i\cdot n_i}{\prod_{i=1}^{L-1}k_i}=\prod_{i=1}^L\frac{R_i\cdot n_i}{k_i}=\prod_{i=1}^L\mathcal{O}_i,
\end{equation}
which implies the desired result, because the random variables $R_1,R_2,\dotsc,R_L$ are mutually independent.
\end{proof}

For the total overhead to remain bounded as we increase the number of layers (which we need to achieve very high fidelities), the individual layer overheads must decay rapidly enough that their product converges. Specifically, if we want the overhead to be independent of the target fidelity $\varepsilon_f$, we need
    $\prod_{i=1}^{\infty} \E[\mathcal{O}_i] < \infty$.
In the next section, we show that our random bilocal Clifford protocol achieves precisely this behavior --- the overhead at each layer decays very quickly with $i$, allowing us to achieve arbitrarily high fidelities with constant total overhead.

While we have focused on expected overhead, in practice we want guarantees about the actual number of resources required. The following lemma shows that we can achieve close to the expected overhead with high probability.

\begin{lemma}\label{lem:markov}
Consider any distillation protocol with expected overhead $\mathbb{E}[\mathcal{O}]$. We can achieve an overhead of $2\mathbb{E}[\mathcal{O}] \cdot \log_2 \delta^{-1}$ with probability at least $1-\delta$.
\end{lemma}

\begin{proof}
By Markov's inequality, a single run of the protocol has overhead at most $2\mathbb{E}[\mathcal{O}]$ with probability at least $1/2$. To achieve the claimed overhead, we simply run the protocol until it either succeeds, or exceeds an overhead budget $2\mathbb{E}[\mathcal{O}]$. Since the protocol will stay within the overhead budget with probability $1/2$, using just $\log_2\delta^{-1}$ independent attempts, we will succeed with probability at least $1-\delta$, and each of these attempts costs at most $2\mathbb{E}[\mathcal{O}]$. Therefore, the total overhead is at most $2 \log_2 \delta^{-1} \cdot \mathbb{E}[\mathcal{O}]$.
\end{proof}

\cref{lem:markov} applies equally well to the total overhead of concatenated protocols, because it depends only on Markov's inequality. Given a concatenated protocol with expected overhead $\mathbb{E}[\mathcal{O}] = \prod_{i=1}^L \mathbb{E}[\mathcal{O}_i]$, we can achieve an overhead of at most $2\mathbb{E}[\mathcal{O}]\cdot\log_2\delta^{-1}$ with probability $1-\delta$.

\section{Analysis of the random bilocal Clifford protocol}

To start with, \cref{eq:Clifford_twirl_2} tells us that the random Clifford operation twirls the input state into a mixture of just two components: the maximally mixed state and the desired maximally entangled state. The relative weights of these components depend only on the initial fidelity $f^n$. We can equivalently express the state in \eqref{eq:Clifford_twirl_2} as follows: 
\begin{equation}\label{eq:Clifford_twirl_2_supp}
\begin{aligned}
    \E_{C}[(C \otimes C^*) \rho_{A^nB^n} (C \otimes C^*)^\dagger] &= \frac{1}{1-4^{-n}} \qty((1-f^n) \frac{\id_{AB}^{\otimes n}}{4^n} + (f^n - 4^{-n}) \Phi_{AB}^{\otimes n})\\
    & = f^n \ketbra{I^{\otimes n}} + \frac{1-f^n}{4^n-1} \sum_{P \in \mathcal{P}_n \setminus \qty{I^{\otimes n}}} \ketbra{P},
\end{aligned}
\end{equation}
where in the second line we have let $\ket{P} \coloneqq (P\otimes\id)\ket{\Phi}^{\otimes n}$ for every $n$-qubit Pauli operator $P \in \mathcal{P}_n\coloneqq\{I,X,Y,Z\}^{\otimes n}$. This representation will be particularly useful when analyzing error correction strategies in the active setting of \cref{alg:random_bilocal_Clifford}.

\subsection{Passive setting}

Let us first introduce two key probability measures that are central to our analysis. We denote by $p_{\acc}$ the probability that the protocol accepts a given input state --- that is, the probability that the measurement outcomes match (in passive mode). Formally,
\begin{equation}\label{eq:pr_accept}
p_{\acc} = \Tr(\sigma_{A^nB^n}(\Pi_{A^mB^m}\otimes \id_{A^kB^k})),
\end{equation}
where $\Pi_{A^mB^m} = \sum_{\vec{x}\in\{0,1\}^m}\ketbra{\vec{x},\vec{x}}$ projects onto the space of matching measurement outcomes, and
\begin{equation}\label{eq:Clifford_twirl}
    \sigma_{A^nB^n} \coloneqq \mathbb{E}_{C}[(C\otimes C^{*})\rho_{A^nB^n}(C\otimes C^{*})^{\dag}].
\end{equation}
Additionally, we define $p_{\acc \land \Phi}$ as the probability of both accepting \emph{and} obtaining the target maximally entangled state $\Phi_{AB}^{\otimes k}$. It holds that
\begin{equation}\label{eq:pr_accept_and_Phi}
p_{\acc \land \Phi} = \Tr(\sigma_{A^nB^n} \left(\Pi_{A^mB^m} \otimes \Phi_{AB}^{\otimes k}\right)).
\end{equation}
The ratio of these probabilities determines the fidelity of the output state $\bar{\rho}_{A^kB^k}$ when the protocol succeeds:
\begin{equation}\label{eq:f-out}
\bra{\Phi}_{AB}^{\otimes k}\bar{\rho}_{A^kB^k}\ket{\Phi}_{AB}^{\otimes k} = \frac{p_{\acc \land \Phi}}{p_{\acc}} = (1-\bar{\varepsilon})^k ,
\end{equation}
where $\bar{f}\coloneqq 1-\bar{\varepsilon}$ is the single-pair fidelity of the output state.

\begin{lemma}\label{lem:probs_supp}
    The acceptance probabilities $p_{\acc}$ and $p_{\acc\wedge\Phi}$ \cref{alg:random_bilocal_Clifford} (passive setting) are
    \begin{align}
        p_{\acc}&=\frac{f^n-4^{-n}}{1-4^{-n}} + \frac{1-f^n}{1-4^{-n}} \frac{2^m \cdot 4^k}{4^n}, \label{eq:pr_accept_2_supp} \\
        p_{\acc\wedge\Phi}&=\frac{f^n-4^{-n}}{1-4^{-n}}+\frac{1-f^n}{1-4^{-n}}\frac{2^{m}}{4^n}. \label{eq:pr_accept_and_phi_2_supp}
    \end{align}
\end{lemma}

\begin{proof}
    This follows immediately using \cref{eq:Clifford_twirl_2_supp,eq:pr_accept,eq:pr_accept_and_Phi}.
\end{proof}

\begin{corollary}[Output fidelity of \cref{alg:random_bilocal_Clifford}, passive setting]\label{cor:simple-f_supp}
For every input state with fidelity at least $f^n \coloneqq (1-\varepsilon)^n$, the output fidelity when the protocol succeeds satisfies
\begin{equation}\label{eq:new-fid_supp}
    (1-\bar{\varepsilon})^k \geq \frac{f^n}{f^n + 2^{-m}(1-f^n)} \geq 1 - 2^{-m}(f^{-n}-1).
\end{equation}
\end{corollary}
\begin{proof}
This follows from plugging the result of \cref{eq:pr_accept_2_supp,eq:pr_accept_and_phi_2_supp} into the definition of $(1-\bar{\varepsilon})^k$ in \cref{eq:f-out}. From \eqref{eq:pr_accept_and_phi_2_supp}, we find that the numerator $p_{\acc\wedge\Phi}$ can be simply bounded from below as follows:
\begin{align}
    p_{\acc\wedge\Phi}&=\frac{f^n-4^{-n}}{1-4^{-n}} + \frac{1-f^n}{1-4^{-n}} \frac{2^m}{4^n}\\
    &=f^n+(1-f^n)\frac{2^m-1}{4^n-1}\\
    &\geq f^n.
\end{align}
Meanwhile, from \eqref{eq:pr_accept_2_supp}, the denominator $p_{\acc}$ can be bounded from above as follows:
\begin{align}
    p_{\acc}&=\frac{f^n-4^{-n}}{1-4^{-n}} + \frac{1-f^n}{1-4^{-n}} \frac{2^m \cdot 4^k}{4^n}\\
    &=f^n+(1-f^n)\frac{2^m\cdot 4^k-1}{4^n-1}\label{eq:pr_accept_3}\\
    &=f^n+(1-f^n)\frac{2^{-m}-4^{-n}}{1-4^{-n}}\\
    &=f^n+(1-f^n)2^{-m}\frac{1-2^{-2n+m}}{1-2^{-2n}}\\
    &\leq f^n+(1-f^n)2^{-m},
\end{align}
where for the final inequality we used the fact that $\frac{1-2^{-2n+m}}{1-2^{-2n}}\leq 1$. The term $2^{-m}(1-f^n)$ can be understood as a `false positive' event, meaning the input state had some non-trivial error, but the protocol nevertheless detected no error. One way to understand this is that $\Pi_{AB}$ can also be written $\frac{\id+Z_A Z_B}{2}$, hence it only allows Paulis $P$ for which the first $m$ characters are all either $I$ or $Z$ to pass. There are $2^m \cdot 4^k$ such Paulis, while there are $4^n$ total Paulis, hence the term $2^{-m}$. Combining the bounds on the numerator and denominator leads to the first inequality in \eqref{eq:new-fid_supp}. The second inequality follows readily from the fact that $\frac{1}{1+x}\geq 1-x$ for $|x|\leq 1$.
\end{proof}

\begin{proposition}
    Consider the random bilocal Clifford protocol (\cref{alg:random_bilocal_Clifford}) in the passive setting. For
    \begin{equation}\label{eq:fid_improvement_condition_supp}
        m\geq\log_2\left(\frac{1-f^n}{f^{n-1}(1-f)}\right),
    \end{equation}
    where $m$ is the number of measured qubits, the fidelity of the output state exceeds the input single-pair fidelity $f$.
\end{proposition}

\begin{proof}
    The output fidelity is given by \cref{eq:f-out}. 
    Using the first inequality in \eqref{eq:new-fid_supp}, the condition $\frac{f^n}{f^n+2^{-m}(1-f^n)}\geq f$ 
    readily leads to the claimed condition on $m$.
\end{proof}

At the value of $m$ specified in the right-hand side of \eqref{eq:fid_improvement_condition_supp}, the success probability of the protocol is
\begin{align}
    p_{\acc}
    &=\frac{f^{n-1}-4^{-n}}{1-4^{-n}},
\end{align}
from which we obtain the expected overhead needed to achieve an improvement in fidelity:
\begin{equation}
    \mathbb{E}[\mathcal{O}]=\frac{n(1-4^{-n})}{(n-\log_2[(1-f^n)/(f^{n-1}(1-f))])(f^{n-1}-4^{-n})}.
\end{equation}
At the critical value of $m$, namely $m=-n\log_2 f$, the success probability is
\begin{equation}
    p_{\acc}=\frac{1-(4f)^n(2-f^n)}{1-4^n},
\end{equation}
and the corresponding expected overhead is
\begin{equation}
    \mathbb{E}[\mathcal{O}]=\frac{1-4^n}{(1+\log_2f)(1-(4f)^n(2-f^n))}.
\end{equation}

\begin{proposition}\label{lem:settings} 
Assume we are given access to Bell pairs with infidelity $\varepsilon \leq 0.05$. If we set
\begin{equation}\label{eq:settings}
    n=\varepsilon^{-1/2} \qq{and} m=\log_2 \varepsilon_f^{-1},
\end{equation}
then for every target infidelity $\varepsilon_f \geq 2^{-\varepsilon^{-1/3}}$, the random Clifford protocol distills $\varepsilon_f$ infidelity Bell pairs with an expected overhead bounded from above as follows:
\begin{equation}\label{eq:single-overhead}
\mathbb{E}[\mathcal{O}]\leq \exp(2\varepsilon^{1/6}).  
\end{equation}
Furthermore, for $\varepsilon\leq 0.001$, the output infidelity is lower than the input infidelity $\varepsilon$.
\end{proposition}

\begin{proof}
First, we consider the constraint that we must output at least one Bell pair, meaning $k \geq 1$. This constraint on $k$, combined with \eqref{eq:settings}, translates to the constraint $\varepsilon_f \geq 2^{1-1/\sqrt{\varepsilon}}$. Next, we want the output fidelity to satisfy our target, i.e., $(1-\varepsilon_f)^k\leq (1-\bar{\varepsilon})^k$. Using the bound from \cref{eq:new-fid_supp}, this will be satisfied if we require
\begin{equation}\label{cond2}
    1-\left((1-\varepsilon)^{-1/\sqrt{\varepsilon}}-1\right)\varepsilon_f \geq (1-\varepsilon_f)^k.
\end{equation}
For $\varepsilon \leq 0.4$, $(1-\varepsilon)^{-1/\sqrt{\varepsilon}}-1$ is bounded from above by $2\sqrt{\varepsilon}$. This leads to the simpler sufficient condition for the output fidelity to satisfy our target:
\begin{equation}
    1-2\sqrt{\varepsilon}\varepsilon_f \geq (1-\varepsilon_f)^k,
\end{equation}
which is satisfied for $k \geq 1$ when $\varepsilon \leq 1/4$. (Since $(1-\varepsilon_f)^k$ decreases with $k$, it holds that $1-\varepsilon_f\geq(1-\varepsilon_f)^k$ for all $k\geq 1$.) Furthermore, we can bound the success probability as
\begin{equation}
    p_{\acc} \geq (1-\varepsilon)^n = (1-\varepsilon)^{1/\sqrt{\varepsilon}} \geq 1-\sqrt{\varepsilon},
\end{equation} 
where the first inequality follows readily from \eqref{eq:pr_accept_3} and the fact that $f\equiv 1-\varepsilon$.

Now we can analyze the overhead. Using our expressions for $n$, $k=n-m$, and $p_{\acc}$, we obtain
\begin{equation}
\mathbb{E}[\mathcal{O}] \leq \frac{1/\sqrt{\varepsilon}}{(1/\sqrt{\varepsilon}+\log\varepsilon_f)(1-\sqrt{\varepsilon})} = \frac{1}{1-\sqrt{\varepsilon}\log_2(1/\varepsilon_f)}\frac{1}{1-\sqrt{\varepsilon}}.  
\end{equation}
If we were to set $\varepsilon_f = 2^{1-1/\sqrt{\varepsilon}}$, the overhead would grow unbounded as $\varepsilon$ approaches zero. This motivates our more stringent condition $\varepsilon_f \geq 2^{-\varepsilon^{-1/3}}$. For $\varepsilon \leq 0.1$, this stronger condition implies our earlier requirement $\varepsilon_f \geq 2^{1-1/\sqrt{\varepsilon}}$. With this choice, we can bound the overhead as
\begin{equation}
\mathbb{E}[\mathcal{O}] \leq \frac{1}{1-\varepsilon^{1/6}}\frac{1}{1-\varepsilon^{1/2}} \leq e^{2\varepsilon^{1/6}},
\end{equation}
where the final inequality holds for $\varepsilon \leq 0.05$.

Finally, let us note that improvement in infidelity is possible for $\varepsilon\leq 0.001$. Indeed, it is straightforward to verify that $2^{-\varepsilon^{-1/3}}\leq\varepsilon$ for $\varepsilon\leq 0.001$.
\end{proof}

We now combine \cref{lem:settings} with a general analysis of concatenated protocols, which establishes our main result.

\begin{theorem}[Concatenated passive bilocal random Clifford protocol]\label{thm:bilocal_Clifford_concat_supp}
Let $\varepsilon_0 \leq 0.0006$ be an initial infidelity and $\varepsilon_f$ be any target infidelity satisfying $0 < \varepsilon_f < \varepsilon_0$. The concatenated passive bilocal Clifford protocol (\cref{alg:random_bilocal_Clifford}, passive setting) succeeds with probability $1 - \delta$ and it requires:
\begin{enumerate}[itemsep=0.3ex]
    \item an overhead of $\mathcal{O} \leq 411\log_2(\delta^{-1})$;
    \item a number of input Bell pairs that is at most $411 (\log_2 \varepsilon_f^{-1})^{3/2} \log_2(\delta^{-1})$;
    \item Alice and Bob to hold at most $(\log_2 \varepsilon_f^{-1})^{3/2}$ Bell pairs in memory at any given time;
    \item $L$ layers of concatenation, where $L \leq O(\log_2^*(\varepsilon_f^{-1}))$;
    \item at most $O(\poly \log \varepsilon_f^{-1})$ gates, which can be organized to run in $O(\poly \log \log \varepsilon_f^{-1})$ depth.
\end{enumerate}
\end{theorem}
\begin{proof}
From \cref{lem:settings}, we know that starting with infidelity $\varepsilon$, our random bilocal Clifford protocol can achieve output infidelity $2^{-\varepsilon^{-1/3}}$ with expected overhead at most $\exp(2\varepsilon^{1/6})$, and improvement in fidelity is possible for $\varepsilon\leq 0.001$. Let us analyze what happens to the expected overhead when we concatenate this protocol.
First, for layer $\ell\in\{1,2,\dotsc\}$, the infidelity is given by $\varepsilon_{\ell}=2^{-\varepsilon_{\ell-1}^{-1/3}}$. Now, because $2^{-\varepsilon^{-1/3}}\leq\frac{1}{2}\varepsilon$ for $\varepsilon\leq 0.0006$, we have that $\varepsilon_{\ell}\leq \frac{1}{2}\varepsilon_{\ell-1}$ for $\varepsilon_{\ell-1}\leq 0.0006$ and $\ell\in\{1,2,\dotsc\}$. In particular, we have that $\varepsilon_{\ell}\leq 2^{-\ell}\varepsilon_0$ for $\ell\in\{1,2,\dotsc\}$ and $\varepsilon_0\leq 0.0006$. From \cref{lem:settings}, it follows that the overhead of layer $\ell$ of the concatenated protocol is bounded from above by $\exp(2\cdot 2^{-(\ell-1)/6}\varepsilon_0^{1/6})$.



Applying \cref{lem:markov}, the total overhead is bounded by $\mathcal{O} \leq 2 \mathbb{E}[\mathcal{O}]\log_2\delta^{-1}$, where
\begin{equation}
    2\mathbb{E}[\mathcal{O}] \le 2\prod_{\ell=1}^{\infty} \exp(2 \cdot 2^{-(\ell-1)/6}\varepsilon_0^{1/6}) = 2\exp(2\sum_{\ell=0}^{\infty}2^{-\ell/6}\varepsilon_0^{1/6}) \leq 411.
\end{equation}
This completes the proof of 1. 

Now, if $L$ is the last layer of the concatenated protocol, the number of Bell pairs produced by this final layer is
\begin{equation}
    \varepsilon_{L-1}^{-1/2} - \log_2 \varepsilon_f^{-1} = \qty((\log_2 \varepsilon_f^{-1})^3)^{1/2} - \log_2 \varepsilon_f^{-1} \leq (\log_2 \varepsilon_f^{-1})^{3/2},
\end{equation}
where we have used the fact that $\varepsilon_f = 2^{-\varepsilon_{L-1}^{-1/3}}$. Therefore the number of required input Bell pairs (including all retries) is at most $N_{tot}=(\log_2 \varepsilon_f^{-1})^{3/2} \cdot\mathcal{O}$. Combining this with 1., we obtain 2. 

Using the recurrence $\varepsilon_{\ell+1} \leq \varepsilon_\ell/2$, we see that the number of layers of concatenation can be upper bounded with $\log_2(\varepsilon_0/\varepsilon_f)$. In fact, this bound can be improved. By repeating the recurrence $\varepsilon_\ell \to 2^{-\varepsilon_\ell^{-1/3}}$ twice, we achieve $\varepsilon_\ell \to 2^{-\varepsilon_{\ell}}$ for any $\varepsilon_\ell$ below some constant $\varepsilon^*$. Therefore, the number of required layers is simply $O(\log_2^*(\varepsilon_f^{-1}))$ where $\log_2^*$ is the iterated logarithm.

Finally, using Ref.~\cite{cleve2016near}, we can implement circuits sampled from an $N$-qubit 2-design using $O(N \log^2 N \log\log N) \leq O(N \log^3 N)$ Clifford gates in $O(\log^2 N)$ depth. The number of repetitions of any layer is at most $\log_2(\varepsilon_f^{-1})^{3/2}$ (since this is the number of inputs to the final layer of the protocol). Therefore, the total number of gates required is at most
\begin{equation}
    O\qty(\log(\varepsilon_f^{-1})^{3/2} \cdot \sum_{\ell=l}^L N_\ell \log^3 N_\ell) \leq O\qty(\log(\varepsilon_f^{-1})^{3/2} \log_2^*(\varepsilon_f^{-1}) N_{L} \log^3 N_{L}) \leq O(\poly \log \varepsilon_f^{-1}).
\end{equation}
where we used item 4., i.e. $L\le O(\log^*(\varepsilon_f^{-1}))$. Meanwhile, the overall depth is at most $O(L \cdot \log^2 N_L) \leq O(\poly \log \log \varepsilon_f^{-1})$.
\end{proof}

\begin{remark}
\cref{thm:bilocal_Clifford_concat_supp} says that by concatenating our passive bilocal random Clifford protocol, we can distill Bell pairs with arbitrarily low target infidelity $\bar{\varepsilon}$ while maintaining a constant (with respect to $\bar{\varepsilon}$) overhead. The key technical innovation is the careful choice of parameters that ensures both improving fidelity and bounded overhead at each layer. The assumption of an initial infidelity of $0.0006$ is an artifact of our proof, and is not limiting for two reasons. From a practical perspective, our numerical results, illustrate that in regimes of practical interest, the overhead of our protocol is lower than existing methods for initial infidelities $\varepsilon_0$ far larger than $0.0006$ (see \cref{fig:results,tab:comparisons}). From a theoretical perspective, one can always use standard distillation schemes~\cite{BDSW96,DEJMPS96} to first reduce the infidelity to $0.0006$ before applying our protocol. Since this initial distillation phase requires only a constant number of operations independent of the final target infidelity $\bar{\varepsilon}$, it does not affect the asymptotic scaling of our protocol's resources.
\end{remark}

\subsection{Active setting}

To analyze the performance of our active strategy, we first need to understand how syndromes behave under random Clifford operations.

\begin{lemma}[Syndrome matching probabilities]\label{lem:syndrome-prob}
For any two Pauli operators $P$ and $P'$, under a uniformly random Clifford operation $C$,
\begin{equation}
    \Pr_{C}[\mathcal{S}(C(P)) = \mathcal{S}(C(P'))] = \begin{cases}
        1 &\qq{if $P=P'$,} \\
        \frac{2^{-m} \cdot 4^n/2 - 2}{4^n/2-2} \leq 2^{-m} &\qq{if $P \neq P'$ and $\comm{P}{P'}=0$,} \\
        2^{-m} &\qq{if $\comm{P}{P'} \neq 0$,}
    \end{cases}
\end{equation}
where $C(P) \equiv C P C^\dagger$ and $\mathcal{S}(P) \in \qty{0,1}^m$ is the syndrome of a Pauli error $P$. The $k$th bit of $\mathcal{S}(P)$ is an indicator variable that is equal to $1$ if the $k$th character in $P$ is $X$ or $Y$.
\end{lemma}
\begin{proof}
The case $P = P'$ is trivial. For commuting $P \neq P'$, the Pauli 2-mixing property of Clifford operations~\cite[Lemma~3]{webb2016clifford} implies that $C(P)$ has $4^n/2-2$ possible values given $C(P')$. Among these, only $2^m \cdot 4^{n-m}/2$ match the syndrome of $C(P')$, but two of these possibilities are disallowed since $C(P) \neq C(P') \neq I^{\otimes n}$. For anticommuting pairs, a similar counting argument gives exactly $2^{-m}$.
\end{proof}

Using these syndrome matching probabilities, we can bound the probability of accepting a state (either through no error being detected or through attempted correction).

\begin{lemma}[Acceptance probability bounds for \cref{alg:random_bilocal_Clifford}, active setting]\label{lem:accept-prob}
For an $\Err$-active correction strategy, the acceptance probability $p_{\acc}$ is bounded by
\begin{equation}\label{eq:active-pacc}
    q \leq p_{\acc} \leq q + 2^{-m} \cdot (\Err+1) (1-q),
\end{equation}
where $q \coloneqq \sum_{\ell=1}^{\Err+1} q_\ell$. The joint probability $p_{\acc \land \Phi}$ is lower bounded by
\begin{equation}\label{eq:active-pjoint}
    p_{\acc \land \Phi} \geq q - 2^{-m} \cdot \Err (q-f^n).
\end{equation}
\end{lemma}
\begin{proof}
Letting $q_P$ be the coefficient associated with $P$, the following equality holds:
\begin{equation}
    p_{\acc} = \sum_{P \in \mathcal{P}_n} q_P \cdot \Pr_C[\mathcal{S}(C(P)) \in \qty{\mathcal{S}(C(P_\ell)) \mid \ell=1,\ldots,\Err+1}],
\end{equation}
since we accept a Pauli error $P$ when its syndrome $\mathcal{S}(C(P))$ is in the set of correctable syndromes, $\qty{\mathcal{S}(C(P_\ell)) \mid \ell=1,\ldots,\Err+1}$. For every error $P$ in our correctable set, $\Pr_C[\mathcal{S}(C(P)) \in \qty{\mathcal{S}(C(P_\ell)) \mid \ell=1,\ldots,\Err+1}]=1$ trivially. This immediately results in the claimed lower bound $p_{\acc}\geq q$. On the other hand, there may be errors outside our correctable set that nevertheless have a syndrome that is identical to one of the errors in the correctable set. Combining \cref{lem:syndrome-prob} with a union bound, we see that in these cases, $\Pr_C[\mathcal{S}(C(P)) \in \qty{\mathcal{S}(C(P_\ell)) \mid \ell=1,\ldots,\Err+1}] \leq 2^{-m} \cdot (\Err+1)$. Hence, we obtain our claimed upper bound on $p_{\acc}$ as follows:
\begin{align}
    p_{\acc}&=\sum_{j=1}^{\Err}q_j+\sum_{j\ge \Err+1}q_{j}\Pr_C[\mathcal{S}(C(P_j)) \in \qty{\mathcal{S}(C(P_\ell)) \mid \ell=1,\ldots,\Err+1}]\\
    &\le q+ 2^{-m}(\Err+1)\sum_{j\ge \Err+1} q_j\\
    &=q+ 2^{-m}(\Err+1)(1-q).
\end{align}

We now show the lower bound on $p_{\acc \land \Phi}$. The first term $f^n$ represents the case where no error occurred. For each additional correctable error $P_\ell$, we succeed in correcting it unless its syndrome matches that of a more probable error (i.e., $P_{\ell'}$ for $\ell' < \ell$). That is, the fidelity will be
\begin{equation}
    f^n + \sum_{\ell=2}^{\Err+1} q_\ell \cdot \Pr_C[\mathcal{S}(C(P_\ell)) \not\in \qty{\mathcal{S}(C(P_{\ell'})) \mid \ell'=0,\ldots,\ell-1}].
\end{equation}
Applying a simple union bound, combined with \cref{lem:syndrome-prob}, we see that $\Pr_C[\mathcal{S}(C(P_\ell)) \not\in \qty{\mathcal{S}(C(P_{\ell')}) \mid \ell'=0,\ldots,\ell-1}] \geq 1-(\ell-1) \cdot 2^{-m} \geq 1-\Err \cdot 2^{-m}$, leading to the claimed lower bound. Indeed, we note that $\sum_{\ell=2}^{\Err+1} q_\ell = q - q_1 = q - f^n$, since $q_1=f^n$. We therefore have $p_{\acc \land \Phi} \geq f^n + (1-\Err \cdot 2^{-m}) (q-f^n)$, which can be simplified to read $p_{\acc \land \Phi} \geq q - \Err \cdot 2^{-m} (q-f^n)$, as required.
\end{proof}

\begin{theorem}[Output fidelity of \cref{alg:random_bilocal_Clifford}, active setting]\label{thm:active-perform_supp}
Under an $\Err$-active correction strategy for \cref{alg:random_bilocal_Clifford}, the fidelity of the output state, conditioned on acceptance, is bounded from below as follows:
\begin{equation}
    (1-\bar{\varepsilon})^k \geq 1 - 2^{-m} \cdot (\Err+1) (q^{-1}-1).
\end{equation}
which strictly generalizes \cref{eq:new-fid_supp}, valid for $\Err=0$.
\end{theorem}

\begin{proof}
Combining the upper bound from \cref{eq:active-pacc} and the lower bound from \cref{eq:active-pjoint}, we get
\begin{equation}
    (1-\bar{\varepsilon})^k \geq \frac{1 - \Err \cdot 2^{-m} (1-q^{-1} f^n)}{1+2^{-m} \cdot (\Err+1)(q^{-1}-1)} \geq \frac{1}{1+2^{-m} \cdot (\Err+1) (q^{-1}-1)} \geq 1 - 2^{-m} \cdot (\Err+1) (q^{-1}-1),
\end{equation}
where for the final inequality we use the fact that $\frac{1}{1+x}\geq 1-x$ for $|x|\leq 1$.
\end{proof}

\subsection{Finite-depth protocols}

To analyze this protocol, we first observe that under random bilocal Clifford operations $C \otimes C^*$, the state of any two qubit pairs can be described as a mixture of product states $\rho_1 \otimes \rho_2$, where each $\rho_i$ is either the two-qubit maximally mixed state $\id_{AB}/4$ or the two-qubit maximally entangled state $\Phi_{AB} = \ketbra{\Phi}_{AB}$. An alternative formulation can take $\rho_i$ to be either $\Phi_{AB}$ or $\Phi_{AB}^\perp$, where $\Phi_{AB}^\perp \coloneqq \frac{1}{3}(\id_{AB}-\Phi_{AB})=\frac{1}{3} (\ketbra{X}+\ketbra{Y}+\ketbra{Z})$. (Note that $\id_{AB} = \Phi_{AB} + 3\Phi_{AB}^\perp$.) Indeed, it readily follows from \cref{eq:Clifford_twirl_2_supp} that for every state $\rho_{A_1A_2B_1B_2}$,
\begin{multline}\label{eq:twirl_two_qubits}
    \mathbb{E}_C[(C_{A_1A_2}\otimes C_{B_1B_2}^{\ast})\rho_{A_1A_2B_1B_2}(C_{A_1A_2}\otimes C_{B_1B_2}^{\ast})^{\dagger}]\\=f^2\Phi_{A_1B_1}\otimes\Phi_{A_2B_2}+(1-f^2)\left(\frac{1}{5}\Phi_{A_1B_1}\otimes\Phi_{A_2B_2}^{\perp}+\frac{1}{5}\Phi_{A_1B_1}^{\perp}\otimes\Phi_{A_2B_2}+\frac{3}{5}\Phi_{A_1B_1}^{\perp}\otimes\Phi_{A_2B_2}^{\perp}  \right).
\end{multline}
From now on, for ease of notation, we suppress the system labels and the tensor product symbol; for example, we will write $\Phi\Phi\equiv\Phi_{A_1B_1}\otimes\Phi_{A_2B_2}$. Using \cref{eq:twirl_two_qubits}, we have the following transformations under the map $\rho \mapsto \E_C[(C \otimes C^*) \rho (C \otimes C^*)^\dagger]$:
\begin{equation}\label{eq:two-q}
    \begin{gathered}
        \Phi \Phi \mapsto \Phi \Phi, \\
        \Phi^\perp \Phi^\perp, \Phi \Phi^\perp, \Phi^\perp \Phi \mapsto \frac{1}{5} \Phi \Phi^\perp + \frac{1}{5} \Phi^\perp \Phi + \frac{3}{5} \Phi^\perp \Phi^\perp.
    \end{gathered}
\end{equation}
These transformation rules define a Markov chain on the space of such product states. Due to the all-to-all connectivity in our protocol (meaning any pair of qubits can be chosen for the next gate), the state of the system can be described by a single parameter: the number $w$ of $\Phi^\perp$ factors, which simply tracks the weight of a Pauli error as it propagates through the circuit. For a system of $n$ pairs of qubits, we can keep track of the probability distribution of these weights using vectors $\kett{w}$, for $w\in\{0,1,\dotsc,n\}$, which represents the uniform distribution over all states with weight $w$ (i.e., all strings in $\{\Phi,\Phi^{\perp}\}^{\otimes n}$ with $w$ occurrences of $\Phi^{\perp}$). As an example, for $n=1$, we simply have $|\mkern-2.5mu|0\rangle\mkern-4mu\rangle\leftrightarrow\Phi$, $|\mkern-2.5mu|1\rangle\mkern-4mu\rangle\leftrightarrow\Phi^{\perp}$; for $n=2$, we have $|\mkern-2.5mu|0\rangle\mkern-4mu\rangle\leftrightarrow\Phi\Phi$, $|\mkern-2.5mu|1\rangle\mkern-4mu\rangle\leftrightarrow\frac{1}{2}(\Phi\Phi^{\perp}+\Phi^{\perp}\Phi)$, $|\mkern-2.5mu|2\rangle\mkern-4mu\rangle\leftrightarrow\Phi^{\perp}\Phi^{\perp}$, and so on. (We suppress the dependence of the vectors $\kett{w}$ on $n$.) The state of the system is then simply a convex combination of these elementary probability vectors $\kett{w}$. In particular, for a system of $n$ pairs of qubits, the initial distribution over error weights for local depolarizing noise with strength $\lambda$ is
\begin{equation}\label{eq:init-state}
    \kett{x_0} = \sum_{w=0}^{n} \binom{n}{w} \lambda^w (1-\lambda)^{n-w} \kett{w}.
\end{equation}

\begin{lemma}[Noiseless transitions]\label{lem:noiseless}
Consider a system of $n\in\{2,3,\dotsc\}$ pairs of qubits. For a state with error weight $w$, under a random two-qubit Clifford gate applied to a random pair of qubits, the transition probabilities for going to an error weight of $w'$ are: 
\begin{equation}\label{eq:noiseless-transition}
T_{w',w} = \begin{cases}
        \frac{6w(n-w)}{5n(n-1)} &\qq{if $w'=w+1 \leq n$,} \\[1ex]
        \frac{4w(n-w)+5(n-w)(n-w-1)+3w(w-1)}{5n(n-1)} &\qq{if $w'=w$,} \\[1ex]
        \frac{2w(w-1)}{5n(n-1)} &\qq{if $w'=w-1 \geq 0$,} \\[1ex]
        0 &\qq{otherwise.}
    \end{cases}    
\end{equation}
\end{lemma}

\begin{proof}
We start by noting that there are $\binom{n}{w}$ strings in $\{\Phi,\Phi^{\perp}\}^{\otimes n}$ with weight $w$. There are also $\binom{n}{2}$ qubits on which a two-qubit gate can be applied by both Alice and Bob. Selecting a pair of qubits corresponding to making a choice of two elements of the $w$-weight string. These two elements can be either $\Phi\Phi^{\perp}$, $\Phi^{\perp}\Phi$, $\Phi\Phi$, or $\Phi^{\perp}\Phi^{\perp}$. 
\begin{itemize}[noitemsep]
    \item For every $w$-weight string in $\{\Phi,\Phi^{\perp}\}^{\otimes n}$, there are $w(n-w)$ ways of choosing either $\Phi^{\perp}\Phi$ or $\Phi\Phi^{\perp}$. The probability of selecting $\Phi^{\perp}\Phi$ or $\Phi\Phi^{\perp}$ is therefore $w(n-w)/\binom{n}{2}=\frac{2w(n-w)}{n(n-1)}$.
    \item There are $\binom{n-w}{2}$ ways of choosing $\Phi\Phi$ from every $w$-weight string. Therefore, the probability of selecting $\Phi\Phi$ is $\binom{n-w}{2}/\binom{n}{2}=\frac{(n-w)(n-w-1)}{n(n-1)}$.
    \item Lastly, there are $\binom{w}{2}$ ways of selecting $\Phi^{\perp}\Phi^{\perp}$ from every $w$-weight string, meaning that the probability of selecting $\Phi^{\perp}\Phi^{\perp}$ is $\binom{w}{2}/\binom{n}{2}=\frac{w(w-1)}{n(n-1)}$.
\end{itemize}
Now, from \cref{eq:two-q}, we observe that when we hit a $\Phi \Phi^\perp$ or $\Phi^\perp \Phi$ pair, it stays the same weight with probability $\frac{2}{5}$, and it increases weight with probability $\frac{3}{5}$. Similarly, when we hit a $\Phi^\perp \Phi^\perp$ pair, we observe that the weight cannot increase --- we either \emph{decrease} weight with probability $\frac{2}{5}$ or stay the same with probability $\frac{3}{5}$. When we hit $\Phi\Phi$, the weight remains unchanged. The claimed transition probabilities follow by combining these cases.
\end{proof}

\begin{remark}
    The transition matrix defined by \cref{eq:noiseless-transition} models a birth-death process; the stationary solutions for such Markov processes can be calculated in closed form (see, e.g., Ref.~\cite{karlin1957classification}). In our case, the stationary distribution is proportional to the ``truncated'' the binomial distribution:
    \begin{equation}
        \kett{x_{\infty}} \propto \sum_{w=1}^n \binom{n}{w} \cdot (3/4)^w \cdot (1/4)^{n-w} \kett{w}.
    \end{equation}
    The reason for this is simple: in the limit of infinite time, the transition matrix models the action of a uniformly random Clifford, which twirls non-identity Paulis into uniformly random Paulis, and the distribution of weights for a uniformly random Pauli is a binomial distribution with parameter $\frac{3}{4}$.
\end{remark}

The situation becomes more interesting when we introduce noise. Consider a two-qubit noise channel $\mathcal{N}$ that acts before every two-qubit gate. The following three parameters fully describe any noise channel for our analysis:
\begin{equation}\label{eq:fidelity_parameters}
\begin{aligned}
    f_{0} &= \Tr[(\Phi_{A_1B_1}\otimes\Phi_{A_2B_2})(\mathcal{N}_{A_1A_2}\otimes\mathcal{N}_{B_1B_2})(\Phi_{A_1B_1}\otimes\Phi_{A_2B_2})]\\
    &=\Tr[(\mathcal{N}_{A_1A_2}^{\dagger}\otimes\mathcal{N}_{B_1B_2}^{\dagger})(\Phi_{A_1A_2B_1B_2})(\Phi_{A_1B_1}\otimes\Phi_{A_2B_2})],\\
    f_{1} &= \Tr[(\Phi_{A_1B_1}\otimes\Phi_{A_2B_2})(\mathcal{N}_{A_1A_2}\otimes\mathcal{N}_{B_1B_2})(\Phi_{A_1B_1}\otimes\Phi_{A_2B_2}^{\perp})]\\
    &=\Tr[(\mathcal{N}_{A_1A_2}^{\dagger}\otimes\mathcal{N}_{B_1B_2}^{\dagger})(\Phi_{A_1A_2B_1B_2})(\Phi_{A_1B_1}\otimes\Phi_{A_2B_2}^{\perp})],\\
    f_{2} &= \Tr[(\Phi_{A_1B_1}\otimes\Phi_{A_2B_2})(\mathcal{N}_{A_1A_2}\otimes\mathcal{N}_{B_1B_2})(\Phi_{A_1B_1}^{\perp}\otimes\Phi_{A_2B_2}^{\perp})]\\
    &=\Tr[(\mathcal{N}_{A_1A_2}^{\dagger}\otimes\mathcal{N}_{B_1B_2}^{\dagger})(\Phi_{A_1A_2B_1B_2})(\Phi_{A_1B_1}^{\perp}\otimes\Phi_{A_2B_2}^{\perp})],
\end{aligned}
\end{equation}
where for the second equality in each line we have used the fact that $\Phi_{A_1B_1}\otimes\Phi_{A_2B_2}=\Phi_{A_1A_2B_1B_2}$ (up to a reordering of the subsystems). The transitions now read as follows:
\begin{equation}\label{eq:two-q-noisy}
    \begin{gathered}
        \Phi \Phi \mapsto f_0 \Phi \Phi + (1-f_0) \qty(\frac{1}{5} \Phi \Phi^\perp + \frac{1}{5} \Phi^\perp \Phi + \frac{3}{5} \Phi^\perp \Phi^\perp), \\
        \Phi \Phi^\perp, \Phi^\perp \Phi \mapsto f_1 \Phi \Phi + (1-f_1) \qty(\frac{1}{5} \Phi \Phi^\perp + \frac{1}{5} \Phi^\perp \Phi + \frac{3}{5} \Phi^\perp \Phi^\perp), \\
        \Phi^\perp \Phi^\perp \mapsto f_2 \Phi \Phi + (1-f_2) \qty(\frac{1}{5} \Phi \Phi^\perp + \frac{1}{5} \Phi^\perp \Phi + \frac{3}{5} \Phi^\perp \Phi^\perp).
    \end{gathered}
\end{equation}

For common noise models, the fidelity parameters in \eqref{eq:fidelity_parameters} take on specific values.
\begin{itemize}[noitemsep]
    \item Let $\mathcal{N}=\mathcal{D}_{\lambda}$ be the two-qubit depolarizing channel:
    \begin{equation}
        \mathcal{D}_{\lambda}(\rho)\coloneqq(1-\lambda)\rho+\frac{\lambda}{15}\sum_{\substack{P\in\mathcal{P}_2\\P\neq I\otimes I}}P\rho P,
    \end{equation}
    where $\rho$ is a two-qubit state and $\lambda\in[0,1]$ is the depolarizing strength. For this channel, we have $f_0=(1-\lambda)^2+\frac{\lambda^2}{15}$, and $f_1=f_2=\frac{\lambda}{15}(2-\frac{16}{15}\lambda)$.
    \item Let $\mathcal{N}=\mathcal{A}_{\gamma}\otimes\mathcal{A}_{\gamma}$, where $\mathcal{A}_{\gamma}$ is the single-qubit amplitude damping channel with damping parameter $\gamma\in[0,1]$, defined as follows:
    \begin{equation}
        \begin{gathered}
            \mathcal{A}_{\gamma}(\rho)\coloneqq K_1\rho K_1^{\dagger}+K_2\rho K_2^{\dagger},\\
            K_1=\begin{pmatrix} 1 & 0 \\ 0 & \sqrt{1-\gamma}\end{pmatrix},\quad K_2=\begin{pmatrix} 0 & \sqrt{\gamma} \\ 0 & 0 \end{pmatrix}.
        \end{gathered}
    \end{equation}
    For this channel, we have $f_0 = (\gamma^2/2-\gamma+1)^2 \approx 1-2\gamma$, $f_1=\frac{\gamma(\gamma^3-2\gamma+4)}{12} \approx \gamma/3$, and $f_2=\frac{\gamma^2(\gamma+2)^2}{36} \approx (\gamma/3)^2$.
\end{itemize}

\begin{lemma}[Noisy transitions]\label{lem:noisy}
Consider a system of $n\in\{2,3,\dotsc\}$ pairs of qubits. For a state with error weight $w$, under noisy two-qubit Clifford gates, the transition probabilities for going to an error weight $w'$ are:
\begin{equation}\label{eq:transition-noise}
    T_{w',w} = \begin{cases}
        \frac{3(1-f_0) (n-w)(n-w-1)}{5n(n-1)} &\qq{if $w'=w+2 \leq n$,} \\[1ex]
        \frac{2(1-f_0)(n-w)(n-w-1)+6(1-f_1)w(n-w)}{5n(n-1)} &\qq{if $w'=w+1 \leq n$,} \\[1ex]
        \frac{5f_0 (n-w)(n-w-1)+4(1-f_1)w(n-w)+3(1-f_2)w(w-1)}{5n(n-1)} &\qq{if $w'=w$,} \\[1ex]
        \frac{10f_1 w(n-w) + 2(1-f_2)w(w-1)}{5n(n-1)} &\qq{if $w'=w-1 \geq 0$,} \\[1ex]
        \frac{f_2 w(w-1)}{n(n-1)} &\qq{if $w'=w-2 \geq 0$,} \\[1ex]
        0 &\qq{otherwise.}
    \end{cases}
\end{equation}
\end{lemma}

\begin{proof}
The proof is analogous to the proof of \cref{lem:noiseless}, except that now we use the transitions in \cref{eq:two-q-noisy} instead of the noiseless transitions in \cref{eq:two-q}.
\end{proof}

After applying $G$ gates, the state distribution is given by $\kett{x_G} = T^G \kett{x_0}$, where $T^G$ denotes the $G$-fold matrix power of the transition matrix $T$ defined by \cref{eq:transition-noise} and the initial vector $\kett{x_0}$ is defined in \cref{eq:init-state}.

\begin{theorem}[Performance of finite-depth random bilocal Clifford protocols]\label{thm:performance_supp}
For every error weight distribution $\kett{x}$, 
\begin{equation}
    \begin{gathered}
        p_{\acc \land \Phi} = \E_{w \sim \kett{x}}\qty[\frac{3^{-w} \binom{m}{w}}{\binom{n}{w}}] \qq{and} \\
        p_\acc = \E_{w \sim \kett{x}}\qty[\frac{3^{-w}}{\binom{n}{w}} \sum_{j=\max\{0,w-m\}}^{\min\{w,k\}} \binom{k}{j} \binom{m}{w-j} 3^j]
    \end{gathered}
\end{equation}
\end{theorem}

\begin{proof}
Consider a state with error weight $w$. To analyze the protocol's performance, we need to understand both how measurement outcomes and fidelity contributions arise from the different possible arrangements of $\Phi$ and $\Phi^\perp$ states.

To calculate $p_\acc$, we sum over all ways the $w$ errors could be distributed between measured and unmeasured positions. Let $j$ denote the number of $\Phi^\perp$ states (i.e., errors) in unmeasured positions, leaving $w-j$ in measured positions. We will pick up a factor $1/3$ for errors that are in the measured positions (since with probability $1/3$, the error will be an undetectable $Z$ error). Therefore, the contribution to $p_\acc$ from a particular error weight $w$ is
\begin{equation}
\frac{1}{\binom{n}{w}} \sum_{j=\max\{0,w-m\}}^{\min\{w,k\}} \binom{k}{j} \binom{m}{w-j} 3^{-(w-j)}.
\end{equation}

For the numerator $p_{\acc \land \Phi}$, we can follow the same approach, considering all ways the $w$ errors could be distributed between measured and unmeasured positions. However, we see that any terms with errors on the unmeasured positions have zero contribution --- that is, we must have $j=0$. Therefore, the contribution to $p_{\acc \land \Phi}$ from an error weight $w$ is just $3^{-w} \frac{\binom{m}{w}}{\binom{n}{w}}$, since $\binom{m}{w}$ counts the number of ways the weight $w$ error can be distributed amongst the $m$ check qubits.
\end{proof}

\end{document}